\documentclass[11pt]{article}
\usepackage[margin=1in]{geometry}

\usepackage{natbib}
\usepackage[hyphens]{url}
\usepackage{hyperref}
\hypersetup{breaklinks=true}

\usepackage{amsmath, amssymb, amsthm}

\usepackage{thm-restate}

\makeatletter
 \def\thm@space@setup{%
   \thm@preskip=\parskip \thm@postskip=0pt
 }
 \makeatother

\newtheorem{theorem}{Theorem}
\newtheorem{corollary}{Corollary}
\newtheorem{lemma}{Lemma}
\newtheorem{example}{Example}
\newtheorem{remark}{Remark}

\theoremstyle{definition}
\newtheorem{assumption}{Assumption}

\usepackage[table]{xcolor}
\usepackage{parskip}
\usepackage{import}
\usepackage{enumerate}

\usepackage{booktabs}
\usepackage{multirow}
\usepackage[suppress]{color-edits}
\addauthor{me}{blue}
\addauthor{sk}{red}
\addauthor{cy}{magenta}
\addauthor{ju}{orange}

\newcommand{\bc}{\mathbf{c}}

\newcommand{\bz}{\mathbf{z}}
\newcommand{\bw}{\mathbf{w}}
\newcommand{\bx}{\mathbf{x}}

\newcommand{\tbz}{\mathbf{\tilde{z}}}
\newcommand{\tbw}{\mathbf{\tilde{w}}}

\newcommand{\cC}{\mathcal{C}}

\newcommand{\cN}{\mathcal{N}}
\newcommand{\cS}{\mathcal{S}}
\newcommand{\cT}{\mathcal{T}}
\newcommand{\cU}{\mathcal{U}}
\newcommand{\cV}{\mathcal{V}}
\newcommand{\cW}{\mathcal{W}}
\newcommand{\cX}{\mathcal{X}}

\newcommand{\Ind}{\mathbb{I}}

\DeclareMathOperator*{\E}{\mathbb{E}}
\DeclareMathOperator{\TTE}{TTE}
\DeclareMathOperator{\var}{var}
\DeclareMathOperator{\cov}{cov}
\DeclareMathOperator{\Comp}{Comp}
\DeclareMathOperator{\Bern}{Bern}

\DeclareMathOperator{\GCR}{GCR}

\usepackage{tikz}
\usetikzlibrary{math, automata, arrows}
\usepackage{pgfplots}
\pgfplotsset{compat=1.17}

\usepackage{algorithm,algorithmic}

\newcommand{\dmax}{d_{\max}}

\renewcommand{\hat}{\widehat}

\title{Low-order outcomes and clustered designs: combining design and analysis for causal inference under network interference}

\author{Matthew Eichhorn$^*$$^\dagger$\and Samir Khan$^*$$^\ddagger$ \and
     Johan Ugander$^\S$\and Christina Lee Yu$^\dagger$
}

\date{\today}

\begin{document}

\maketitle
\def\thefootnote{*}\footnotetext{Equal contribution. Please direct correspondence to \texttt{meichhorn@cornell.edu}.}
\def\thefootnote{$\dagger$}\footnotetext{Cornell University, Departments of Computer Science and Operations Research and Information Engineering }
\def\thefootnote{$\ddagger$}\footnotetext{Meta}
\def\thefootnote{$\S$}\footnotetext{Yale University, Departments of Statistics and Data Science}

\begin{abstract}
    \noindent 

Variance reduction for causal inference in the presence of network interference is often achieved through either outcome modeling, typically analyzed under unit-randomized Bernoulli designs, or clustered experimental designs, typically analyzed without strong parametric assumptions. 
In this work, we study the intersection of these two approaches and make the following threefold contributions. 
First, we present an estimator of the total treatment effect (or global average treatment effect) in low-order outcome models when the data are collected under general experimental designs, generalizing previous results for Bernoulli designs. We refer to this estimator as the pseudoinverse estimator and give bounds on its bias and variance in terms of properties of the experimental design. Second, we evaluate these bounds for the case of Bernoulli graph cluster randomized (GCR) designs. 
Its variance scales like the smaller of the variance obtained by the estimator derived under a low-order assumption, and the variance obtained from cluster randomization, showing that combining these variance reduction strategies is preferable to using either individually.
When the order of the potential outcomes model is correctly specified, our estimator is always unbiased, and under a misspecified model, we upper bound the bias by the closeness of the ground truth model to a low-order model.
Third, we give empirical evidence that our variance bounds can be used to select a good clustering that minimizes the worst-case variance under a cluster randomized design from a set of candidate clusterings. Across a range of graphs and clustering algorithms, our method consistently selects clusterings that perform well on a range of response models, suggesting the practical use of our bounds.

\end{abstract}

\section{Introduction} \label{sec:intro}

Standard methods for estimating causal effects from randomized experiments typically rely on the stable unit treatment value assumption (SUTVA), which ensures that a unit's outcome depends only on its treatment, and not the treatments of other units. A recent line of work has revisited this assumption and developed methods for estimating treatment effects while allowing for interference between units. One common approach in this line of work is to model the interference through an exposure mapping \citep{aronow2017estimating}, and then construct experimental designs that mitigate the impact of the interference under the assumed exposure mapping \citep{ugander2013, eckles2017design, holtz2020reducing, karrer2021network, li2022interference}. Another common approach is to specify a parametric outcome model that captures the interference, and then estimate causal effects as functions of the model parameters \citep{toulis2013estimation, cai2015social, basse2018model, chin2019regression, cortez2023exploiting}. Both approaches reduce variance by exploiting either structure in the interference network or the outcomes model. In this paper, we consider the joint benefits of using an estimator derived from a low-order interactions model combined with a cluster randomized design that exploits network structure. We show that jointly \mereplace{optimizing over}{considering} the choice of the estimator \meedit{(i.e., the posited degree of the potential outcomes model)} and the design \meedit{(corresponding to a clustering of the network)} can lead to performance gains that combine the benefits of both.

We assume the interference has a network structure, in which an underlying graph dictates the extent of the spillovers in treatment. We take as a starting point the work of \citet{YuCortezEichhorn22} and \citet{cortez2023exploiting}, which introduced a low-order interaction model for network interference, requiring that if a unit's outcome depends on the treatment status of neighbors, it does so in a way that is linear in the treatment status of only small subsets of neighbors. \citet{cortez2023exploiting} study the problem of estimating the parameters of this model using data collected under Bernoulli designs, in which each unit $i$ is assigned to treatment independently with probability $p_i$. We go beyond their setting and study the problem of estimating the parameters of a low-order interaction model using data collected under general designs, with specialized results for graph cluster randomized (GCR) designs \citep{ugander2013}. Graph cluster randomized designs partition the units into clusters based on the underlying interference graph, and then assign treatments at the cluster level rather than the individual level. This increases the probability of a unit and its neighbors being treated together, which reduces the variance of treatment effect estimators under a neighborhood exposure assumption.

Our setting raises several questions, both narrow and broad. Under parametric outcome model assumptions, specifically a low-order interaction model, can clustered experimental designs still reduce estimator variance under interference? How do we estimate the model parameters from cluster-randomized data? What are the bias and variance properties of the resulting estimators? How should we select the clustering for a cluster-randomized experiment? More generally, is there value in using both a complex outcome model and a complex experimental design, or is using either one alone sufficient to handle interference effectively? In this paper, we answer all of these questions by developing a general theory for estimation in the low-order interaction model using data collected from an arbitrary experimental design, and then instantiating this theory for the case of graph cluster randomized (GCR) designs. The main text focuses on Bernoulli randomization (over individuals or clusters), while extensions to complete randomization are included in the appendix.

\subsection{Summary Of Results}

As mentioned above, we present both a general theory on estimation for low-order outcome models under arbitrary experimental designs and specific results for GCR designs. The $\beta$-order interactions assumption posits that the potential outcome $Y_i(\bz)$ of unit $i$ with respect to treatment vector $\bz$ can be written as the inner product of an unknown coefficient vector $\bc_i$, and the vector denoting whether each neighborhood subset of size at most $\beta$ is treated or not, represented by $\tbz_i^\beta$.

\paragraph{Pseudoinverse estimators.} Assuming a $\beta$-order outcome model, we extend the SNIPE estimator introduced in \citet{cortez2023exploiting} from Bernoulli unit randomized designs to arbitrary experimental designs. We refer to the general estimator as the $\beta$-order pseudoinverse estimator due to its connection to a similar estimator proposed in \cite{swaminathan2017off} for the seemingly different setting of semi-parametric estimation. 
While the estimator $\widehat{\TTE}_{\beta}$ is derived assuming the class of $\beta$-order interaction models, we analyze its performance under general potential outcomes models, providing an exact characterization for the bias of the estimator in Theorem \ref{thm:bias_bound_general}, and an upper bound for the variance of this estimator in Theorem \ref{thm:tte_var_bound} as a function of the general experimental design. This bound can be used to evaluate any proposed experimental design with only sampling access to the design, which can be used to select designs that have lower variance in the worst case. 

\paragraph{Bernoulli GCR designs.} We provide simplified and interpretable bounds on the bias and variance of the pseudoinverse estimators for Bernoulli randomized GCR designs, where Theorem~\ref{thm:onesummary} below summarizes the results of Corollaries~\ref{cor:gcr_unbiased} and \ref{cor:gcr_bias_bound} from Section~\ref{sec:bias} and Corollary~\ref{thm:pi_gcr} from Section~\ref{sec:variance}.
\begin{theorem}[Re-statement of Corollaries \ref{cor:gcr_unbiased}, \ref{cor:gcr_bias_bound} and \ref{thm:pi_gcr}] \label{thm:onesummary}
\[|\E[\widehat{\TTE}_{\beta}] - \TTE| \leq \|\bc_i^{ >\beta}\|_1, \]
where $\bc_i^{>\beta}$ contains coefficients that correspond to causal effects of neighborhood subsets of size larger than $\beta$. The variance is upper-bounded by
\[
    \var(\widehat{\TTE}_{\beta}) \leq
        \frac{2 B^2 C N \dmax}{n} \cdot \min\left(\frac{1}{p^C}, \frac{C^\beta}{p^\beta}\right),
\]
where $B$ is an absolute bound on the potential outcomes, $C$ is the maximum number of clusters an individual is influenced by, $N$ is the maximum cluster size, $\dmax$ is the maximum neighborhood size of an individual, $n$ is the number of individuals, and $p$ is the marginal treatment probability of a cluster.       
\end{theorem}

If the potential outcomes indeed satisfy $\beta$-order interactions, then the $\beta$-order pseudoinverse estimator is unbiased. The variance scales with the minimum of $p^{-C}$ and $(C/p)^{\beta}$, and inversely in $n/N$, the number of clusters. \meedit{Thus, selecting a clustering to minimize variance amounts to partitioning the graph into sufficiently many clusters (to reduce $N$) while limiting the effects that span across cluster boundaries (to reduce $C$).} If we focus on the term exhibiting exponential dependence, the variance of the Horvitz-Thompson estimator under a graph cluster randomized design scales as $p^{-C}$ \citep{ugander2013}, and the variance of the $\beta$-order pseudoinverse estimator under a Bernoulli unit-randomized design scales as $(C/p)^\beta$ \citep{cortez2023exploiting}. In our fine-grained variance bound, we show that the selection of the best between these two scalings occurs at the granularity of unit-level contributions to the overall variance. In particular, for units in the interior of a cluster such that the unit is only adjacent to a single cluster, i.e., $C=1$, they enjoy the benefit of contributing very little to the variance due to the GCR design; for units at the boundaries of clusters that could be adjacent to many more than $\beta$ clusters, their contribution to the variance is bounded by $p^{-\beta}$ due to using the $\beta$-order pseudoinverse estimator. This results in a performance that can be better than that achieved by using the vanilla Horvitz-Thompson estimator with GCR designs or using the $\beta$-order pseudoinverse estimator with simple Bernoulli unit randomized designs. Our result exhibits the mutually beneficial gains of \mereplace{jointly optimizing over the selection of the design and estimator, which has not been previously studied in the literature.}{imposing a low-order outcomes assumption and selecting a clustered randomized design informed by this assumption.}

\paragraph{Empirical validation for optimizing design of clustered experiments.} Finally, we consider the practical problem of choosing the experimental design that optimizes performance for the $\beta$-order pseudoinverse estimator. We propose selecting a clustering for a GCR experiment by minimizing a tractable upper bound on the mean-squared-error formed by combining our bias and variance bounds. We show that this approach exhibits favorable behavior in synthetic simulations across a range of response models, graphs, and clustering algorithms.

\subsection{Related Work}
\label{subsec:related}

The problem of causal inference under fully general interference is known to be intractable \citep{basse2018limitations}, so the literature on this problem contains many different assumptions that aim to make the problem more tractable. In this section, we survey several strands of work that either make similar assumptions to ours or study similar problems.

\paragraph{Neighborhood exposure mappings.}

Exposure mappings, introduced by \citet{aronow2017estimating} and closely related to work on ``effective treatments'' \citep{manski2013identification}, specify exactly which treatment assignments impact the outcome of a unit. We follow much of the existing network interference literature in making a neighborhood exposure mapping assumption \citep{ugander2013, eckles2017design, sussman2017elements, chin2019regression, bargagli2020heterogeneous}, in which the outcome of an individual depends only on the treatment assignments within their local neighborhood in an interference graph. More specifically, we make what is known as a neighborhood treatment response assumption, in which an individual with all their neighbors treated responds the same as they would to global treatment (see Assumption~\ref{ass:neighborhood_interference} for an exact statement). This approach differs from other works that make a fractional neighborhood treatment response assumption, in which an individual responds as they would to global treatment if a certain fraction of their neighbors are treated, such as \citet{eckles2017design} or \citet{holtz2020limiting}.

\paragraph{Graph structure assumptions.}

Even under a neighborhood exposure assumption, the set of possible outcomes for each individual is exponential in the size of the neighborhood. Thus, na\"{i}vely running a Bernoulli design and using the Horvitz--Thompson estimator \citep{horvitz1952generalization} will lead to a variance that scales exponentially with the maximum degree in the interference graph. As such, it is standard in the network interference setting to assume that this maximum degree is bounded \citep{ugander2013, ugander2023randomized, cortez2023exploiting}, an assumption we make in most of our bounds as well. 

Going beyond Bernoulli designs, one approach to variance reduction is to use a graph cluster randomized (GCR) design \citep{ugander2013}, which lowers the variance of the Horvitz--Thompson estimator under appropriate assumptions on the clustering and graph. For example, the results of \citet{ugander2013} assume that the graph satisfies $\kappa$-restricted growth, meaning that the neighborhoods of a node increase in size by at most a multiplicative factor of $\kappa$ as their size increases. Thus a node has at most $\dmax$ nodes in its neighborhood, at most $\kappa\dmax$ nodes in its 2-neighborhood, and so on. When analyzing the variance of the pseudoinverse estimator under GCR designs in Section~\ref{sec:variance}, we consider the dependence of our variance bounds on $\kappa$ under such a $\kappa$-restricted growth assumption.
We remark that the $\kappa$-restricted growth assumption is closely related to certain assumptions made by \citet{leung2022causal}. In particular, $\kappa$-restricted growth implies the exponential growth rate condition given in (18) of \citet{leung2022causal} with $C=\dmax$ and $\beta=\log\kappa$, and \citet{leung2022causal} show that this condition (together with appropriate corresponding conditions on the interference) is sufficient to obtain a central limit theorem for their proposed estimator. The common idea between \citet{ugander2013}, \citet{leung2022causal}, and the present work is to use growth rate conditions, when true, to control the size of neighborhoods in the graph, and thus bound the amount of interference.

Lastly, in a different kind of assumption on the graph structure, \citet{li2022random} assume the graph is sampled from a graphon, and derive asymptotics under this assumption. Such assumptions are not necessary in our work since we do not work in a large sample limit, but deriving central limit theorems for our estimators under such assumptions is a promising area for future work.

\paragraph{Outcome model assumptions.}

Rather than positing assumptions on the graph structure and constructing experimental designs to exploit them, another line of work instead makes assumptions on the structure of the potential outcomes that go beyond the neighborhood exposure assumption. Much of this work assumes that the potential outcomes are linear in an individual's own treatment and some simple statistics of the treatment assignments of their neighbors, such as the number or proportion of treated neighbors \citep{toulis2013estimation,gui2015network,parker2017optimal, holtz2020reducing}, or linear in engineered features of the treatment assignments \citep{chin2019regression,han2023model}. Under such assumptions, the number of model parameters is less than the number of individuals, so standard regression-based techniques for parameter estimation are applicable. 

In contrast, the recent work by \cite{yu22, YuCortezEichhorn22, cortez2023exploiting} that we build on also assumes a parametric model for the potential outcomes, but allows the number of model parameters to potentially exceed the number of observations, making regression-based estimators infeasible. That work instead assumes that the order of interactions present between units is bounded, enabling parameter estimation and allowing for variance bounds that scale polynomially in the graph degree (rather than exponentially). This setting, in which the potential outcome functions are confined to a function space without rendering the problem fully parametric, is also the setting of \citet{harshaw2022design}, who propose the use of so-called Riesz estimators. Indeed, \citet{cortez2023exploiting} show that their estimator is an example of a Riesz estimator under a Bernoulli design. When considering more complex experimental designs, positivity---which is a necessary condition for the existence of a Riesz estimator---may not be satisfied. The pseudoinverse estimator we study can be derived for any arbitrary experimental design, and is similar to the estimator proposed in \cite{swaminathan2017off} in the seemingly different context of semi-parametric estimation. The analysis, however, is quite different as they assume independent samples, whereas our measured outcomes are correlated via the interference network.

\paragraph{Design of clustered experiments.} 

In Section~\ref{subsec:choose_cluster}, we show how our bounds can inform the choice of clustering for an experiment. A related problem of designing the clustering in a GCR experiment has also been considered by \citet{leung2022rate}, who calculate the optimal number of square clusters in a spatial experiment as a function of the rate of decay of the interference. This spatial decay model is also studied by \citet{leung2023design}, who give results on the number of clusters that should be used for a particular downstream estimator that drops certain units on the boundaries of the clusters. 

For a more general take on the problem of optimizing over the space of ``causal clusterings,'' \citet{viviano2023causal} develop an algorithm for constructing a clustering to minimize an upper bound on the mean-squared-error in a GCR experiment, though without any consideration to low-order outcome models. Another key difference between our work and theirs is the choice of estimator: \citet{viviano2023causal} study a modification of the difference-in-means estimator that is a sum of independent terms, and thus more straightforward to analyze. In contrast, because it assumes and exploits a low-order outcome model, our estimator is a sum of dependent terms, and much of our effort focuses on controlling this dependence. Comparing our approach to their approach (or the more recent work of \citet{zhao2024simple} on optimized causal clustering for Horvitz--Thompson estimators), but under a low-order outcome model, is an interesting future direction that would further explore the relationship between designs and estimators for causal inference under interference.

Finally, in a different vein, \citet{candogan2023correlated} consider correlating the treatment assignments between clusters to achieve nearly equal numbers of treated and control units. We do not consider such designs in this work, but our methods and bounds extend readily and could be used to obtain new results for such correlated designs.

\section{Set-up and Notation} \label{sec:model}

Here, we introduce the class of outcome models, the target estimand, and the experimental designs of interest to us. Throughout, we consider a population of $n$ individuals, denoted $[n]$, each of whom is assigned a binary treatment $z_i \in \{0,1\}$. We interpret $z_i = 0$ as an assignment of $i$ to the control group and $z_i = 1$ as an assignment of $i$ to the treatment group. We collect these treatments into the vector $\bz = (z_1, \hdots, z_n)$. We measure outcomes $Y_i$, which may depend on the treatment assignments of other individuals due to interference. Thus each $Y_i$ can be naturally encoded as a function $Y_i \colon \{0,1\}^n \to \mathbb{R}$. We assume a design-based framework in which these functions $Y_i$ are deterministic so that the only randomness in the problem is in the distribution of $\bz$. As such, all expectations we take are with respect to the distribution $\bz$, which is specified where necessary.


\subsection{Low-order Outcome Models}

Our outcome model of interest is the low-order model of \citet{cortez2023exploiting}. This model is characterized by two assumptions, which we introduce below. For the remainder of the paper, we will always assume that Assumption~\ref{ass:neighborhood_interference} and \ref{ass:low_deg} hold without explicit statement.

For the first assumption, we will suppose that the interference can be modeled by a directed graph $G$ in which the edge $(j,i)$ indicates that the treatment of individual $j$ has an effect on the outcome of individual $i$. For individual $i$, we let $\cN_i$ represent the in-neighborhood of $i$ in this graph (which includes individual $i$ itself, since the treatment of individual $i$ certainly affects its outcome).

\begin{assumption}[Neighborhood interference] 
	\label{ass:neighborhood_interference}
	If treatment vectors $\bz$ and $\bz'$ have $z_j=z_j'$ for all $j\in \cN_i$, then $Y_i(\bz)=Y_i(\bz')$.
\end{assumption}

Assumption~\ref{ass:neighborhood_interference} corresponds to a common choice of exposure mapping \citep{aronow2017estimating} that is often referred to as \textit{neighborhood interference}. Our bounds are often stated in terms of the degrees $d_i = |\cN_i|$ of these in-neighborhoods. In particular, we use $\dmax = \max_{i} d_i$ to denote the maximum degree, which we assume to be a fixed constant that does not depend on $n$. Using Assumption~\ref{ass:neighborhood_interference} along with the fact that each $z_i \in \{0,1\}$, we can (without loss of generality) represent $Y_i$ as a polynomial,
\begin{align} 
	Y_i(\bz) &= \sum_{T \subseteq \cN_i} a_{i,T} \; \Ind \big( z_j=1\text{ for all }j\in T, z_j=0\text{ for all }j\not\in T\big) \notag \\
	&= \sum_{T \subseteq \cN_i} a_{i,T} \prod_{j \in T} z_j \prod_{j' \in \cN_i \setminus T} (1 - z_{j'}).\label{eq:general_pom_a}
\end{align}
Here, $a_{i,T}$ is a parameter that captures the contribution to $i$'s outcome when the individuals in $T$ within its neighborhood are treated and the individuals in $\cN_i \setminus T$ are not treated. Following \citet{cortez2023exploiting}, we reparametrize this polynomial representation of $Y_i$ in terms of the additive treatment effect on individual $i$ from each subset of their neighborhood that is completely treated, regardless of the treatment assignments of its other neighbors. This corresponds to writing
\begin{equation} \label{eq:general_pom_c}
	Y_i(\bz) = \sum_{S \subseteq \cN_i} c_{i,S} \; \Ind \big( z_j=1\text{ for all }j \in S \big) = \sum_{S \subseteq \cN_i} c_{i,S} \prod_{j \in S} z_j,
\end{equation}
and is a transformation of \eqref{eq:general_pom_a} to the monomial polynomial basis, with the exact mapping between coefficients given by
\[
    c_{i,S} = \sum_{T \subseteq S} (-1)^{|S \setminus T|} \cdot a_{i,T}.  
\]

Using the representation of $Y_i$ from \eqref{eq:general_pom_c}, we can now introduce the second assumption, which is the central assumption of the low-order model.

\begin{assumption}[$\beta$-order interactions]
	\label{ass:low_deg}
	We have $c_{i,S}=0$ whenever $|S|>\beta$.
\end{assumption}

Assumption~\ref{ass:low_deg} posits that only sufficiently small subsets of an individual's neighborhood can have an additive effect on their outcome. The parameter $\beta$ represents the maximum possible size of such a subset. Under this assumption, we hypothesize a simplification of the outcome model \eqref{eq:general_pom_c} to
\begin{equation} \label{eq:pom_beta}
    Y_i(\bz) = \sum_{\substack{S \subseteq \cN_i \\ |S| \leq \beta}} c_{i,S} \prod_{j \in S} z_j =: \sum_{S \in \cS_i^{\beta}} c_{i,S} \prod_{j \in S} z_j.
\end{equation}
Here we write $\cS_i^\beta$ to represent the set of all subsets of $\cN_i$ with size at most $\beta$. Note that the empty set always belongs to $\cS_i^\beta$ and its corresponding coefficient $c_{i,\varnothing}$ is the observed outcome under global control, i.e.,~$c_{i, \varnothing}=Y_i(\mathbf{0})$. The parameter $\beta$ controls the richness of the assumed model class: $\beta = 1$ corresponds to a heterogeneous linear potential outcomes model \citep{yu22}, and $\beta = \dmax$ corresponds to fully general neighborhood interference. Throughout, we view $\beta$ as a parameter chosen by the practitioner. Many of our results (e.g., Theorem~\ref{thm:bias_bound_general} in Section~\ref{sec:bias}) consider the possibility of \textit{misspecification}, where the chosen $\beta$ differs from the order of the ground-truth potential outcomes model. In this case, we'll use $\beta^*$ to denote this ground-truth model order.

\begin{remark}
    While Assumption~\ref{ass:low_deg} was introduced by \citet{YuCortezEichhorn22}, many other models of graph interference in the literature are implicitly low-degree. The model of \citet{holtz2020reducing} (which is a simplification of the model in \citet{eckles2017design}) is, up to a noise term, 
\begin{equation}
    Y_i(\bz)=\alpha_i+\delta z_i+\gamma \rho_i,
\end{equation}
where $\alpha_i$ is an intercept, $\delta$ (relabeled from $\beta$ in their model to avoid a clash of notation) is the direct effect of treatment, $\gamma$ is the indirect effect of treatment, and $\rho_i$ is the fraction of treated neighbors. This model can be written as a $\beta=1$ low-order model with
\begin{equation}
    c_{i,\emptyset} = \alpha_i,
    \qquad c_{i, \{i\}}=\delta,
    \qquad c_{i,\{j\}}=\frac{\gamma}{d_i}\text{ for } j\neq i.
\end{equation}

\citet{gui2015network} analyze a similar ``linear-in-number-of-treated-neighbors'' interference model that is also equivalent to a low-order model with $\beta=1$; \citet{sussman2017elements} analyze a model of the form $Y_i(\bz)=\alpha_i+\delta_i z_i+\Gamma(\text{\# of treated neighbors})$, where $\Gamma:\{1,\cdots, d_i\}\to \mathbb{R}$ is an arbitrary function, that is again equivalent to a low-order model with $\beta=1$. 

\citet{deng2024unbiased} consider a setting of interference, in which outcomes are modeled as a sum of effects on an individual's incoming edges in the interference network. The effects on each edge depend on the treatment assignments of \textit{both} endpoints, giving rise to the model
\begin{equation}
    Y_i(\bz) = \sum_{j \in \cN_i} \big( \alpha_{j,i} + \delta_{j,i} \cdot z_j + \gamma_{j,i} \cdot z_i + \zeta_{j,i} \cdot z_i z_j \big).
\end{equation}
This can be rewritten as a $\beta=2$ low-order model with
\begin{equation}
    c_{i,\emptyset} = \sum_{j \in \cN_i} \alpha_{j,i},
    \qquad c_{i, \{i\}} = \sum_{j \in \cN_i} \gamma_{j,i},
    \qquad c_{i,\{j\}} = \delta_{j,i}, 
    \qquad c_{i,\{i,j\}} = \zeta_{j,i}.
\end{equation}

Finally, \citet{shankar2024b} and \citet{yuan2021causal} consider potential outcomes that depend on ``motifs'' of an individual's neighborhood, which are small induced subgraphs with labeled treatment assignments. The model of \citet{shankar2024b} models outcomes with additive effects for a collection of motifs of size at most $\beta$, while \citet{yuan2021causal} posits a model where the outcomes can be expressed as a weighted linear combination of the count of each from a sufficiently small set of motifs. Both models are subsumed by a $\beta$-order interactions model with $\beta$ equal to the largest motif size.
\end{remark}

\noindent Some of our bounds further require a standard boundedness assumption\mereplace{, which is analogous to an assumption on the boundedness of $Y_i(\mathbf{0})$ and $Y_i(\mathbf{1})$ in a standard potential outcomes model without interference.}{.}

\begin{assumption}[Bounded outcomes]
	\label{ass:bounded}
	We have $|Y_i(\bz)| \leq B$ for all $i \in [n]$ and $\bz \in \{0,1\}^n$.
\end{assumption}

\noindent This assumption says that each individual's outcome $Y_i(\bz)$ is uniformly bounded under \textit{any} treatment configuration. This is analogous to an assumption on the boundedness of $Y_i(\mathbf{0})$ and $Y_i(\mathbf{1})$ in a standard potential outcomes model without interference.

Without knowing the ground truth $\beta^*$, in order to employ our estimator with some low-order model, we must somehow choose $\beta$. Given a choice of $\beta$, we let $\tbz_i^\beta$ denote the vector indexed by $S \in \cS_i^\beta$ (in some canonical order with $\varnothing$ in the first position) with entries $(\tbz_i^\beta)_S = \prod_{j \in S} z_j$. We similarly collect the model coefficients $c_{i, S}$ into a vector $\bc_i^\beta$, where $(\bc_i^\beta)_S = c_{i,S}$. In this notation, we have that $Y_i(\bz) = \big\langle \tbz_i^\beta, \bc_i^{\beta} \big\rangle$. For notational simplicity, we sometimes suppress the dependence on $\beta$ in $\tbz_i$, $\bc_i$, etc. Throughout, we use both $\langle v, w\rangle$ and $v^\intercal w$ to denote the Euclidean inner product of $v$ and $w$.

\subsection{Total Treatment Effect}
\label{subsec:estimand}

Our causal estimand of interest is the \textit{total treatment effect} (otherwise referred to as the \textit{global average treatment effect} or GATE),
\[
    \TTE = \frac{1}{n} \sum_{i=1}^{n} \Big( Y_i(\mathbf{1}) - Y_i(\mathbf{0}) \Big),
\]
the average difference between an individual's outcome under global treatment and their outcome under global control. Under Assumptions~\ref{ass:neighborhood_interference} and \ref{ass:low_deg}, we can use the representation of $Y_i$ given in \eqref{eq:pom_beta} to express the total treatment effect in terms of the coefficients $c_{i,S}$:
\begin{equation} \label{eq:tte_c}
    \TTE = \frac{1}{n} \sum_{i=1}^{n} \sum_{\substack{S \in \cS_i^{\beta} \\ S \ne \varnothing}} c_{i,S} = \frac{1}{n} \sum_{i=1}^{n} \langle \theta_i^\beta, \bc_i^{\beta} \rangle,
\end{equation}
where we define $\theta_i^\beta = [ \, 0 \; 1 \; \cdots \; 1 \,]^\intercal$ to be the length-$|\cS_i^\beta|$ column vector with $0$ in its first entry (corresponding to $S = \varnothing$) and $1$ in all other entries.

We remark that although we focus on the $\TTE$ in this work, we expect that our results extend to other causal estimands, such as the direct or indirect effects of \citet{hu2022average}, that can be written as linear combinations of the $\bc_i^{\beta}$. That is to say, our results do not rely at all on the particular form of $\theta_i^\beta$ \textemdash if we were interested in $\langle v, \bc_i^{\beta}\rangle$ for some other vector $v$, our analysis would go through identically with $\theta_i^\beta$ replaced by $v$, and we would obtain similar estimators and bounds.

\subsection{Experimental Designs}
\label{subsec:designs}

While we present bias and variance bounds that can be computed for any general randomized design, to obtain interpretable insights, we focus on the class of {\em Bernoulli graph cluster randomized} designs, which we denote $\textrm{GCR}(\cC,p)$. Under such a design, we first fix a clustering $\cC$ of the individuals, which partitions the units into a set of $m$ clusters $\cC=\{C_1,\cdots, C_m\}$ with each $C_j\subseteq [n]$. Then, each cluster $C \in \cC$ is independently assigned to treatment with probability $p$, where $w_C \sim \textrm{Bern}(p)$ denotes the treatment decision of cluster $C$. For an individual $j$, we set $z_j = w_{\cC(j)}$ where $\cC(j)$ is the cluster containing individual $j$. That is, a unit inherits the treatment of the cluster that contains it. When $\cC$ consists of the collection of singleton sets, then $\textrm{GCR}(\cC,p)$ is equivalent to a Bernoulli unit randomized design. In Appendix~\ref{sec:crd} we extend some of our results to \textit{completely randomized GCR designs}, where a uniform random subset $S$ of $k$ clusters is selected, which introduces some negative correlation between the cluster treatment assignments.
\section{The Pseudoinverse Estimator} \label{sec:pinv}

In this section, we introduce the \textit{pseudoinverse estimator} for the TTE, which depends on a practitioner-selected order parameter $\beta$ as well as the $\beta^{\text{th}}$-moments of the design. This generalizes the estimator of \citet{cortez2023exploiting} beyond unit-randomized designs to any choice of randomized design. 

\subsection{Deriving the Pseudoinverse Estimator} \label{subsec:derive}

Recall from Section~\ref{subsec:estimand} that, under an assumed $\beta$-order outcome model (Assumption~\ref{ass:low_deg}), the total treatment effect is a linear combination of the effect coefficients $c_{i,S}$. Thus, to estimate the $\TTE$, it is natural to consider an approach employing good estimators for each $\bc_i^\beta$, which we can then extend by linearity to a good estimator for the $\TTE$. We restrict our attention to estimating $\bc_i^\beta$ for an arbitrary $i \in [n]$. Note that $\bc_i^\beta$ only affects the outcome $Y_i(\bz)$ of unit $i$ and not the outcomes of any other units, thus we can focus on the relationship $Y_i(\bz) = \big\langle \tbz_i^\beta, \bc_i^\beta \big\rangle$. By pre-multiplying both sizes by $\tbz_i^\beta$ and taking an expectation with respect to the randomized design, it follows that
\begin{equation} \label{eq:normal}
    \E \big[ Y_i(\bz) \cdot \tbz_i^\beta \big] = \E \big[ \tbz_i^\beta {(\tbz_i^\beta)}^\intercal \big] \bc_i^\beta. 
\end{equation}

The matrix $\E \big[ \tbz_i^\beta {(\tbz_i^\beta)}^\intercal \big]$ has entries that are functions only of the experimental design. As such, we refer to it as the \textit{design matrix}.  Note that the design matrix is not data-dependent and can be computed by the practitioner given the exact distribution of $\bz$, or approximated to arbitrary precision given sampling access to the distribution of $\bz$ (see Section~\ref{subsec:monte_carlo} for details). From here, we consider two possibilities.

\textbf{Case 1:} First, suppose that the design matrix is invertible, as was shown for unit randomized designs by \citet{cortez2023exploiting}. Then we can rearrange \eqref{eq:normal} to solve for $\bc_i^\beta$:
\[
    \bc_i^\beta = \E \big[ \tbz_i^\beta {(\tbz_i^\beta)}^\intercal \big]^{-1} \E \big[  Y_i(\bz) \cdot \tbz_i^\beta \big].
\]
By approximating this latter expectation by its experimental realization, we obtain the estimator,
\[
    \widehat{\bc}_i^\beta := Y_i(\bz) \, \E \big[ \tbz_i^\beta {(\tbz_i^\beta)}^\intercal \big]^{-1} \tbz_i^\beta,
\]
which is unbiased by linearity. The core ideas in this case --- deriving a system of unbiasedness equations whose solution identifies an estimator --- are akin to those of \citet{harshaw2022design}, and we expect that our estimator coincides with the Riesz estimator whenever the design matrix is invertible, just as it does for Bernoulli designs \citep{cortez2023exploiting}.

\textbf{Case 2:} Alternatively, suppose that the design matrix is not invertible. This can be the case in cluster-randomized designs, as the rank of the design matrix is limited to the number of subsets including units from at most $\beta$ clusters, which can be significantly smaller than $|\cS_i^{\beta}|$. In this case, the system of equations given by \eqref{eq:normal} is underdetermined, and we can use the Moore-Penrose pseudoinverse to obtain the minimum $\ell^2$-norm solution,
\[
    \tilde{\bc}_i^\beta = \E \big[ \tbz_i^\beta {(\tbz_i^\beta)}^\intercal \big]^{\dagger} \, \E \big[ Y_i(\bz) \tbz_i^\beta \big].
\]

Using the same reasoning as in Case 1, we obtain an estimator of $\bc_i^\beta$ by replacing the second expectation by the single observation:
\[
    \widehat{\bc}_i^\beta := Y_i(\bz) \, \E \big[ \tbz_i^\beta {(\tbz_i^\beta)}^\intercal \big]^{\dagger} \, \tbz_i^\beta.
    \label{eq:c_hat}
\]
This estimator subsumes the estimator from Case 1, so we utilize this form throughout the rest of the paper. Now, substituting the estimated $\widehat{\bc}_i^\beta$ for each $i$ into \eqref{eq:tte_c}, we obtain an estimator for the total treatment effect:

\begin{equation} \label{eq:tte_hat_bstar}
    \widehat{\TTE}_{\beta} := \frac{1}{n} \sum_{i=1}^{n} \sum_{\substack{S \in \cS_i^{\beta} \\ S \ne \varnothing}} (\widehat{\bc}_i^\beta)_S 
    \quad = \frac{1}{n} \sum_{i=1}^{n} \big\langle \theta_i^\beta, \widehat{\bc}_i^\beta \big\rangle
    \quad = \frac{1}{n} \sum_{i=1}^{n} Y_i(\bz) \Big\langle \E \big[ \tbz_i^\beta {(\tbz_i^\beta)}^\intercal \big]^{\dagger} \theta_i^\beta, \tbz_i^\beta \Big\rangle,
\end{equation}

where $\theta_i^\beta$ is the vector with $(\theta_i^\beta)_\varnothing = 0$ and with all other entries $(\theta_i^\beta)_S = 1$. We refer to \eqref{eq:tte_hat_bstar} as the \textit{pseudoinverse estimator} for the $\TTE$. We borrow this nomenclature from \citet{swaminathan2017off}, who use a similar reasoning to develop an estimator in the setting of off-policy evaluation. In fact, under Bernoulli unit-randomized designs, our estimator coincides with theirs for a particular choice of parameters, as noted in Section 4.3 of \citep{cortez2023exploiting}.

We remark that the form of this estimator was identified by \citet{cortez2023exploiting} for the invertible case, but not evaluated or analyzed for any design other than unit-level Bernoulli randomization. In contrast, we handle both the invertible and non-invertible cases and derive new results for more complex designs in both cases.

Note that the construction of $\widehat{\TTE}_\beta$ depends on the value of $\beta$ selected by the practitioner, and our motivation implicitly assumes its correct specification. In Sections~\ref{sec:bias} and \ref{sec:variance}, we will contend with the possibility of misspecification, where the choice of $\beta$ differs from the ground-truth potential outcomes model order $\beta^*$. In this case, the observations $Y_i(\bz) = \langle \tbz_i^{\beta^*}, \bc_i^{\beta^*} \rangle$ depends on $\beta^*$, whereas the weights on these observations in $\widehat{\TTE}_\beta$ depend on the value $\beta$ used to construct the estimator. This results in the following ``mixed superscripts'':
\begin{equation} \label{eq:tte_hat}
    \widehat{\TTE}_{\beta} \quad= \frac{1}{n} \sum_{i=1}^{n} \langle \tbz_i^{\beta^*}, \bc_i^{\beta^*} \rangle \Big\langle \E \big[ \tbz_i^{\beta} (\tbz_i^{\beta})^\intercal \big]^{\dagger} \theta_i, \tbz_i^{\beta} \Big\rangle.
\end{equation}
For notational tractability, we will often suppress the dependence on $\beta$ in the estimator weights.

\subsection{Instantiating the Pseudoinverse Estimator for Bernoulli GCR designs} \label{subsec:pi_gcr}

Recall that a clustering $\cC$ partitions the individuals $[n]$ into $m$ disjoint clusters $\{ C_1, \hdots, C_m \}$. We can view $\cC$ as a map $[n] \to [m]$, where $\cC(j)$ is the cluster containing node $j$, and we extend this notation to define $\cC(S) := \big\{ \cC(j) \colon j \in S \big\}$. In particular, $\cC(\cN_i)$ denotes the set of clusters that include at least one neighbor of $i$. Throughout this section, we will use standard uppercase letters (e.g., $S, T$) to denote sets of \textit{individuals} and calligraphic letters (e.g., $\cU, \cV$) to denote sets of \textit{clusters}.

GCR designs assign a treatment decision $w_{C} \in \{0,1\}$ to cluster $C \in \cC$. Then, the treatment assignments of an individual $j \in [n]$ are inherited from their cluster, so that $z_j = w_{\cC(j)}$. To analyze the pseudoinverse estimator $\widehat{\TTE}_\beta$ under GCR designs, we must define cluster-analogs of our previous notation. We'll let $\cC_i^\beta$ denote the set of all subsets of \textit{clusters} in $\cC(\cN_i)$ with size at most $\beta$. Additionally, we'll let $\tbw_i^\beta$ (abbreviated $\tbw_i$) be the analog of $\tbz_i^\beta$ for the cluster treatment decisions. That is, $\tbw_i$ is a $\{0,1\}$-vector indexed by $\cC_i^\beta$ with $\big[\tbw_i\big]_{\cU} = \prod_{C \in \cU} w_C$. From the definition of a GCR design, each entry of $\tbz_i$ is the same as some entry of $\tbw_i$. Given $S \in \cS_i^\beta$
\begin{equation} \label{eq:ztilde_to_wtilde}
    \big(\tbz_i\big)_S = \prod_{j \in S} z_j = \prod_{j \in S} w_{\cC(j)} = \prod_{C \in \cC(S)} w_C = \big( \tbw_i \big)_{\cC(S)}.
\end{equation}
We can define a (tall and narrow) $|\cS_i^{\beta}| \times |\cC_i^{\beta}|$ transformation matrix $D_i^\beta$ (abbreviated $D_i$) with $(D_i)_{S, \cU} = \Ind(\cC(S) = \cU)$ for $S \in \cS_i^{\beta}$ and $\cU \in \cC_i^{\beta}$ that encodes cluster membership. From (\ref{eq:ztilde_to_wtilde}), we see that $D_i$ has a block diagonal structure,
\[
    D_i = \begin{bmatrix}
        \mathbf{1}_{M_i^\beta(\cU_1) \times 1} & & 0 \\
        & \ddots & \\
        0 & & \mathbf{1}_{M_i^\beta(\cU_\ell) \times 1}
    \end{bmatrix}
    \in \mathbb{R}^{|\cS_i^\beta| \times |\cC_i^\beta|}.
\]
where $\mathbf{1}_{a \times b}$ refers to the $a \times b$ submatrix of all ones. The height of each block, corresponding to some cluster subset $\cU \in \cC_i^\beta$, is denoted by $M_i^\beta(\cU)$ and is the number of subsets $S \in \cS_i^\beta$ with $\cC(S) = \cU$. One can calculate $M_i^\beta(\cU)$ for $\cU = \{ C_{j_1}, \hdots, C_{j_\ell} \}$ with the formula
\[
    M_i^\beta(\cU) = \sum_{\substack{ a_1\,,\,\hdots\,,\,a_\ell \\ a_1 \,+\, \hdots \,+\, a_\ell \;\leq\; \beta \\ a_1 \,,\, \hdots \,,\, a_\ell \;\geq\; 1}} 
    \hspace{4pt} \prod_{k=1}^{\ell} \binom{|C_{j_k} \cap \cN_i|}{a_j}.
\]

By construction, $\tbz_i = D_i \tbw_i$, resulting in the design matrix $\E[\tbz_i \tbz_i^\intercal]$ having a block structure due to the block diagonal matrix $D_i$. The below examples illustrate the block structure of the design matrix as it relates to the design matrix of a smaller instance of unit-randomized design.

\begin{example} \label{ex:gcr_design_matrix_1}
Consider the design matrix for individual $I$ with neighborhood $\cN_I = \{ I, J \}$ in an interference graph with $\beta = 2$ under a Bern($p$) \textit{unit-randomized} design.

\begin{center}
\begin{minipage}{0.3\textwidth}
    \begin{center}
    \begin{tikzpicture}[-latex,>=latex]
        \node[circle,inner sep=0pt,minimum size=6pt,fill] (i) at (1.5,0) {};
        \node at (1.5,-0.4) {$I$};
    
        \node[circle,inner sep=0pt,minimum size=6pt,fill] (j) at (0,0) {};
        \node at (0,-0.4) {$J$};
    
        \draw (i) edge[loop above] (i);
        \draw (j) edge (i);
    \end{tikzpicture}
    \end{center}
\end{minipage}
\begin{minipage}{0.5\textwidth}
    \begin{center}
    \begin{tikzpicture}
    \node at (-1.25,1.5) {$\displaystyle \E \big[ \tbz_I \tbz_I^\intercal \big] =$};
    \node at (0,2.25) {\scriptsize $\varnothing$};
    \node at (0,1.5) {\scriptsize $\{I\}$};
    \node at (0,0.75) {\scriptsize $\{J\}$};
    \node at (0,0) {\scriptsize $\{I,J\}$};

    \node at (1,3) {\scriptsize $\varnothing$};
    \node at (1.75,3) {\scriptsize $\{I\}$};
    \node at (2.5,3) {\scriptsize $\{J\}$};
    \node at (3.25,3) {\scriptsize $\{I,J\}$};

    \draw[thick] (0.75,-0.5) -- (0.5,-0.5) -- (0.5,2.75) -- (0.75,2.75);

    \draw[thick] (3.5,-0.5) -- (3.75,-0.5) -- (3.75,2.75) -- (3.5,2.75);

    \node at (1,2.25) {$1$};
    \node at (1,1.5) {$p$};
    \node at (1,0.75) {$p$};
    \node at (1,0) {$p^2$};

    \node at (1.75,2.25) {$p$};
    \node at (1.75,1.5) {$p$};
    \node at (1.75,0.75) {$p^2$};
    \node at (1.75,0) {$p^2$};

    \node at (2.5,2.25) {$p$};
    \node at (2.5,1.5) {$p^2$};
    \node at (2.5,0.75) {$p$};
    \node at (2.5,0) {$p^2$};

    \node at (3.25,2.25) {$p^2$};
    \node at (3.25,1.5) {$p^2$};
    \node at (3.25,0.75) {$p^2$};
    \node at (3.25,0) {$p^2$};
    \end{tikzpicture}
    \end{center}
\end{minipage}

\vspace{20pt}
\end{center}

Next let's consider a neighborhood $\cN_i = \{i,i',j\}$ in an interference graph with $\beta = 2$. Suppose that the graph has been clustered so that $\cC(i) = \cC(i') \ne \cC(j)$. Letting $\cC(i) = I$ and $\cC(j) = J$, we see that the connectivity between clusters mirrors the connections between $I$ and $J$ above.

    \vspace{10pt}
    \begin{center}
    \begin{tikzpicture}[-latex,>=latex]
        \node[circle,inner sep=0pt,minimum size=6pt,fill] (i) at (1,0) {};
        \node at (1.25,0) {$i$};
    
        \node[circle,inner sep=0pt,minimum size=6pt,fill] (ip) at (0,1.5) {};
        \node at (-0.25,1.5) {$i'$};
    
        \node[circle,inner sep=0pt,minimum size=6pt,fill] (j) at (2,1.5) {};
        \node at (2.25,1.5) {$j$};
    
        \draw (i) edge[loop below] (i);
        \draw (ip) edge (i);
        \draw (j) edge (i);
    
        \draw[rotate=123,dashed] (0.3,-0.8) ellipse (48pt and 20pt);
        \node at (-1,2.2) {$\cC(i)$};
    
        \draw[dashed] (2,1.5) circle (16pt);
        \node at (2.8,2.2) {$\cC(j)$};
    \end{tikzpicture}
    \end{center}

    The design matrix for a $\GCR(\cC,p)$ design for this instance is shown below. The dashed lines show the block structure, where the values of the matrix align with the unit-randomized instance above.
    
    \vspace{10pt}
    \begin{center}
    \begin{tikzpicture}
        \node at (-1.5,2.625) {$\E \big[ \tbz_i \tbz_i^\intercal \big] =$};
    
        \node at (0,4.5) {\scriptsize $\varnothing$};
        \node at (0,3.75) {\scriptsize $\{i\}$};
        \node at (0,3) {\scriptsize $\{i'\}$};
        \node at (0,2.25) {\scriptsize $\{i,i'\}$};
        \node at (0,1.5) {\scriptsize $\{j\}$};
        \node at (0,0.75) {\scriptsize $\{i,j\}$};
        \node at (0,0) {\scriptsize $\{i',j\}$};
    
        \node at (1,5.25) {\scriptsize $\varnothing$};
        \node at (2,5.25) {\scriptsize $\{i\}$};
        \node at (3,5.25) {\scriptsize $\{i'\}$};
        \node at (4,5.25) {\scriptsize $\{i,i'\}$};
        \node at (5,5.25) {\scriptsize $\{j\}$};
        \node at (6,5.25) {\scriptsize $\{i,j\}$};
        \node at (7,5.25) {\scriptsize $\{i',j\}$};
    
        \draw[thick] (0.75,-0.5) -- (0.5,-0.5) -- (0.5,5) -- (0.75,5);
    
        \draw[thick] (7.5,-0.5) -- (7.75,-0.5) -- (7.75,5) -- (7.5,5);
    
        \node at (1,4.5) {$1$};
        \node at (1,3.75) {$p$};
        \node at (1,3) {$p$};
        \node at (1,2.25) {$p$};
        \node at (1,1.5) {$p$};
        \node at (1,0.75) {$p^2$};
        \node at (1,0) {$p^2$};
    
        \node at (2,4.5) {$p$};
        \node at (2,3.75) {$p$};
        \node at (2,3) {$p$};
        \node at (2,2.25) {$p$};
        \node at (2,1.5) {$p^2$};
        \node at (2,0.75) {$p^2$};
        \node at (2,0) {$p^2$};
    
        \node at (3,4.5) {$p$};
        \node at (3,3.75) {$p$};
        \node at (3,3) {$p$};
        \node at (3,2.25) {$p$};
        \node at (3,1.5) {$p^2$};
        \node at (3,0.75) {$p^2$};
        \node at (3,0) {$p^2$};
    
        \node at (4,4.5) {$p$};
        \node at (4,3.75) {$p$};
        \node at (4,3) {$p$};
        \node at (4,2.25) {$p$};
        \node at (4,1.5) {$p^2$};
        \node at (4,0.75) {$p^2$};
        \node at (4,0) {$p^2$};
    
        \node at (5,4.5) {$p$};
        \node at (5,3.75) {$p^2$};
        \node at (5,3) {$p^2$};
        \node at (5,2.25) {$p^2$};
        \node at (5,1.5) {$p$};
        \node at (5,0.75) {$p^2$};
        \node at (5,0) {$p^2$};
    
        \node at (6,4.5) {$p^2$};
        \node at (6,3.75) {$p^2$};
        \node at (6,3) {$p^2$};
        \node at (6,2.25) {$p^2$};
        \node at (6,1.5) {$p^2$};
        \node at (6,0.75) {$p^2$};
        \node at (6,0) {$p^2$};
    
        \node at (7,4.5) {$p^2$};
        \node at (7,3.75) {$p^2$};
        \node at (7,3) {$p^2$};
        \node at (7,2.25) {$p^2$};
        \node at (7,1.5) {$p^2$};
        \node at (7,0.75) {$p^2$};
        \node at (7,0) {$p^2$};
    
        \draw[dashed] (0.5,1.125) -- (7.75,1.125);
        \draw[dashed] (0.5,1.875) -- (7.75,1.875);
        \draw[dashed] (0.5,4.125) -- (7.75,4.125);
        \draw[dashed] (1.5,-0.5) -- (1.5,5);
        \draw[dashed] (4.5,-0.5) -- (4.5,5);
        \draw[dashed] (5.5,-0.5) -- (5.5,5);

        \draw[decorate,decoration={brace,amplitude=5pt,mirror}] (8,4.14) -- (8,5) node[midway, xshift=36pt] {$M_i^\beta(\varnothing) = 1$};
        \draw[decorate,decoration={brace,amplitude=5pt,mirror}] (8,1.89) -- (8,4.11) node[midway, xshift=52pt] {$M_i^\beta\Big(\big\{\cC(i)\big\}\Big) = 3$};
        \draw[decorate,decoration={brace,amplitude=5pt,mirror}] (8,1.14) -- (8,1.86) node[midway, xshift=52pt] {$M_i^\beta\Big(\big\{\cC(j)\big\}\Big) = 1$};
        \draw[decorate,decoration={brace,amplitude=5pt,mirror}] (8,-0.5) -- (8,1.11) node[midway, xshift=68pt] {$M_i^\beta\Big(\big\{\cC(i),\cC(j)\big\}\Big) = 2$};
    \end{tikzpicture}
    \end{center}
\end{example}

\vspace{10pt}

Given the mapping $\tbz$ and $\tbw$, we can analytically invert the design matrix for Bernoulli GCR designs to show the explicit form of the estimator as it relates to the clustering.

\begin{lemma} \label{lem:gcr_pi_design}
    Fix an individual $i \in [n]$ and a clustering $\cC$ of $[n]$ with $|\cC| = m$. If the design of $\bz$ assigns treatment at the cluster level with cluster treatment indicators $\bw \in \{0,1\}^m$, for each $S,T \in \cS_i^\beta$, the pseudoinverse of the design matrix has entries
    \[
        \Big[ \E \big[ \tbz_i \tbz_i^\intercal \big]^\dagger \Big]_{S,T} = \frac{1}{M_i^\beta(\cC(S)) \cdot M_i^\beta(\cC(T))} \cdot \Big[ \E \big[ \tbw_i \tbw_i^\intercal \big]^{\dagger} \Big]_{\cC(S),\cC(T)}.
    \]
\end{lemma}

\proof{\it Proof.}
    Recall that $\cC_i^\beta$ refers to all subsets of clusters $\cU \subseteq \cC(\cN_i)$ with $|\cU| \leq \beta$. As $\tbz_i = D_i \tbw_i$, it follows that 
    $\E \big[ \tbz_i \tbz_i^\intercal \big] = D_i^\beta \E \big[ \tbw_i \tbw_i^\intercal \big] D_i^\intercal$. This means that the design matrix $\E \big[ \tbz_i \tbz_i^\intercal \big]$ has a block structure with blocks corresponding to pairs of subsets $\cU, \cV \in \cC_i^\beta$. The $(\cU,\cV)$ blocks have dimension $M_i^\beta(\cU) \times M_i^\beta(\cV)$ and entry $\Big[ \E \big[ \tbw_i \tbw_i^\intercal \big] \Big]_{\cS,\cT}$. 
    As the matrix $D_i$ has full (column) rank, this factorization $\E \big[ \tbz_i \tbz_i^\intercal \big] = D_i \cdot \E \big[ \tbw_i \tbw_i^\intercal \big] D_i^\intercal$ is a rank decomposition. Using the fact that $(AB)^{\dagger}=B^{\dagger}A^{\dagger}$ when $AB$ is a rank decomposition, it follows that
\begin{align} \label{eq:inverse_design_GCR}
        \E \big[ \tbz_i \tbz_i^\intercal \big]^{\dagger}
        = \big(D_i^\intercal\big)^\dagger \E \big[ \tbw_i \tbw_i^\intercal \big]^\dagger D_i^\dagger 
        = D_i \Big( D_i^\intercal D_i \Big)^{-1} \E \big[ \tbw_i \tbw_i^\intercal \big]^\dagger \Big( D_i^\intercal D_i \Big)^{-1} D_i^\intercal.
\end{align}
    Noting that $D_i^\intercal D_i = \textrm{diag} \Big( M_i^\beta(\cS_1), \hdots, M_i^\beta(\cS_\ell) \Big)$, we have
    \begin{align*}
        \Big[ \E \big[ \tbz_i \tbz_i^\intercal \big]^{\dagger} \Big]_{S,T} = \tfrac{1}{M_i^\beta(\cC(S))} \cdot \Big[ \E \big[ \tbw_i \tbw_i^\intercal \big]^{\dagger} \Big]_{\cC(S),\cC(T)} \cdot \tfrac{1}{M_i^\beta(\cC(T))}.
    \end{align*}
\hfill $\square$
\endproof

This lemma tells us that the entries within each block of this cluster design matrix are normalized entries of the unit-randomized design matrix corresponding to $\bw$. The normalization factors correspond to the sizes of each matrix block. As a result, we can recast the analysis of cluster-randomized designs as the analysis of a unit-randomized design in which clusters play the roles of units.

Since $\E \big[ \tbw_i \tbw_i^\intercal \big]$ has the structure of a $\Bern(p)$ design matrix, we use the explicit formula for these entries derived by \citet{cortez2023exploiting} (their Lemma 1), where the clusters are viewed as units that are independently treated with probability $p$. This relationship between the Bernoulli unit randomized design and the Bernoulli graph clustered randomized design enables us to easily obtain an explicit form for the pseudoinverse estimator under GCR designs, as stated below.

\begin{lemma}
\label{lem:pi_gcr_est_form}
    Fix a clustering $\cC$ and suppose that $\bz \sim \GCR(\cC,p)$. The potential outcomes model can be equivalently expressed as a function of $\bw$ instead of $\bz$, for each individual $i \in [n]$,
\begin{align*}
Y_i(\bz) = \langle \tbw_i, \bx_i \rangle =: Y_i(\bw), ~\text{where}~ \bx_i = D_i^\intercal \bc_i.
\end{align*}
The pseudoinverse estimator for the $\TTE$ is
    \begin{align}
        \label{eq:gcr_bern_tte_hat}
        \widehat{\TTE}_{\beta} 
        &= \frac{1}{n} \sum_{i=1}^{n} Y_i(\bw) \Big\langle \E \big[ \tbw_i \tbw_i^\intercal \big]^{\dagger} \psi_i, \tbw_i \Big\rangle, \\
        &= \frac{1}{n} \sum_{i=1}^{n} \; Y_i(\bw) \sum_{\cU \in \cC_i^\beta} \bigg( \prod_{C \in \cU} \frac{w_C - p}{p} - \prod_{C \in \cU} \frac{w_C - p}{p-1} \bigg), \nonumber
    \end{align}
where $\psi_i = [ \, 0 \; 1 \; \cdots \; 1 \,]^\intercal$ denotes the length-$|\cC_i^\beta|$ column vector with $0$ in its first entry (corresponding to $\cU = \varnothing$) and $1$ in all other entries.
\end{lemma}

\proof{\it Proof.}
By plugging in $\tbz_i = D_i \tbw_i$ to the potential outcomes model and rearranging the inner product, it follows that $Y_i(\bz) = \langle \tbz_i, \bc_i \rangle = \langle D_i \tbw_i, \bc_i \rangle = \langle \tbw_i, D_i^\intercal \bc_i \rangle = \langle \tbw_i, \bx_i \rangle$.

Substituting $\tbz_i^{\beta} = D_i \tbw_i$, $\theta_i = D_i \psi_i$, and \eqref{eq:inverse_design_GCR} into \eqref{eq:tte_hat}, it follows that
\begin{align*}
    \widehat{\TTE}_{\beta} 
    &= \frac{1}{n} \sum_{i=1}^{n} Y_i(\bw) \Big\langle D_i \Big( D_i^\intercal D_i \Big)^{-1} \E \big[ \tbw_i (\tbw_i)^\intercal \big]^\dagger \Big( D_i^\intercal D_i \Big)^{-1} D_i^\intercal D_i \psi_i, D_i \tbw_i \Big\rangle \\
    &= \frac{1}{n} \sum_{i=1}^{n} Y_i(\bw) \Big\langle \E \big[ \tbw_i \tbw_i^\intercal \big]^\dagger \psi_i, \tbw_i \Big\rangle.
\end{align*}
Expressed in this way, we note that $\Big\langle \E \big[ \tbw_i \tbw_i^\intercal \big]^\dagger \psi_i, \tbw_i \Big\rangle$ only depends on $\bw$, which follows a Bernoulli unit randomized design over $m$ units, where each cluster is treated as a unit. As a result, the exact expression of the estimator is given by the SNIPE estimator as stated in equation (4.4) of \citet{cortez2023exploiting}, but treating clusters as units, resulting in 
\begin{equation} \label{eq:tte_hat_cluster}
    \widehat{\TTE}_\beta = \frac{1}{n} \sum_{i=1}^{n} Y_i(\bw) \sum_{\cU \in \cC_i^\beta} \Big( \prod_{C \in \cU} \frac{w_C-p}{p} - \prod_{C \in \cU} \frac{w_C-p}{p-1} \Big).
\end{equation}
\hfill $\square$
\endproof

\subsection{Relationship to the Horvitz-Thompson Estimator} 
\label{subsec:ht}

\meedit{The classical Horvitz-Thompson estimator \citep{horvitz1952generalization} defined with respect to the neighborhood exposure mapping is an inverse-probability-weighted estimator for the total treatment effect given by 
\begin{equation}
	\label{eq:tte_ht}
	\widehat{\TTE}_{\textrm{HT}} = \frac{1}{n} \sum_{i=1}^{n} Y_i(\bz) \bigg[ \frac{\Ind\big( \bigcap_{j \in \cN_i} z_j = 1 \big)}{\Pr\big( \bigcap_{j \in \cN_i} z_j = 1 \big)} - \frac{\Ind\big( \bigcap_{j \in \cN_i} z_j = 0 \big)}{\Pr\big( \bigcap_{j \in \cN_i} z_j = 0 \big)} \bigg].
\end{equation}

\noindent The Horvitz-Thompson estimator is unbiased whenever Assumption~\ref{ass:neighborhood_interference} (neighborhood interference) is satisfied, although it may have prohibitively large variance when the exposure probabilities $\Pr\big(\bigcap_{j\in \cN_i} z_j=1\big)$ are small \citep{ugander2013}. For an individual $i$ with degree $d_i \leq \beta$ (or whose neighborhood intersects $\leq \beta$ clusters under a Bernoulli GCR design, i.e. $|\cC(\cN_i)| \leq \beta$), the weights on $Y_i$ are the same under the pseudoinverse and Horvitz-Thompson estimators. Using the cluster notation from \eqref{eq:tte_hat_cluster}, we can rewrite the summation over $\cU \in \cC_i^\beta$ as simply the summation $\cU \subseteq \cC_i$, and simplify:
\[
    \sum_{\cU \subseteq \cC_i} \Big( \prod_{C \in \cU} \frac{w_C-p}{p} - \prod_{C \in \cU} \frac{w_C-p}{p-1} \Big) = \prod_{C \in \cC_i} \frac{w_C}{p} - \prod_{C \in \cC_i} \frac{1-w_C}{1-p} = \frac{\Ind\big( \bigcap_{j \in \cN_i} z_j = 1 \big)}{\Pr\big( \bigcap_{j \in \cN_i} z_j = 1 \big)} - \frac{\Ind\big( \bigcap_{j \in \cN_i} z_j = 0 \big)}{\Pr\big( \bigcap_{j \in \cN_i} z_j = 0 \big)}.
\]

In particular, this shows that $\widehat{\TTE}_\beta$ and $\widehat{\TTE}_\textrm{HT}$ coincide when $\beta \geq \dmax$ or when $\beta \geq C$ under a Bernoulli GCR design where each individual is connected to at most $C$ clusters. Intuitively, the low-order interactions assumption does not effectively impose any further constraints beyond neighborhood interference. As a result, both estimators are derived from satisfying unbiasedness conditions from the same model. 

For units where $\beta < |\cC(\cN_i)|$, the two estimators differ in how they treat units whose neighborhoods are partially treated. The Horvitz-Thompson estimator ignores these units; it considers only fully treated or fully untreated units, and it aggressively scales their outcomes to achieve unbiasedness (resulting in a high variance). On the other hand, the pseudoinverse estimator leverages the low-order structure of the potential outcomes (\ref{ass:low_deg}) to appropriately scale each unit's potential outcome to interpolate between these extremes. In graphs that can be clustered well, such that many units are ``internal'' to their cluster (i.e., have neighborhoods entirely contained within a single cluster), the estimators will align, and both perform well. We illustrate this claim in a set of experiments described in Section~\ref{subsec:gcr_exp}, where increasing the cluster size increases the prevalence of these internal units. On the other hand, in graphs that don't admit a natural clustering, such that many units have neighborhoods spanning multiple clusters, $\widehat{\TTE}_\beta$ can leverage the $\beta$-order effects assumption to achieve a significantly lower variance.}


\section{Bias} \label{sec:bias}

In this section, we present bounds on the bias of the pseudoinverse estimator, $\widehat{\TTE}_{\beta}$, as a function of the experimental design and use these bounds to identify conditions under which it is unbiased. Within our analysis, we will see that selecting a value $\beta$ less than the ground-truth model parameter $\beta^*$ serves as a potential source of bias. To quantify this bias, we'll define $\tbz_i^{> \beta} \in \mathbb{R}^{|\cS_i^{\beta^*} \setminus \cS_i^\beta|}$ to collect the treatment indicators of all subsets $S$ with $\beta < |S| \leq \beta^*$, and similarly define $\bc_i^{> \beta}$ to collect the coefficients on these subsets. We'll also let $\bc_i^\beta$ collect the coefficients for subsets $S$ with $|S| \leq \beta$ and $\theta_i^\beta$ be the projection of the vector $\theta_i$ onto these coordinates. Using this new notation, we can state the following theorem.

\begin{theorem} \label{thm:bias_bound_general}
    Suppose the potential outcomes $Y_i(\bz)$ satisfy Assumption~\ref{ass:low_deg} for some ground-truth parameter $\beta^*$. If $\widehat{\TTE}_{\beta}$ is the pseudoinverse estimator corresponding to a selected model order $\beta$, then its exact bias $\E[\widehat{\TTE}_{\beta}] - \TTE$ is,
    \begin{equation} \label{eq:tte_bias_exact}
        \frac{1}{n}\sum_{i=1}^n \bigg( \Big\langle \bc_i^\beta, \Big( \E\big[\tbz_i^\beta \big(\tbz_i^\beta\big)^{\intercal}\big] \E\big[\tbz_i^\beta \big(\tbz_i^\beta\big)^{\intercal}\big]^\dagger - I \Big) \, \theta_i^\beta \Big\rangle 
        + \Big\langle \bc_i^{>\beta}, \E\big[\tbz_i^{>\beta} \big(\tbz_i^\beta\big)^{\intercal}\big] \E\big[\tbz_i^\beta \big(\tbz_i^\beta\big)^{\intercal}\big]^\dagger \theta_i^\beta - \mathbf{1} \Big\rangle \bigg),
    \end{equation}
    where $I$ denotes the identity matrix and $\mathbf{1}$ denotes the all-1s vector of the appropriate dimensions.
\end{theorem}

The first term on the right-hand side of \eqref{eq:tte_bias_exact} measures the extent to which the design matrix fails to be invertible in the direction of $\theta_i$. In other words, it represents the inability of the chosen experimental design to capture the TTE estimand. The second term, which is non-zero only when $\beta < \beta^*$, measures bias due to misspecification of the model order $\beta^*$. When $\beta < \beta^*$, the pseudoinverse estimator fails to model some interference effects, contributing to the bias.

\proof{\it Proof.}
For each unit $i$, we have $Y_i(\bz) = \big\langle \bc_i^{\beta^*}, \tbz_i^{\beta^*}  \big\rangle = \big\langle \bc_i^{\beta}, \tbz_i^{\beta} \big\rangle + \big\langle \bc_i^{> \beta}, \tbz_i^{> \beta} \big\rangle$. The actual contribution of unit $i$ to the $\TTE$ is $\big\langle \bc_i, \theta_i \big\rangle = \big\langle \bc_i^\beta, \theta_i^\beta \big\rangle + \big\langle \bc_i^{> \beta}, \mathbf{1} \big\rangle$, where $\mathbf{1}$ is the vector of all ones of length $|\cS_i^{\beta^*} \setminus \cS_i^\beta|$. On the other hand, our estimate of $\bc_i$ is 
\begin{align*}
    \hat{\bc_i} &= Y_i(\bz) \cdot \E\big[\tbz_i^\beta \big(\tbz_i^\beta\big)^{\intercal}\big]^\dagger \tbz_i^\beta 
    = \E\big[\tbz_i^\beta \big(\tbz_i^\beta\big)^{\intercal}\big]^\dagger \tbz_i^\beta \big(\tbz_i^\beta\big)^{\intercal} \bc_i^\beta + \E\big[\tbz_i^\beta \big(\tbz_i^\beta\big)^{\intercal}\big]^\dagger \tbz_i^\beta \big(\tbz_i^{>\beta}\big)^{\intercal} \bc_i^{>\beta}.
\end{align*}
Therefore, the expected contribution of unit $i$ to $\widehat{\TTE}_\beta$ is 
\begin{equation*}
    \E\Big[\big\langle \hat{\bc_i}, \theta_i \big\rangle \Big] 
    = \Big\langle \bc_i^\beta, \E\big[\tbz_i^\beta \big(\tbz_i^\beta\big)^{\intercal}\big] \E\big[\tbz_i^\beta \big(\tbz_i^\beta\big)^{\intercal}\big]^\dagger \theta_i^\beta \Big\rangle 
    + \Big\langle \bc_i^{>\beta}, \E\big[\tbz_i^{>\beta} \big(\tbz_i^\beta\big)^{\intercal}\big] \E\big[\tbz_i^\beta \big(\tbz_i^\beta\big)^{\intercal}\big]^\dagger \theta_i^\beta \Big\rangle.
\end{equation*}
Thus, the bias contribution from unit $i$ is 
\begin{align*}
    \E\Big[\big\langle \hat{\bc_i}, \theta_i \big\rangle \Big] - \big\langle \bc_i, \theta_i \big\rangle
    &= \Big\langle \bc_i^\beta, \Big( \E\big[\tbz_i^\beta \big(\tbz_i^\beta\big)^{\intercal}\big] \E\big[\tbz_i^\beta \big(\tbz_i^\beta\big)^{\intercal}\big]^\dagger - I \Big) \, \theta_i^\beta \Big\rangle \\
    &\qquad + \Big\langle \bc_i^{>\beta}, \E\big[\tbz_i^{>\beta} \big(\tbz_i^\beta\big)^{\intercal}\big] \E\big[\tbz_i^\beta \big(\tbz_i^\beta\big)^{\intercal}\big]^\dagger \theta_i^\beta - \mathbf{1} \Big\rangle.
\end{align*}
Summing over $i$ gives the result.
\hfill $\square$
\endproof

\noindent The following corollary identifies a condition under which $\widehat{\TTE}_\beta$ is unbiased.

\begin{corollary} \label{cor:gcr_unbiased}
	If $\beta \geq \beta^*$ and $\theta_i$ lies in the column space of $\E \big[ \tbz_i^\beta \big(\tbz_i^\beta\big)^{\intercal} \big]$ for each $i \in [n]$, then $\E\Big[\widehat{\TTE}_{\beta}\Big]=\TTE$. It follows that for a fixed clustering $\cC$ and $\bz\sim \GCR(\cC, p)$, the pseudoinverse estimator $\widehat{\TTE}_{\beta}$ of \eqref{eq:gcr_bern_tte_hat} is unbiased, so that $\E\Big[\widehat{\TTE}_{\beta}\Big]=\TTE$. 
\end{corollary}

\proof{\it Proof.} 
If $\theta_i^\beta$ lies in the column space of $\E\big[\tbz_i^\beta \big(\tbz_i^\beta\big)^{\intercal}\big]$, then $\E\big[\tbz_i^\beta \big(\tbz_i^\beta\big)^{\intercal}\big]^\dagger \E\big[\tbz_i^\beta \big(\tbz_i^\beta\big)^{\intercal}\big] \theta_i^\beta = \theta_i^\beta$, so the first term of \eqref{eq:tte_bias_exact} is zero. The latter terms are also zero as $\beta \geq \beta^*$, so $\bc_i^{>\beta}$ and $\tbz_i^{>\beta}$ are empty.

For a fixed clustering $\cC$ and $\bz \sim \GCR(\cC, p)$, it suffices to verify that $\theta_i^\beta$ lies in the column space of the design matrix $\E\big[ \tbz_i^\beta \big(\tbz_i^\beta\big)^\intercal \big]$ for graph cluster randomization. As $\theta_i^\beta = D_i \, \psi_i$, it follows that $\theta_i$ lies in the column space of $D_i$. Since $\E \big[ \tbw_i \tbw_i^\intercal \big] D_i^\intercal$ has full (row) rank, $\theta_i$ also lies the column space of $\E\big[ \tbz_i^\beta \big(\tbz_i^\beta\big)^\intercal \big] = D_i \E \big[ \tbw_i \tbw_i^\intercal \big] D_i^\intercal$.
\hfill $\square$
\endproof

For a given design, we can verify the unbiasedness of the pseudoinverse estimator by checking a linear algebraic condition. Note that a sufficient condition for unbiasedness is the invertibility of the design matrix (as is the case for unit Bernoulli randomization), which ensures that the column space spans $\mathbb{R}^{|\cS_i^\beta|}$. However, the above corollary shows that this condition is not necessary; GCR designs admit an unbiased pseudoinverse estimator even though their design matrix is rank-deficient for any non-trivial clustering.

When $\widehat{\TTE}$ is biased, it is not an instance of the Riesz estimator defined by \citet{harshaw2022design}, since the Riesz estimator is always unbiased. In the case of completely randomized designs (see Appendix~\ref{sec:crd}), this is because the positivity condition necessary for the unbiasedness of the Riesz estimator is not satisfied.

\begin{corollary} \label{cor:gcr_bias_bound}
    If $\beta < \beta^*$, for a clustering $\cC$ and $\bz\sim \GCR(\cC, p)$, the bias is upper bounded as
    \begin{align} \label{eq:gcr_bias_bound}
        |\E[\widehat{\TTE}_{\beta}] - \TTE| &\leq \|\bx_i^{>\beta}\|_1 \leq \|\bc_i^{ >\beta}\|_1,
    \end{align}
    where $\bx_i^{>\beta}$ contains only the coefficients that correspond to subsets $\cU \in \cC_i^{\beta^*}$ with $|\cU| > \beta$.
\end{corollary}

A proof of this corollary is given in Appendix~\ref{sec:proofs_pi_gcr}. In the proof, we re-express \eqref{eq:tte_bias_exact} in terms of $\tbw_i$, $\bx_i$, $\psi_i$ to reduce to the setting of Bernoulli unit randomization (viewing each cluster as a unit). Then, we use the definition of $\GCR(\cC, p)$ designs and a formula from \citet{cortez2023exploiting} to obtain expressions for the entries of the matrices in this formula. Finally, we invoke the binomial theorem and other combinatorial identities to simplify the bias to an expression that we can bound.

The form of this bound is quite natural: when the higher-order coefficients, $\bc_i^{ >\beta}$, which we forego estimating when choosing to use a small $\beta$ order parameter, are small, the bias from this choice is small as well. This is analogous to bias from misspecification in ordinary least squares, which depends on the distance between the true parameter and the projection of the true parameter onto the set from which an estimate was chosen, but the proof in our setting is more subtle. 
This subtlety is because using, e.g., $\beta=1$ instead of $\beta=2$ is not equivalent to replacing the $\beta=2$ estimates of $c_{i, S}$ with $|S|=2$ by 0. In this example, our estimates of the coefficients $c_{i,S}$ with $|S|=1$ also change, since the estimate of every coefficient depends on the choice of $\beta$ through the pseudoinverse of the design matrix. Thus, we must carefully track the effect of misspecification, as we do in Corollary~\ref{cor:gcr_bias_bound}.
\section{Variance} \label{sec:variance}

    We now present a generic bound that identifies features of the experimental design and interference graph that control the variance of $\widehat{\TTE}_\beta$. This variance bound only relies on the neighborhood interference assumption and boundedness of the potential outcomes. In particular, it does not rely on the true outcomes satisfying the low-order assumption (Assumption~\ref{ass:low_deg}).
    
    \begin{theorem}
    \label{thm:tte_var_bound}
    Suppose that the potential outcomes satisfy neighborhood interference (Assumption~\ref{ass:neighborhood_interference}) and boundedness, i.e. $|Y_i(\bz)| \leq B$ for all $\bz \in \{0,1\}^n$ (Assumption~\ref{ass:bounded}). Then for any selected model order $\beta$, 
    \begin{equation} \label{eq:tte_var_bound}
	\var(\widehat{\TTE}_\beta) 
        \leq \frac{B^2}{n^2} \sum_{i,j=1}^n \gamma_i\gamma_j \cdot \Ind\big(\cov\big( \big\langle \widehat{\bc}_i^\beta, \theta_i^\beta \big\rangle, \big\langle \widehat{\bc}_j^\beta, \theta_j^\beta \big\rangle \big) > 0\big) 
        \leq \frac{B^2}{n^2}\sum_{i,j=1}^n \gamma_i\gamma_j \cdot \Ind \big(\bz_{\cN_i} \not\perp \bz_{\cN_j} \big).
    \end{equation}
    where $\displaystyle \gamma_i=\sqrt{(\theta_i^\beta)^\intercal\E\big[\tbz_i^\beta (\tbz_i^\beta)^\intercal\big]^\dagger\theta_i^\beta}$. 
    \end{theorem}

\proof{\it Proof.} 
    Using Bienaym\'e's identity, we may write 
    \begin{equation} \label{eq:tte_var}
        \var\Big( \widehat{\TTE}_\beta \Big) 
        = \frac{1}{n^2} \sum_{i,j=1}^n \cov\big( \big\langle \widehat{\bc}_i^\beta, \theta_i^\beta \big\rangle, \big\langle \widehat{\bc}_j^\beta, \theta_j^\beta \big\rangle \big)
        = \frac{1}{n^2} \sum_{i,j=1}^n \cov\big( \big\langle \widehat{\bc}_i^\beta, \theta_i^\beta \big\rangle, \big\langle \widehat{\bc}_j^\beta, \theta_j^\beta \big\rangle \big) \cdot \Ind \big(\bz_{\cN_i} \not\perp \bz_{\cN_j} \big).
    \end{equation} 
    Here, the second equality follows because $\big\langle \widehat{\bc}_i^\beta, \theta_i^\beta \big\rangle$ is a function of $Y_i(\bz)$ and $\tbz_i^\beta$ and $\big\langle \widehat{\bc}_j^\beta, \theta_j^\beta \big\rangle$ is a function of $Y_j(\bz)$ and $\tbz_j^\beta$, so under the neighborhood interference assumption, $\bz_{\cN_i} \perp \bz_{\cN_j}$ implies that $\cov\big( \big\langle \widehat{\bc}_i^\beta, \theta_i^\beta \big\rangle, \big\langle \widehat{\bc}_j^\beta, \theta_j^\beta \big\rangle \big) = 0$. Restricting the double summation to only positive covariance terms gives an upper bound on the variance.
    \begin{equation} \label{eq:tte_var_ub}
        \var\Big( \widehat{\TTE}_\beta \Big) 
        \leq \tfrac{1}{n^2}\sum_{i,j=1}^n \cov\big( \big\langle \widehat{\bc}_i^\beta, \theta_i^\beta \big\rangle, \big\langle \widehat{\bc}_j^\beta, \theta_j^\beta \big\rangle \big) \cdot \Ind\big( \cov\big( \big\langle \widehat{\bc}_i^\beta, \theta_i^\beta \big\rangle, \big\langle \widehat{\bc}_j^\beta, \theta_j^\beta \big\rangle \big) >0 \big).
    \end{equation}

By Cauchy-Schwarz,
\[
    \Big|\cov\big( \big\langle \widehat{\bc}_i^\beta, \theta_i^\beta \big\rangle, \big\langle \widehat{\bc}_j^\beta, \theta_j^\beta \big\rangle \big)\Big|
    \:\leq\: \sqrt{\var\big(\big\langle \widehat{\bc}_i^\beta, \theta_i^\beta \big\rangle\big) \cdot \var\big(\big\langle \widehat{\bc}_j^\beta, \theta_j^\beta \big\rangle\big)}
    \:\leq\: \sqrt{\E\big[\big\langle \widehat{\bc}_i^\beta, \theta_i^\beta \big\rangle^2\big] \cdot \E\big[\big\langle \widehat{\bc}_j^\beta, \theta_j^\beta \big\rangle^2\big]}.
\]
Using Assumption~\ref{ass:bounded}, we have
\begin{align}
\E\Big[\big\langle \widehat{\bc}_i^\beta, \theta_i^\beta \big\rangle^2\Big]
&= \E\Big[ Y_i(\bz)^2 \cdot \Big\langle \E \big[ \tbz_i^\beta (\tbz_i^\beta)^\intercal \big]^{\dagger} \tbz_i^\beta \,,\, \theta_i^\beta \Big\rangle^2\,\Big] \nonumber \\
&\leq B^2 \cdot \E \Big[ \Big\langle \E \big[ \tbz_i^\beta (\tbz_i^\beta)^\intercal \big]^{\dagger} \theta_i^\beta \,,\, \tbz_i^\beta \Big\rangle^2\,\Big] \label{eq:B_bound} \\
&= B^2 \cdot \E\Big[ (\theta_i^\beta)^\intercal \E \big[ \tbz_i^\beta (\tbz_i^\beta)^\intercal \big]^{\dagger} \tbz_i^\beta (\tbz_i^\beta)^\intercal \E \big[ \tbz_i^\beta (\tbz_i^\beta)^\intercal \big]^{\dagger} \theta_i^\beta \Big] \nonumber  \\
&= B^2 \cdot \big(\theta_i^\beta\big)^\intercal \E \big[ \tbz_i^\beta (\tbz_i^\beta)^\intercal \big]^{\dagger} \E \big[ \tbz_i^\beta (\tbz_i^\beta)^\intercal \big] \E \big[ \tbz_i^\beta (\tbz_i^\beta)^\intercal \big]^{\dagger} \theta_i^\beta \nonumber  \\
&= B^2 \cdot (\theta_i^\beta)^\intercal \E \big[ \tbz_i^\beta (\tbz_i^\beta)^\intercal \big]^{\dagger} \theta_i^\beta. \nonumber 
\end{align}
Here, the final equality follows from properties of the pseudoinverse. Therefore, it follows that $\Big|\cov\big( \big\langle \widehat{\bc}_i^\beta, \theta_i^\beta \big\rangle, \big\langle \widehat{\bc}_j^\beta, \theta_j^\beta \big\rangle \big)\Big| \leq B^2 \cdot \gamma_i \gamma_j$ for $\gamma_i = \sqrt{(\theta_i^\beta)^\intercal \E \big[ \tbz_i^\beta (\tbz_i^\beta)^\intercal \big]^{\dagger} \theta_i^\beta}$. Plugging this into \eqref{eq:tte_var} and \eqref{eq:tte_var_ub} completes the proof. Note that the second expression in the inequality chain of \eqref{eq:tte_var_bound} is a lower bound on the third because it includes a subset of its (entirely non-negative) terms.

$\square$
\endproof

    Each term of the ultimate variance bound of Theorem~\ref{thm:tte_var_bound} has two components. The first component is the discrete $\Ind\big(\tbz_i^\beta\not\perp\tbz_j^\beta\big)$, which tracks which terms contribute to the variance. This independence condition is easy to interpret and verify for many standard experimental designs, including $\GCR(\cC, p)$, which we consider next. We can further constrain this indicator to include only unit pairs $(i,j)$ with $\cov\big( \big\langle \widehat{\bc}_i^\beta, \theta_i^\beta \big\rangle, \big\langle \widehat{\bc}_j^\beta, \theta_j^\beta \big\rangle \big) > 0$. This form of the bound is better for the analysis of completely randomized designs (see Appendix~\ref{sec:crd}), where most covariance terms are negative.

    The second component is the continuous term $\gamma_i\gamma_j$, which captures the contributions of units $i$ and $j$ to the variance. Under this bound, the quantity $\gamma_i$ is a natural proxy for the contribution of unit $i$ to the variance of the pseudoinverse estimator under a particular design and interference graph. To better understand $\gamma_i$, note that it depends on the quadratic form $(\theta_i^\beta)^\intercal \E \big[ \tbz_i^\beta (\tbz_i^\beta)^\intercal \big]^{\dagger} \theta_i^\beta$, which is a sum of certain entries of $\E \big[ \tbz_i^\beta (\tbz_i^\beta)^\intercal \big]^{\dagger}$. These entries will be small when the entries of $\E \big[ \tbz_i^\beta (\tbz_i^\beta)^\intercal \big]$ are large, so $\gamma_i$ will be small when the treatments within $\cN_i$ are correlated. Thus, Theorem~\ref{thm:tte_var_bound} can be thought of as a quantitative form of the intuition that correlating the treatments of adjacent nodes in the interference graph reduces the variance of treatment effect estimates.

    The variance bounds from Theorem~\ref{thm:tte_var_bound} depend only on the \textit{selected} model order $\beta$ and do not require the ground truth outcomes to satisfy the low-order assumption. We will see below that the variance scales exponentially in $\beta$, so over-estimating $\beta^*$ can have a deleterious effect on the error of the pseudoinverse estimator. Often, the bias arising from the misspecification of $\beta$ is small; for Bernoulli GCR designs, Corollary~\ref{cor:gcr_bias_bound} bounds the bias by the sum of magnitudes of the higher-order effects, implying that the bias is small if the outcomes are approximately low-order. This suggests that, in many cases, a practitioner without knowledge of the ground truth $\beta^*$ is better off selecting a small value $\beta$ (such as $\beta = 1$).

    Theorem~\ref{thm:tte_var_bound} can be understood as a bound on the worst-case variance over all outcome models with $|Y_i(\bz)| \leq B$, and so it can be used to select a clustering that is likely to have moderate variance under any such model. We explore this approach empirically in Section~\ref{subsec:choose_cluster} and find that it has good properties in practice. As an extension, both the bounds of Theorem~\ref{thm:tte_var_bound} and thus the practice of cluster selection can be refined based on more precise domain knowledge of $\bc_i$. A full execution of the \citet{viviano2023causal} ``causal clustering'' framework under low-order outcomes---minimizing a mean squared error objective directly instead of bounds---remains an interesting direction for future work. 

    We next focus on Bernoulli GCR designs to build intuition for how the variance bound reflects properties of the design as it relates to the clustering. The results are facilitated by access to the explicit analytical forms of $\E\big[ \tbz_i^\beta (\tbz_i^\beta)^\intercal \big]$ and $\E\big[ \tbz_i^\beta (\tbz_i^\beta)^\intercal \big]^{\dagger}$, but we emphasize that the utility of our bounds is in no way limited to the setting where such analytical calculations are feasible. We use these calculations to provide easily interpretable theoretical results, but our bounds can be evaluated numerically given nothing more than sampling access to the distribution. This sampling access can be used to estimate $\E\big[ \tbz_i^\beta (\tbz_i^\beta)^\intercal \big]$, which can be pseudoinverted numerically, and then all bounds can be computed efficiently. This approach is also possible for designs even more complex than those we consider (e.g., \citet{pouget2019variance}, \citet{ugander2023randomized}, \citet{candogan2023correlated}, or \citet{kandiros2024conflict}). By sampling and numerical pseudoinversion, our bounds can, e.g., be used to select a clustering for graph clustered designs, as we demonstrate in Section~\ref{subsec:choose_cluster}. 

\begin{remark}
    When using \eqref{eq:tte_var_bound} to compare the variance across multiple candidate designs, the tightness of this bound becomes relevant. There are only a few inequalities used in the proof of Theorem \ref{thm:tte_var_bound}. Perhaps the most important step of the bound is that the covariance terms $\cov( \langle \widehat{\bc}_i^\beta, \theta_i^\beta \rangle, \langle \widehat{\bc}_j^\beta, \theta_j^\beta \rangle )$ are upper bounded by their magnitude, and then subsequently bounded by Cauchy-Schwarz. For Bernoulli designs under outcome models where the causal effects have the same sign, pairs $i$ and $j$ for which $\cN_i \cap \cN_j \neq \emptyset$ will contribute positive covariance terms. However, if the causal effects are allowed to have opposite signs, then pairs that have shared neighbors could have negatively correlated estimates, leading to a looser bound. The Cauchy-Schwarz inequality is only satisfied with equality if the estimates associated with $i$ and $j$ are linearly related, which could be satisfied in the unlikely situation that $i$ and $j$ have the same neighborhood sets, and the corresponding causal effects are proportional to each other. The Cauchy-Schwarz inequality is loosest for pairs $i,j$ whose neighborhoods only barely overlap, as the covariance could be very small. One could estimate the tightness of the variance bound by restricting the summation to only pairs whose neighborhood overlap is large. For graphs that exhibit strong clustering, one would expect that most pairs with overlapping neighborhoods also have a large overlap, indicating that the bound is loosest for graphs that are not well clustered. After applying Cauchy-Schwarz, we bound the variance of $\langle \widehat{\bc}_i^\beta, \theta_i^\beta \rangle$ by its second moment. This is only loose by an additive term of $B^2$ under the condition that the outcomes take values in the range $[0,B]$ as 
    \begin{align*}
        \var(\langle \widehat{\bc}_i^\beta, \theta_i^\beta \rangle) 
        = \E[\langle \widehat{\bc}_i^\beta, \theta_i^\beta \rangle^2] - \langle \bc_i^\beta, \theta_i^\beta \rangle^2 
        = \E[\langle \widehat{\bc}_i^\beta, \theta_i^\beta \rangle^2] - (Y_i({\bf 1}) - Y_i({\bf 0}))^2
        \geq \E[\langle \widehat{\bc}_i^\beta, \theta_i^\beta \rangle^2] - B^2.
    \end{align*}
    As the variance is eventually bounded by a factor of $B^2$, this additive gap of $B^2$ is not likely to be very significant when comparing variance across designs. The final inequality used in \eqref{eq:B_bound} upper bounds $Y_i(\bz)$ by $B$, pulling it out of the expectation. If the outcomes are furthermore lower bounded by a fraction of $B$ such that the ratio of the largest and smallest outcomes is bounded by $1/r$, then the inequality in \eqref{eq:B_bound} would be loose by at most a factor of $1/r$, suggesting that the overall bound is still useful for comparing across designs as long as $r$ is not too small.
\end{remark}

\subsection{Variance Bounds under Bernoulli GCR Designs} 

We now more precisely characterize the variance structure of the pseudoinverse estimator for $\GCR$ Bernoulli designs, which involves evaluating when $\bz_{\cN_i} \perp \bz_{\cN_j}$ and bounding $\gamma_i$. Under a Bernoulli clustered design, $\bz_{\cN_i} \perp \bz_{\cN_j}$ if and only if the corresponding cluster neighborhoods of $i$ and $j$ are disjoint. Lemma~\ref{lem:gamma_bound} provides an upper bound on $\gamma_i$ for the $\GCR$ Bernoulli design.

\begin{restatable}{lemma}{gammabound}
\label{lem:gamma_bound}
Suppose that $\bz\sim \GCR(\cC,p)$. Then,
\begin{equation} \label{eq:gamma_bound_gcr}
    \gamma_i^2 \leq 2 \min\left(p^{-|\cC(\cN_i)|}, |\cC(\cN_i)|^{\beta} p^{-\beta}\right).
\end{equation}
\end{restatable}

Note that this Lemma also holds if we replace $|\cC(\cN_i)|$ with any upper bound on the cluster neighborhood size.
This bound elucidates the interplay between the $\beta$-order assumption and the cluster randomization. When $|\cC(\cN_i)| \leq \beta$, we find that $\gamma_i^{2}$ scales like $p^{-|\cC(\cN_i)|}$, like it would under cluster randomization. When $|\cC(\cN_i)| \gg \beta$, we find that $\gamma_i^{2}$ scales like $|\cC(\cN_i)|^{\beta} p^{-\beta}$, like it would under the $\beta$-order assumption. 
This lemma sheds light on the basic value of low-order modeling: when a node $i$  has cluster neighborhood degree $|\cC(\cN_i)| \gg \beta$, it still contributes terms of order $p^{-\beta}$ to the variance, rather than $p^{-|\cC(\cN_i)|}$. The size of the cluster neighborhood is at most the degree of node $i$, i.e., $|\cC(\cN_i)| \leq |\cN_i| =: d_i$, but it can be much smaller for a node in the interior of a cluster $|\cC(\cN_i)| = |\cN_i| =: d_i$. Putting these together results in the following variance upper bound for the $\beta$-order pseudoinverse estimator under $\GCR(\cC,p)$ design.

\begin{corollary}[Clusterings with Bounded Cluster Neighborhood Sizes] \label{thm:pi_gcr}
Suppose that the potential outcomes satisfy neighborhood interference (Assumption~\ref{ass:neighborhood_interference}) and boundedness, i.e. $|Y_i(\bz)| \leq B$ (Assumption~\ref{ass:bounded}). The variance of the $\beta$-order pseudoinverse estimator under $\GCR(\cC,p)$ design is upper bounded by
\begin{align} \label{var_bd_GCR_general}
\frac{2 B^2}{n^2}\sum_{i,j=1}^n \sqrt{\min\left(p^{-|\cC(\cN_i)|}, |\cC(\cN_i)|^{\beta} \cdot p^{-\beta}\right), \min\left(p^{-|\cC(\cN_j)|}, |\cC(\cN_j)|^{\beta} \cdot p^{-\beta}\right)} \cdot \Ind \big(\cC(\cN_i) \cap \cC(\cN_j) \neq \emptyset \big).
\end{align}
Suppose further that each individual $i$ has a bounded cluster neighborhood size $|\cC(\cN_i)|\leq C$ and each cluster has at most $N$ units. Then
\begin{align} \label{var_bd_GCR_simplified}
    \var(\widehat{\TTE}_{\beta}) \leq 
        \frac{2 B^2 C N \dmax}{n} \cdot \min\left(\frac{1}{p^C}, \frac{C^\beta}{p^\beta}\right).
\end{align}
\end{corollary}

\proof{\it Proof.} 
The bound in \eqref{var_bd_GCR_general} follows from applying Lemma \ref{lem:gamma_bound} to Theorem \ref{thm:tte_var_bound} along with the fact that under cluster randomized designs, $\Ind \big(\bz_{\cN_i} \not\perp \bz_{\cN_j} \big) = \Ind \big(\cC(\cN_i) \cap \cC(\cN_j) \neq \emptyset \big)$. As each unit $i$ is connected to at most $C$ clusters, and since each of those clusters has at most $N$ units, each of which has at most $\dmax$ neighbors, there are at most $n C N \dmax$ non-zero terms in \eqref{var_bd_GCR_general}, each of which are bounded by $\min\left(\frac{1}{p^C}, \frac{C^\beta}{p^\beta}\right)$, implying the bound in \eqref{var_bd_GCR_simplified}.
$\square$
\endproof

The variance bound in \eqref{var_bd_GCR_general} can be used to evaluate different choices of clusterings $\cC$ in order to choose a clustering that minimizes this worst-case variance bound. Corollary \ref{thm:pi_gcr} shows that using the $\beta$-order pseudoinverse estimator together with cluster randomization achieves a variance that combines the benefits gained by both the low-order structure and the clustering structure. Both the low-order pseudoinverse estimator and the cluster-randomized design aim to reduce variance by reducing the ``effective degree" of a node. Without either assumption, the variance of the Horvitz-Thompson estimator under a Bernoulli unit randomized design typically scales like $p^{-d_i}$. In comparison, the variance of the $\beta$-order pseudoinverse estimator scales like $p^{-\beta}$, and the variance of the Horvitz--Thompson estimator under cluster-randomization scales like $p^{-|\cC(\cN_i)|}$. By combining the two, the variance scales as the minimum of $p^{-|\cC(\cN_i)|}$ and $p^{-\beta}$.

In contrast to the literature, which typically only considers optimal estimators for fixed designs or optimal designs for fixed estimators, our result shows that jointly considering the choice of estimator and design can lead to performance gains. In particular, \eqref{var_bd_GCR_general} shows that by combining the $\beta$-order pseudoinverse estimator with a clustered randomized design, we ensure that nodes for which $|\cC(\cN_i)| \ll \beta$ contribute order $p^{-|\cC(\cN_i)|}$ to the variance, and nodes for which $|\cC(\cN_i)| \gg \beta$ contribute order $p^{-\beta}$ to the variance. When the cluster neighborhood degree could be highly non-uniform over the network, combining both low-order estimators with clustered designs can thus perform much better than a low-order estimator with a unit-randomized design or the Horvitz-Thompson estimator with a clustered design. Many natural clusterings do exhibit highly varying cluster neighborhood degrees, as nodes in the interior of a cluster have $|\cC(\cN_i)|=1$, whereas nodes at the boundaries of clusters can have cluster neighborhood degrees as large as $|\cN_i|$.

The Horvitz-Thompson estimator, which is equivalent to the pseudoinverse estimator when $\beta=\dmax$, pays a cost exponential in $|\cC(\cN_i)|$ in the $\gamma_i$ terms of the variance. This means that even if the clustering is ``mostly good,'' i.e., a majority of individuals have low cluster neighborhood sizes, the exponential contribution to the variance from a small fraction of individuals with large cluster neighborhood sizes can quickly accumulate. In contrast, the $\beta$-order pseudoinverse estimator upper bounds each of the $\gamma_i$ by exponential in $\min(|\cC(\cN_i)|,\beta)$. As a result, the pseudoinverse estimator is less sensitive to the tail of the distribution of cluster neighborhood sizes. This is more amenable to real-world graphs, where even the best clusterings may have a small number of individuals with large cluster neighborhood sizes.

The quality of the clustering affects the above variance bound with respect to the number/size of the clusters, along with the cluster connectivity. In particular, the variance scales inversely with the number of clusters, which is lower bounded by $n/N$, instead of the number of individuals as under unit randomized designs. The variance scales exponentially with the minimum of $C$ and $\beta$, again showing how the combination of both the $\beta$-order pseudoinverse estimator and the cluster randomized design helps to reduce variance. Our result in Corollary \ref{thm:pi_gcr} recovers the bound from Theorem 1 of \citet{cortez2023exploiting} for unit randomized designs by setting $C=\dmax$ and $N=1$.

\citet{ugander2013} show in their Proposition 4.2 that a 3-net clustering of a graph $G$ with restricted growth coefficient $\kappa$ constrains the cluster neighborhood size to $|\cC(\cN_i)|\leq \kappa^3$ for all units $i$, allowing us to control the number of pairs of units that share cluster neighborhoods. Additionally as $|\cC(\cN_i)|\leq \dmax$, this implies $|\cC(\cN_i)| \leq \min(\kappa^3, \dmax)$. 
Corollary~\ref{thm:pi_gcr} thus implies a variance upper bound that is polynomial in $\kappa$ for $\kappa$-restricted growth graphs for the $\beta$-order pseudoinverse estimator, rather than exponential in $\kappa$ like variance bounds for Horvitz-Thompson under cluster randomization without a low-order assumption. 

\subsection{Improvement on Horvitz--Thompson}

The $\beta$-order pseudoinverse estimator is derived under the assumption of a $\beta$-order potential outcomes model, which is a smaller class of outcome functions for $\beta \ll d_{\max}$. As such, one might wonder whether it will always have lower variance than the Horvitz-Thompson estimator, which does not impose any structure on the potential outcomes. Unfortunately, this is not always the case, as demonstrated by the following example.

\begin{example}
    Consider a small instance with $n=3$, $\beta^* = 2$, $p = \frac{1}{2}$ and unit randomization (singleton clusters). Further suppose that each $\cN_i = [3]$ with each $c_{i,\cS} = 1$ when $|\cS| = 1$ and $c_{i,\cS} = -1$ when $|\cS| = 2$. Then, $\widehat{\TTE}_{\textrm{HT}} \equiv 0$ so $\var \big(\widehat{\TTE}_{\textrm{HT}} \big) = 0$. However, $\Pr\big(\widehat{\TTE}_2 = 0\big) = \frac{1}{4}$ and $\Pr\big(\widehat{\TTE}_2 = 6\big) = \Pr\big(\widehat{\TTE}_2 = -6\big) = \frac{3}{8}$, so $\var \big(\widehat{\TTE}_2 \big) = 27$.
\end{example}

This example is specifically designed to exploit the particularities of the Horvitz-Thompson estimator. Note that in this example, the observed outcomes of each unit under complete treatment and complete non-treatment are $0$. The Horvitz-Thompson estimator zeros out observations from units whose neighborhoods are partially treated (the only scenario in which non-zero outcomes are observed in this example), whereas the pseudoinverse estimator extracts signal from these observations. In practice (i.e., outside of the artificial construction of the above example), we find that the pseudoinverse estimator has a lower variance than the Horvitz--Thompson estimator since it exploits the $\beta$-order interactions structure of the potential outcomes. The Horvitz--Thompson estimator, in contrast, does not exploit this structure, leading to higher variance (but with the advantage of being unbiased for a larger class of outcome models). 

In particular, when either the potential outcomes under complete treatment or non-treatment are non-zero, the $\Omega(\frac{1}{p^{|\cC(\cN_i)|}})$ scaling in the Horvitz-Thompson variance \citep{ugander2013} becomes unavoidable. A sufficient condition to guarantee this is through a \textit{strong monotonicity assumption}, wherein the signs of $c_{i, S}$ are constant across all $i \in [n]$ and $S \in \cS_i^{\beta^*}$. Under this assumption, the variance of $\widehat{\TTE}_{\beta^*}$ always improves on the variance of the Horvitz--Thompson estimator defined in \eqref{eq:tte_ht} (which assumes a neighborhood exposure model as in Assumption~\ref{ass:neighborhood_interference}), suggesting that we should prefer the pseudoinverse to the Horvitz--Thompson estimator when using a $\GCR(\cC,p)$ design. We formalize this point with the following remark. 

\begin{restatable}{remark}{gcrhtbound} \label{rem:gcr_pi_better}
Fix a clustering $\cC$ and suppose that $\bz \sim \GCR(\cC,p)$. Assume that the sign of $c_{i, S}$ is constant across all $i$ and $S\subseteq \cS_i^{\beta}$. Then for $p \in (0,\frac12)$ and $\beta \geq \beta^*$, we have $\var\Big( \widehat{\TTE}_{\beta} \Big) \leq \var \Big(\widehat{\TTE}_{\textrm{HT}} \Big)$.
\end{restatable}

The proof of this remark, which relies on the explicit form of the Horvitz--Thompson estimator, can be found in Appendix~\ref{sec:proofs_pi_gcr}. 
\meedit{
\section{Practical Considerations}

In this section, we discuss considerations of interest to a practitioner using our approach, detailing the computational cost of the pseudoinverse estimator, factors to consider in the choice of $\beta$, and how to select a clustering for the randomized design. Many of the points that we address here are validated by our experimental results in Section~\ref{sec:experiments}.

\subsection{Computational Efficiency}

At a glance, the pseudoinverse estimator may look computationally costly; however, under Bernoulli GCR designs, we can leverage its symmetries to simplify the computation significantly.
Rather than pseudo-inverting each design matrix to calculate $\widehat{\TTE}_{\beta}$, the estimator can be written as a simple function of the number of neighboring clusters that are in treatment/control. By Lemma \ref{lem:pi_gcr_est_form}, the pseudoinverse estimator is equivalent to the SNIPE estimator from \citep{cortez2023exploiting}, where each cluster is viewed as a unit. In their equation (4.5), \citet{cortez2023exploiting} derive an alternate form of the SNIPE estimator that expresses the coefficient on each outcome $Y_i(\bz)$ as a summation over counts of treated subsets of different sizes. Generalizing this to Bernoulli GCR, let us define $\mathcal{T} = \{ C \in \cC \mid w_C = 1 \}$ to be the clusters assigned to treatment. Then, we may alternatively represent the pseudoinverse estimator as

\[
    \widehat{\TTE}_\beta = \tfrac{1}{n} \sum_{i=1}^n Y_i(\bz) \sum_{\substack{k,\ell \geq 0 \\ k+\ell \leq \beta}} \binom{|\cC(\cN_i)\!\cap\!\cT|}{k} \binom{|\cC(\cN_i)\!\!\setminus\!\! \cT|}{\ell} (-1)^{k+\ell} \left( \left(\tfrac{p-1}{p}\right)^{k} - \left(\tfrac{-p}{1-p}\right)^{\ell} \right).
\]
For each individual, we require $O(C)$ time to calculate $|\cC(\cN_i) \cap \mathcal{T}|$ and $|\cC(\cN_i) \setminus \mathcal{T}|$ and then $O(\beta^2)$ time to iterate over the terms of the inner summation, giving an overall computation time of $O\big(n (C + \beta^2) \big)$. This alternate form of the pseudoinverse estimator arises since the entries of the design matrix (i.e., the joint treatment probabilities) depend only on the \textit{number} of treated and untreated neighboring clusters, rather than the \textit{identities} of these clusters. This allows us to reframe the pseudoinverse calculation as counting the number of each of $O(\beta^2)$ possible distinct entries in this matrix. 

As a point of comparison, the computation of the Horvitz-Thompson estimator from Section~\ref{subsec:ht} requires the evaluation of the joint treatment indicator functions, the numerators of \eqref{eq:tte_ht}, and the joint probabilities in the denominators. For straightforward designs such as (clustered) Bernoulli and completely randomized designs, these can be computed in $O(d_{\max})$ time by iterating over the individual's neighborhood and counting their treated neighbors. Thus, the overall computational efficiency of the Horvitz-Thompson estimator in these cases is $O(n d_{\max})$. Under a cluster randomized design, by first applying a reduction from units to clusters, the cost can be reduced to $O(n C)$, where $C$ is a bound on the number of neighboring clusters $\cC(\cN_i)$. When $\beta \ll |\cC(\cN_i)|$, the increased computation time of $O(n\beta^2)$ over the Horvitz-Thompson estimator is negligible.

For fully general randomized designs with less symmetry, where equal-sized subsets of an individual's neighborhood may have different joint treatment probabilities (such as in a Bernoulli randomized design with non-uniform treatment probabilities), we must revert to computing the full pseudoinverse matrices. Modern numerical linear-algebra solvers have highly optimized routines to approximate pseudoinverse matrices to a desired precision; however, it may still seem costly to invert a matrix for every unit. In reality, most designs we would consider do have significant symmetries that would allow for reduced computation. For example, consider the completely randomized graph cluster randomization. Since the marginal distribution of neighboring cluster treatments is the same for nodes that have the same number of neighboring clusters, we only need to compute the pseudoinverse of $C$ matrices in total, one for each possible size of neighboring cluster sets. As the cost of pseudo-inversion is constant, the overall computational cost is dominated by $O(n \binom{C}{\beta}^2)$, which is the cost of a matrix-vector product for each unit.

A practitioner may also be interested in evaluating randomized designs beyond Bernoulli, for which an analytical expression for $\E\big[ \tbz_i^\beta (\tbz_i^\beta)^\intercal \big]$ may be non-obvious. In this case, a Monte Carlo approximation of the design moment matrix $\E\big[ \tbz_i^\beta (\tbz_i^\beta)^\intercal \big]$ can be used instead. In Section \ref{subsec:monte_carlo}, we show that the pseudoinverse of a Monte Carlo estimate of the design moment matrix converges fairly quickly to the pseudoinverse of the true design moment matrix.

\subsection{Cluster Selection}

Given a real-world graph, a practitioner can use the variance bound in Theorem \ref
{thm:tte_var_bound} to optimize over the choice of clustering for a $\GCR$ design. That involves first generating a selection of different clusterings from any graph clustering algorithm of the practitioner's choice, varying the size and number of the clusters. We can choose the clustering with the lowest variance bound as evaluated from Theorem \ref{thm:tte_var_bound}, where we plug in the Bernoulli $\GCR$ design constructed with each candidate clustering. Since the Bernoulli $\GCR$ design is unbiased, the variance bound also serves as a bound on the mean squared error. As the variance bound does not depend on $Y_i$, it does not require knowledge of the potential outcomes model to evaluate and is thus robust to the outcome model. One can think of the bound as giving a guarantee on the variance over the worst-case bounded potential outcomes model.

When computing these bounds, a practitioner should numerically evaluate the quadratic form $\theta_i^T\E[\tbz_i\tbz_i^\intercal]^{\dagger}\theta_i$ and then evaluate \eqref{eq:tte_var_bound}, rather than using the upper bounds on $\gamma_i$ that result in Corollary~\ref{thm:pi_gcr}. Those bounds were useful to obtain a simple interpretable bound, but can be loose when selecting a clustering. In Section \ref{subsec:choose_cluster}, we provide simulations on a real-world network that illustrate that the shape of the error bounds as a function of the candidate clusterings in fact matches the shape of the true error function, such that using our bounds to select among candidate clusterings indeed results in a near-optimal choice of clustering.

\subsection{Model Selection}

When choosing the degree $\beta$ to use for the estimator, there are a few considerations to keep in mind. First, a practitioner should consider whether there is reason to believe the outcomes model is low degree or at least approximately low-degree. The low degree property could arise from the interference mechanism arising from interactions of small groups of neighbors, or the outcomes model could be approximately linear in the treated neighbors. In fact, the majority of empirical studies of interference in the literature use linear models. 

Our results also show that there is a bias-variance tradeoff that arises in the selection of $\beta$, as is typical in all statistical problems involving model selection. As Corollary \ref{cor:gcr_bias_bound} indicates, for $\GCR$ designs, the bias of using a lower order $\beta$ model to estimate a higher degree $\beta^*$ model will scale with a sum of the coefficients of the higher degree terms (greater than $\beta$) in the outcomes model. Thus, even if the true model is high, but is well approximated by the lower degree coefficients, then the bias will be limited. To use the bias bound to compare bias-variance trade-offs for different choices of $\beta$ will require additional assumptions on how the magnitude of the coefficients decay as a function of the degree, which may be provided by a practitioner.

As Lemma \ref{lem:gamma_bound} and Corollary \ref{thm:pi_gcr} show, for $\GCR$ designs, the contribution to the variance for unit $i$ scales exponentially in $\min(|\cC(\cN_i)|,\beta)$, which can result in an overall dependence that is exponential in $\min(|\cC(\cN_i)|,\beta)$, where $C$ is a bound on the number of neighboring clusters. This indicates that in settings with poor clustering, the variance is very sensitive to the choice of $\beta$, but for settings with good clustering, the choice of $\beta$ has significantly less impact on the variance. This is directly related to the observation discussed in Section \ref{subsec:ht} that the coefficient in the estimator on the outcome $Y_i$ is exactly the same for the pseudoinverse and Horvitz-Thompson estimators when $|\cC(\cN_i)| \leq \beta$. Thus, one can think of the Horvitz-Thompson estimator as equivalent to choosing a large value of $\beta$.

In Section \ref{subsec:gcr_exp}, we provide simulations that compare the MSE between the pseudoinverse and Horvitz-Thompson estimator under the correctly specified setting ($\beta = \beta^*$). The results show that for poor clusterings where $C$ is large, the variance for Horvitz-Thompson can be significantly larger than the variance of the pseudoinverse estimator with $\beta = \beta^*$. This difference in variance is an upper bound on the penalty in variance for choosing a parameter $\beta$ that is much larger than the ground truth degree $\beta^*$.

In Section \ref{sec:expt_misspec}, we provide simulations that compare the impact of misspecified choices of $\beta$, in particular allowing $\beta < \beta^*$, which may incur bias. Our empirical results show that it may be beneficial to always choose a low-degree estimator even when the true model is higher degree, because the impact to the bias of underestimating $\beta^*$ is bounded for bounded outcomes models, whereas the variance cost for overestimating $\beta$ can be unbounded and exponentially large up to the size of $|\cC(\cN_i)|$. As a result, we suggest that practitioners choose small values of $\beta$ in real applications.
}
\section{Experiments} \label{sec:experiments}

Here, we present a set of experiments to supplement our theoretical results. In Section~\ref{subsec:gcr_exp}, we consider a stylized setting in which the underlying interference graph is a cycle graph and show that the interplay between the low-order assumption and the graph clustering identified in the variance bounds of Section~\ref{sec:variance} correctly captures the variance of $\widehat{\TTE}_{\beta}$ in simulation, focusing on the setting where the parameter of the estimator $\beta$ is equal to the true degree of the potential outcomes model $\beta^*$. In Section~\ref{sec:expt_misspec}, we compare against misspecified settings where $\beta$ may be different from $\beta^*$, showing the bias-variance tradeoff that arises. In Section~\ref{subsec:choose_cluster}, we move to a more realistic setting in which we study interference on both stochastic block model graphs and real-world graphs to demonstrate the practical value of variance bounds when choosing the clustering in a graph-cluster randomized experiment. Finally, in Section~\ref{subsec:monte_carlo}, we show that the pseudoinverses appearing in our bounds can be accurately approximated with numerical methods, making our bounds useful even in settings where the pseudoinverse cannot be computed in closed form.
Through these experiments, we find that our bounds correctly predict the behavior of $\widehat{\TTE}_{\beta}$ in simulation and that they can be used to select a clustering with favorable results. 

\subsection{Interplay Between Low-Order Modeling and Cluster Randomization}
\label{subsec:gcr_exp}

In this section, we empirically verify our results showing that $\widehat{\TTE}_{\beta}$ improves on the Horvitz-Thompson estimator. We consider an interference graph $G$ that is a power of the cycle graph on $n=840$ nodes, in which each vertex has degree $d=7$ because it is connected to itself and the three vertices closest to it on each side; see Figure~\ref{fig:ring_clusters} for an illustration of this network.

\begin{figure}
    \centering
    \begin{minipage}{0.3\textwidth}
        \begin{tikzpicture}
            \node[circle,minimum size=120pt] (c) {};
            \node[circle,minimum size=108pt] (d) {};
    
            \foreach \t in {-60, -40, ..., 240} {
                \node[draw, circle, inner sep=0pt, minimum size=12pt] (a\t) at (c.\t) {}; 
                \draw (a\t) edge[loop, out={\t-15}, in={\t+15}, min distance=16pt,-latex, thick, red] (a\t); 
            }
            \foreach \t in {-55, -35, ..., 225} {
                \tikzmath{\s = \t + 10;}
                \draw (c.\s) edge[latex-latex, thick, red] (c.\t);
            }
    
            \foreach \t in {-60, -40, ..., 180} {
                \tikzmath{\s = \t + 60;}
                \draw (d.\t) edge[out={\t+140},in={\s-140},latex-latex, thick, red!40] (d.\s);
            }
            \foreach \t in {-60, -40, ..., 200} {
                \tikzmath{\s = \t + 40;}
                \draw (d.\t) edge[out={\t+120},in={\s-120},latex-latex, thick, red] (d.\s);
            }
    
            \node at (c.270) {\Large $\hdots$};
        \end{tikzpicture}
    \end{minipage}
    \hfill
    \begin{minipage}{0.27\textwidth}
        \begin{tikzpicture}[scale=0.25]
            \node[circle,minimum size=100pt] (c) {};
            \node[circle,minimum size=85pt] (i) {};
            \node[circle,minimum size=120pt] (o) {};
            
            \foreach \t in {-70, -30, ..., 210} {
                \tikzmath{\s = \t + 40;}
                \tikzmath{\u = \t + 2;}
                \tikzmath{\v = \s - 2;}
                \fill[red!20, rounded corners=10pt] (i.\t) -- (o.\u) -- (o.\v) -- (i.\s) -- cycle; 
            }
        
            \foreach \t in {-60, -40, ..., 240} {
                \node[draw, circle, inner sep=0pt, minimum size=10pt] (a\t) at (c.\t) {}; 
            }
    
            \node at (c.270) {\Large $\hdots$};
            \node at (0,0) {$w=2$};
        \end{tikzpicture} 
    \end{minipage}
    \begin{minipage}{0.27\textwidth}
        \begin{tikzpicture}[scale=0.25]
            \node[circle,minimum size=100pt] (c) {};
            \node[circle,minimum size=85pt] (i) {};
            \node[circle,minimum size=120pt] (o) {};
            
            \foreach \t in {-70, 10, ..., 210} {
                \tikzmath{\s = \t + 80;}
                \tikzmath{\u = \t + 2;}
                \tikzmath{\v = \s - 2;}
                \tikzmath{\a = \t + 19;}
                \tikzmath{\b = \t + 38;}
                \tikzmath{\c = \t + 57;}
                \fill[red!20, rounded corners=10pt] (i.\t) -- (o.\u) -- (o.\a) -- (o.\b) -- (o.\c) -- (o.\v) -- (i.\s) -- (i.\c) -- (i.\b) -- (i.\a) -- cycle; 
            }
        
            \foreach \t in {-60, -40, ..., 240} {
                \node[draw, circle, inner sep=0pt, minimum size=10pt] (a\t) at (c.\t) {}; 
            }
    
            \node at (c.270) {\Large $\hdots$};
            \node at (0,0) {$w=4$};
        \end{tikzpicture} 
    \end{minipage}
    
    \caption{The graph $G$ (left) that we consider in our experiments showing that $\widehat{\TTE}_{\beta}$ has lower variance than the Horvitz--Thompson estimator. $G$ is the third power of a cycle on $n=840$ vertices, so each vertex has degree 7 (3 neighbors on each side and a self-loop). Two clusterings of $G$ that we consider are shown on the right. The value of $w$ controls the average cluster-degree $|\cC(\cN_i)|$. When $w=1$, we have that $|\cC(\cN_i)|=d_i=7$ for all $i$, while for $w=7$, the nodes in the center of each cluster have $|\cC(\cN_i)|=1$.}
    \label{fig:ring_clusters}
\end{figure}

We generate responses according to a $\beta^*$-order potential outcome model for $\beta^* \in \{1,2,3,4\}$. For each choice of $\beta^*$, we set the coefficients $c_{i, S}$ so that 
\begin{equation}
    c_{i, S}=\binom{d_i}{|S|}^{-1}2^{-|S|}.
    \label{eq:ring_responses}
\end{equation}
With this choice of $c_{i, S}$ the baseline effects are $c_{i, \emptyset}=1$, and each non-empty subset $S$ of size 1 has $c_{i, S}=1/(2d_i)$, so all subsets of size 1 together contribute $1/2$ to the treatment effect. Similarly, all subsets of size 2 together contribute $1/4$ to the treatment effect, and so on for larger subsets. This response model is thus consistent with the intuition that effects from higher-order subsets should be smaller than effects from lower-order subsets.

We consider a family of clusterings in which contiguous subsets of $w$ adjacent nodes in the cycle are assigned to the same cluster for $w \in \{1,2,3,4,5,6, 7, 8\}$ (here $w$ is the ``width'' of the cluster). Note that such clusterings only exist when $w$ is a divisor of the sample size, hence why we use a sample size of $n=\text{lcm}(1, 2, 3, 4, 5, 6, 7, 8)=840$ in this experiment. The choice of $w$ controls the cluster degree, $|\cC(\cN_i)|$ of each node\textemdash when $w=1$, we have $|\cC(\cN_i)|=7$ for all $i$, since all 7 neighbors are in different clusters. As $w$ increases, the graph is divided into fewer, larger clusters. For larger $w$, there are then more nodes with smaller values of $|\cC(\cN_i)|$, since nodes in the interior of the cluster will be connected to fewer other clusters. For example, when $w=7$, there are nodes in the center of each cluster with $|\cC(\cN_i)|=1$. See Figure~\ref{fig:ring_clusters} for an illustration of these clusterings with $w=2$ and $w=4$.

For each choice of $w$ and $\beta^*$ we estimate the total treatment effect using both the Horvitz--Thompson estimator and pseudoinverse estimators with $\beta=\beta^*$ with data from an experiment with a $\GCR(\cC, 0.25)$ design. For each estimator and each choice of parameters, we calculate the experimental mean-squared error (MSE) over 1000 trials. The results are plotted in Figure~\ref{fig:ring_results}.

\begin{figure}
    \centering
    \includegraphics[width = \textwidth]{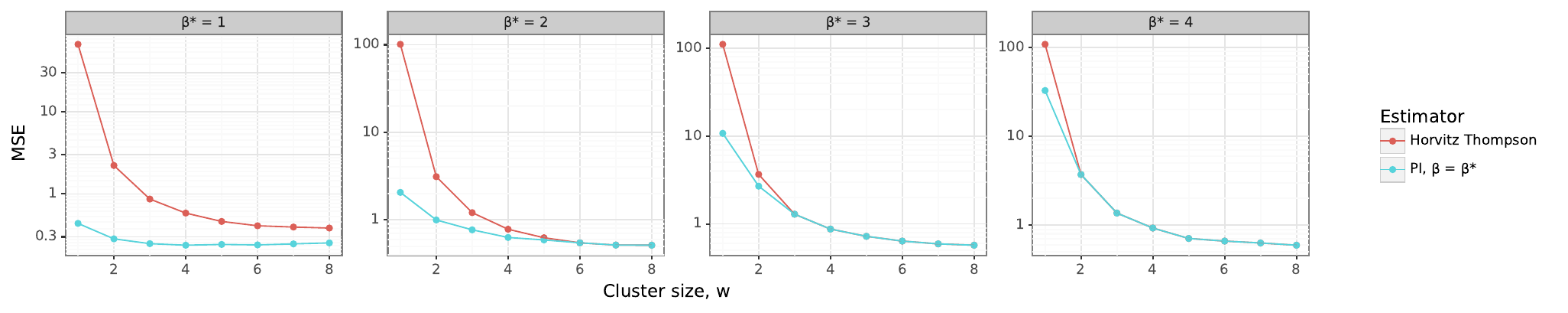}
    \caption{MSE of the Horvitz--Thompson and well-specified pseudoinverse estimators for different values of $\beta^*$ and $w$ averaged over 1000 trials. Note that the number of clusters is $m=n/w$. We see across both designs that (a) the pseudoinverse estimator consistently improves on the Horvitz--Thompson estimator; (b) this improvement is largest for small values of $w$ (which are poor clusterings of the graph); (c) the variance of the pseudoinverse is less sensitive to the quality of the clustering for small values of $\beta$.}
    \label{fig:ring_results}
\end{figure}

Figure~\ref{fig:ring_results} corroborates both of our theoretical claims about the variance of $\widehat{\TTE}_{\beta}$ under $\GCR$ designs. Across all values of $w$ and $\beta$, the variance of $\widehat{\TTE}_{\beta}$ is less than the variance of the Horvitz--Thompson estimator, as shown in Remark~\ref{rem:gcr_pi_better}. The gap between the two is largest when $\beta=1$, in which case the low-order assumption is quite strong and offers significant variance reduction, and decreases as $\beta$ increases (as the low-order assumption becomes weaker). Indeed, if we had $\beta=7$ (the degree of each node in this graph), $\widehat{\TTE}_{\beta}$ under the $w=1$ design (which is unit-randomization) would coincide exactly with the Horvitz--Thompson estimator, as shown by \citet{cortez2023exploiting}.

Figure~\ref{fig:ring_results} also supports the insight from Corollary~\ref{thm:pi_gcr} that the variance of $\widehat{\TTE}_{\beta}$ depends on the minimum of the degree $\beta$ and the cluster neighborhood sizes $|\cC(\cN_i)|$. In the left-most panel, where $\beta=1$, the variance of $\widehat{\TTE}_{\beta}$ is nearly constant in $w$. This arises since the $\beta=1$ assumption is always stronger than any choice of clustering, and so the variance does not vary significantly with the clustering. As $\beta$ increases, we see more dependence on the cluster size $w$, since for larger values of $w$ there are more units with smaller values of $|\cC(\cN_i)|$, and in particular more units with values of $|\cC(\cN_i)|$ smaller than $\beta$. As shown in Corollary~\ref{thm:pi_gcr}, these are the units for which the $\GCR$ design provides variance reduction beyond what the low-order assumption alone would.

\cyedit{As discussed in Section \ref{subsec:ht}, for units where $|\cC(\cN_i)| \leq \beta$, their contribution to the estimator is exactly the same between the Horvitz-Thompson estimator and $\widehat{\TTE}_{\beta}$. For cluster width 2, the maximum number of adjacent clusters is 4; thus, when $\beta^* = 4$, the two estimators are exactly the same. Similarly, for cluster widths 3 - 5, the maximum number of adjacent clusters is 3, such that for $\beta^* \geq 3$ the two estimators are the same. For cluster widths 6 - 8, the maximum number of adjacent clusters is 2, such that for $\beta^* \geq 2$ the two estimators are the same. This illustrates the general concept that when there is an obvious clustering that minimizes cross-cluster edges, then the simple Horvitz-Thompson estimator performs well. However, real-world networks are rarely so cleanly clusterable, and when there is not an obvious cluster, using the pseudoinverse estimator with a modest $\beta$ may have significant gains. Even under the best clustering, the cluster-degrees can still be quite high to the point where the Horvitz-Thompson estimator is simply not feasible due its instability.}

\subsection{The Effect of Misspecification}\label{sec:expt_misspec}

In this section, we repeat the experiment of Section~\ref{subsec:gcr_exp} with misspecification. For each choice of $w$ and $\beta^*\in\{1,4\}$, we estimate the total treatment effect using both the pseudoinverse estimator with $\beta=1$ and $\beta=4$ with data from experiments using a $\GCR(\cC, 0.25)$ design. For each estimator and each choice of parameters, we calculate the experimental mean-squared error (MSE), bias, and variance over 1000 trials. The results are plotted in Figure~\ref{fig:ring_results_misspecified}.

\begin{figure}
    \centering
    \includegraphics[width = \textwidth]{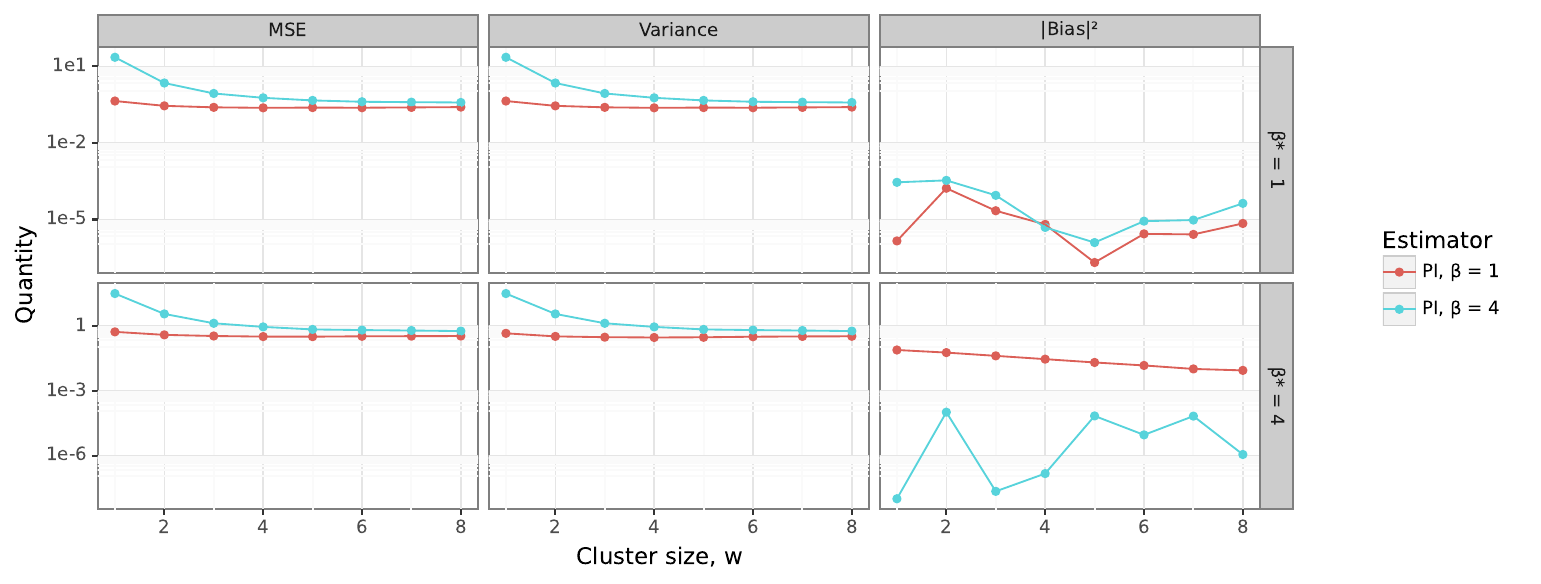}
    \caption{Bias, variance, and MSE \juedit{(all in log scale)} of the pseudoinverse estimator with $\beta=1$ and $\beta=4$ for different values of $w$ and $\beta^*\in\{1, 4\}$ averaged over 1000 trials. Note that the number of clusters is $m=n/w$. We see that the $\beta=1$ estimator has lower MSE in both settings: this is because the bias it incurs when $\beta^*=4$ is negligible compared to the additional variance incurred by the $\beta=4$ estimator.}
    \label{fig:ring_results_misspecified}
\end{figure}

We first consider the case where we have $\beta^*=1$, so that the estimator with $\beta=1$ is well-specified and the estimator with $\beta=4$ overestimates $\beta$. In this case, we see that both estimators are unbiased, as expected, and that the $\beta=4$ estimator has higher variance, also as expected. Thus, we see that the $\beta=1$ estimator has lower MSE.

Next, we consider the case where $\beta^*=4$, so that the estimator with $\beta=1$ underestimates $\beta$ and the estimator with $\beta=4$ is well-specified. In this case, we see that the estimator with $\beta=4$ is unbiased, since it is well-specified, but that the estimator with $\beta=1$ has a small amount of bias. This is the bias incurred from not estimating the higher-order coefficients, but it is small because these coefficients are small relative to the lower-order coefficients by virtue of the specification in \eqref{eq:ring_responses}. However, the variance of the $\beta=4$ estimator is much larger than the variance of the $\beta=1$ estimator, and the scale of the variance in this setting is much larger than the scale of the variance. Thus, the bias from misspecification contributes negligibly to the MSE, and the $\beta=1$ estimator again has lower MSE. 

The fact that the $\beta=1$ estimator is preferable both when $\beta^*=1$ and $\beta^*=4$ suggests that the bias-cost of underestimating $\beta^*$ is much less than the variance-cost of overestimating $\beta$ when higher-order coefficients are small. This is consistent with the predictions made by Corollary~\ref{cor:gcr_bias_bound}, and what we expect from real-world data as well, leading us to recommend that practitioners use small values of $\beta$ such as $\beta=1$ or $\beta=2$ in applications.

\subsection{Choosing From a Set of Clusters}
\label{subsec:choose_cluster}

Next, we demonstrate how our bounds from Section~\ref{sec:variance} are useful in practice to select cluster assignments. We consider a setting where an analyst has a fixed interference graph and a set of possible candidate clusterings and would like to choose between them without any \textit{a priori} knowledge of the response model. We now describe each aspect of our experiment in greater detail.

\paragraph{Graphs.}

We consider four graphs, two synthetic and two taken from real-world data. The two synthetic graphs are drawn from a stochastic block model with $n=500$ units divided into $20$ communities. The edge probability within each community is $\pi_{ii}=0.5$ and the edge probability between communities is $\pi_{ij}=0$ for the first synthetic graph and $\pi_{ij}=0.2$ for the second synthetic graph. These two graphs represent a sparse graph of disconnected clusters (with an obvious choice of clustering), and a denser connected graph with a less obvious choice of clustering.

The first real-world graph we consider is taken from \citet{cai2015social}. Following \citet{viviano2023causal}, who also use this graph in their experiments, we connected two individuals by an edge if at least one of them indicates the other as a friend. This construction can connect individuals in a given village to other individuals outside of their village, since people may indicate a person outside of their village as a friend. We take the largest connected component of this graph, which contains $n=3677$ individuals with average degree $\bar{d}=6.25$, suggesting that it is relatively sparse.

The second real-world graph we consider is the Caltech Facebook friendship graph from \citet{traud2012social}, also previously used to study graph cluster randomized designs by \citet{ugander2023randomized}. We take the largest connected component of this graph, which removes a small number of isolated individuals. This leaves $n=762$ individuals with average degree $\bar{d}=44.7$, suggesting that this graph is relatively dense.

\paragraph{Response models.}

We generate responses according to an order $\beta^*=1$ model and use pseudoinverse estimators with $\beta=\beta^*=1$. We consider three response models with varying levels of interference and treatment effect sizes. To specify an order $\beta^*=1$ model requires specifying the baseline effect, $c_{i, \emptyset}$, and also each of the coefficients $c_{i, \{j\}}$ for $j\in \cN_i$. In all of our response models, we sample $c_{i, \emptyset}\sim N(0.5, 0.1)\cdot d_i/\dmax$ to introduce correlation between baseline effects and graph degree, a common feature in real-world data \citep{basse2018model}. The models differ in how we set $c_{i,\{j\}}$ for $j\in \cN_i$:

\begin{itemize}
	\item \emph{Null model, $\TTE_i=0$, for all $i$}: in this model, we set $c_{i, \{j\}}=0$ for all $j\in \cN_i$, so the treatment has no effect and there is no interference. The implied form of $Y_i(\bz)$ is $Y_i(\bz)=c_{i,\emptyset}$ for all $i$, $\bz$.
	\item \emph{Weak treatment effect model, $\TTE_i=1$, for all $i$}: in this model, we set $c_{i, \{j\}}=\frac{1}{2(d_i-1)}$ if $j\neq i$ and $c_{i, \{i\}}=1/2$. Thus, the total treatment effect for each individual, $c_i^\intercal\theta_i$, is 1, with half the effect coming when the individual is treated, and the other half from interference. The implied form of $Y_i(\bz)$ is 
 \begin{equation*}
 Y_i(\bz) = c_{i, \emptyset} + \frac{1}{2} z_i+\sum_{\substack{j\in \cN_i, j\neq i}} \frac{1}{2(d_i-1)}z_j.
 \end{equation*}
	\item \emph{Strong treatment effect model, $\TTE_i\propto d_i$}: in this model, we set $c_{i, \{j\}}=1/2$ if $j\neq i$, and $c_{i, \{i\}}=d_i/2$. Thus, the total treatment effect for each individual, $c_i^\intercal\theta_i$, is of order $d_i$, with roughly half the effect coming when the individual is treated, and the other half from interference. The implied form of $Y_i(\bz)$ is 
 \begin{equation*}
 Y_i(\bz) = c_{i, \emptyset} + \frac{d_i}{2}z_i+\sum_{\substack{j\in \cN_i, j\neq i}} \frac{1}{2}z_j.
 \end{equation*}
 
\end{itemize}

\paragraph{Clusterings.}

We generate the clusterings in our experiment using two different algorithms.
The first algorithm we consider is the Louvain algorithm \citep{blondel2008fast}, which has previously been shown to have good performance when designing graph cluster randomized experiments \citep{karrer2021network}. We vary the resolution parameter of the algorithm across the exponential grid $\{2^{-5}, 2^{-4},\cdots, 2^4, 2^5\}$ to generate 11 clusterings. Clusterings generated with small values of the resolution parameter tend to have one large cluster containing a majority of the nodes, and then several small clusters, while clusterings generated with larger values of the resolution parameter tend to have many small clusters of roughly equal size. 

The second algorithm we consider is the METIS algorithm, a balanced partitioning algorithm \citep{karypis1998fast}. This means that, unlike the Louvain algorithm, which is meant to detect communities, METIS returns $m$ clusters of approximately equal size for a user-specified value of $m$. We generate 10 METIS clusterings by varying $m$ across the grid $\{5, 10, 15,\cdots, 45, 50\}$.

Our goal in this experiment is to select from the set of 11 clusterings produced by the Louvain algorithm, as well as from the set of 10 clusterings produced by METIS. Note that we consider these two tasks separately, i.e., we consider first the problem of picking from the set of Louvain clusterings and second the problem of picking from the set of METIS clusterings. Our approach can easily be used to pick from the combined set as well, but we separate the two out to highlight that our approach works well within multiple families of clusterings.

We compute the variance bounds of Section~\ref{sec:variance} for a $\GCR$ design for each clustering, which also bounds the mean-squared-error because of our unbiasedness results, and then select the clustering with the lowest bound. When computing these bounds, we numerically evaluate the quadratic form $\theta_i^T\E[\tbz_i\tbz_i^\intercal]^{\dagger}\theta_i$ and then evaluate \eqref{eq:tte_var_bound}, rather than using the upper bounds on $\gamma_i$ that result in Corollary~\ref{thm:pi_gcr}. Those bounds were useful to obtain a simple interpretable bound, but can be loose when selecting a clustering. We refer to the chosen clustering as $\hat{\cC}^*$, since it approximates an oracle clustering $\cC^*$ defined below. Note that $\hat{\cC}^*$ does not depend on the response model. 

For each response model and clustering, we simulate $1000$ replications of a $\GCR(\cC, 0.25)$ design for each clustering $\cC$ to estimate the true root-mean-squared-error (RMSE), and identify the clustering with the true lowest RMSE, which we refer to as $\cC^*$. We benchmark against $\cC^*$ since it represents the best that any clustering-selection procedure could hope to do, and thus is a stronger benchmark than any specific other procedure. Note that $\cC^*$ may vary from one response model to another since different designs will perform well under different response models. 
We evaluate our selection by computing 
\begin{equation}
	\frac{\text{Monte Carlo RMSE of}\GCR(\hat{\cC^*}, 0.25)}{\text{Monte Carlo RMSE of}\GCR({\cC}^*, 0.25)},
	\label{eq:rmse_ratio}
\end{equation}
for each response model. This ratio will be close to 1 if we have selected a nearly optimal clustering for that response model and small if we have not. 

\begin{table}[t]
	\centering

	\begin{tabular}{lllll}
		\toprule
        &&\multicolumn{3}{c}{Response Model}\\
        \cmidrule{3-5}
		Clustering&Graph&{$\TTE_i=0$}&{$\TTE_i=1$}&{$\TTE\propto d_i$}\\
		\midrule
		&$\text{SBM}(0.5, 0)$&1.000&1.000&1.000\\
		&$\text{SBM}(0.5, 0.2)$&1.000&1.000&1.000\\
		&Cai et al.&1.000&0.780&0.839\\
\multirow{-4}{*}{Louvain}&Caltech&1.000&0.888&0.928\\[.1in]
&$\text{SBM}(0.5, 0)$&1.000&1.000&1.000\\ 
&$\text{SBM}(0.5, 0.2)$&0.844&0.837&0.824\\ 
&Cai et al.&1.000&1.000&1.000\\ 
\multirow{-4}{*}{METIS}&Caltech&1.000&1.000&1.000\\ 
\bottomrule
	\end{tabular}

	\caption{Ratio of RMSE when using clustering selected by minimizing an upper bound on the RMSE to RMSE when using oracle optimal clustering for each response model, graph, randomization strategy. We see that the chosen clustering is consistently very close to optimal.}
	\label{tab:cluster_selection}
\end{table}

\paragraph{Results.}
The ratios in \eqref{eq:rmse_ratio} for each of the settings are given in Table~\ref{tab:cluster_selection}. 
For the Louvain clusterings, we see that for the two stochastic block models, our method consistently selects optimal clusterings across response models. For the more realistic Cai et al.\ and Caletch graphs, our method continues to select nearly optimal clusterings. This slight decrease in performance is a consequence of different clusterings being favorable for different response models for the two real-world graphs\textemdash our chosen design is optimal for the $\TTE_i=0$ response model, but not the others. This trade-off is fundamental to the problem of experimental design and is not a drawback of our particular approach. 

We also see similarly strong results for the METIS clusterings, with our method selecting optimal clusterings across the board for all graphs except for the $\text{SBM}(0.5, 0.2)$, where we select slightly sub-optimal clusterings. As mentioned above, this reflects a trade-off inherent to the experimental design problem\textemdash our bounds optimize the worst-case RMSE over a class of response models, but this does not mean that the clusterings for which our bounds are small are optimal for all response models.

We further visualize the results of this experiment in Figure~\ref{fig:sbm_var_bounds}, which shows the empirical RMSE and theoretical RMSE bounds for two response models in each of the two stochastic block models as a function of the resolution parameter. 
In the $\text{SBM}(0.5, 0)$ graph, the Louvain algorithm identifies the correct clustering into communities for a wide range of resolution parameters, and our bounds correctly identify that this is the optimal clustering for both response models. In the $\text{SBM}(0.5, 0.2)$ graph, the optimal clustering for the $\GCR$ design under both response models is obtained for a small value of the resolution parameter, as our bound also correctly identifies. 

Taken together, Table~\ref{tab:cluster_selection} and Figure~\ref{fig:sbm_var_bounds} show that our bounds correctly capture the dependence of the RMSE when using a $\GCR(\cC, p)$ on the clustering $\cC$, and that using our bounds to select a clustering gives favorable results in practice.

\begin{figure}[t]
	\centering
	\includegraphics[width=\textwidth]{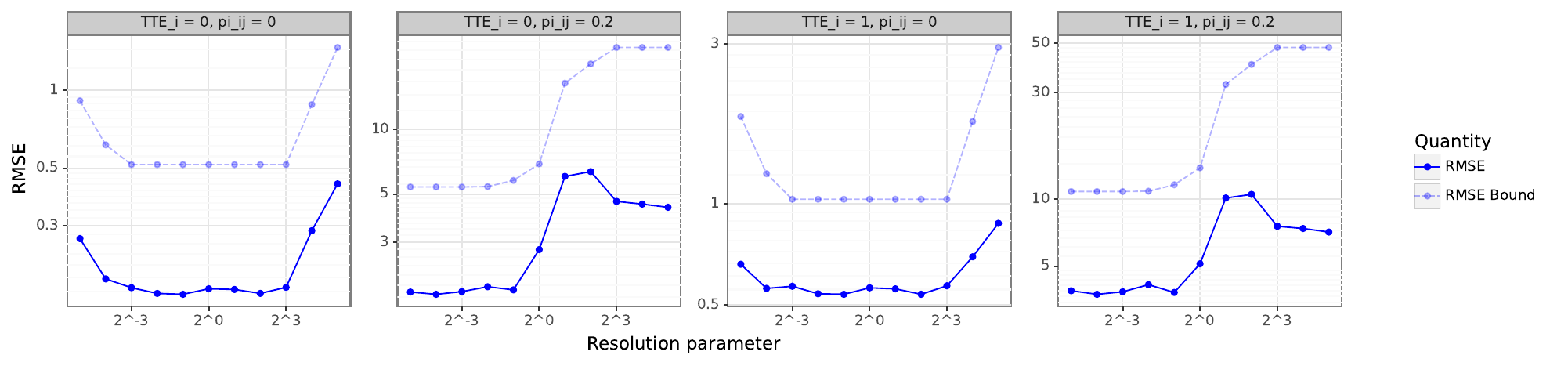}
	\caption{Visualization of our theoretical RMSE bounds (dashed) and actual RMSE in simulation (solid) across a range of Louvain algorithm resolution parameters for the $\text{SBM}(0.5, 0)$ and $\text{SBM}(0.5, 0.2)$ graphs under a $\GCR$ design and the $\TTE_i=1$ and $\TTE_i=0$ response models. We see that the qualitative shape of our bounds correctly reflects the dependence of the RMSE on the resolution parameter.}
	\label{fig:sbm_var_bounds}
\end{figure}

\subsection{Monte Carlo Estimation of the Design Matrix} 
\label{subsec:monte_carlo}

\begin{figure}
\centering
    \includegraphics[width=0.4\linewidth]{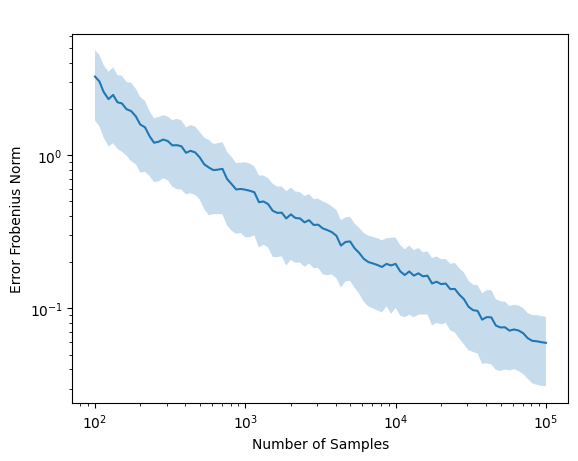} 
\caption{Monte Carlo error when estimating design matrices for $\GCR$ designs. We see that our estimates converge to the desired quantity in large sample sizes.}
\label{fig:montecarlo}
\end{figure}

In our final set of experiments, we verify that the pseudoinverses needed to compute our bounds can be numerically approximated given sampling access to the experimental design and do not need to be calculated analytically. This broadens the practical use cases of our result to include complex designs that may not admit simple analytical bounds.

Throughout our work, we leverage analytic expressions for the design matrices to carry out our theoretical analysis. Here, we demonstrate that these closed-form expressions are not necessary to use pseudoinverse estimation in practice. Rather, it is sufficient to have sampling access to the randomized designs, which we can leverage to construct Monte Carlo estimates of the necessary joint treatment probabilities. This insight will be crucial for analyzing more sophisticated randomized designs, such as the randomized GCR designs of \citet{ugander2023randomized}, where it is unclear whether the design matrix entries, let alone the entries of the pseudoinverse, have a tractable closed form.

Suppose that $\mathcal{D}$ denotes the joint distribution over the treatment assignment vectors for a randomized design of interest. Then, given $R$ independent samples $\bz^1, \hdots, \bz^R \overset{iid}{\sim} \mathcal{D}$, and given any $S \subseteq [n]$, we may estimate 
\[
    \Pr_{\mathcal{D}} \Big( S \textrm{ fully treated} \Big) \approx \frac{1}{R} \sum_{r=1}^{R} \prod_{j \in s} z^r_j.
\]
The Law of Large Numbers ensures the convergence of these estimates as $R \to \infty$. By carrying out this process for each entry of each design matrix $\E \big[ \tbz_i \tbz_i^\intercal \big]$, we obtain an estimate matrix, which we may pseudoinvert numerically to obtain an approximation to our desired matrix.

Figure~\ref{fig:montecarlo} demonstrates the convergence of this Monte Carlo estimation for $\GCR$ designs on the Caltech Facebook friendship graph from \citet{traud2012social}. We use the METIS clustering with $m=10$ clusters, and we set the marginal treatment probability to $p=0.2$. The plot visualizes the Frobenius norm of the error between the true design matrix pseudoinverse and the pseudoinverse calculated by the Monte Carlo procedure described above. We vary the number of samples $R$ used for the Monte Carlo calculations. The blue line depicts the average Frobenius norm across the 762 vertices, and the shading depicts the standard deviation. The axes use a logarithmic scale.

The plot illustrates the decrease in approximation error as the number of samples used in the Monte Carlo probability calculations increases. With sufficient samples, it is clear that these probabilities can be computed to an arbitrary precision.

\section{Conclusion}

We have presented an estimator for the total treatment effect ($\TTE$) in a low-order model, given bounds on its bias and variance under general experimental designs, and shown how those bounds can be used to obtain new theoretical results and inform the design of clustered experiments.

One possible future direction that builds directly on our results is to better understand the dependence of $\gamma_i$ in Theorem~\ref{thm:tte_var_bound} on the experimental design. The value of the quadratic form $\theta_i\E[\tbz_i\tbz_i^\intercal]^\dagger\theta_i$ will depend on the spectral properties of $\E[\tbz_i\tbz_i^\intercal]$, and further analyzing this dependence would shed further light on the variance properties of $\widehat{\TTE}_{\beta}$ under different designs. 

More generally, there are many natural extensions of the low-order model itself. For example, one could encode further constraints on the $\bc_i$, such as the fact that direct effects are typically larger than indirect effects or that higher-order effects are typically smaller than lower order effects, into the estimation procedure. Taking a step further back, one could also constrain the potential outcomes $Y_i(\bz)$ to lie in a different subspace by constraining the vectors $\bc_i$ to be ``low-order" in some other basis. We view the current work as a preliminary step in understanding the value of thinking of $Y_i(\bz)$ as a structured polynomial in $\bz$ as well as the interaction of this thinking with other approaches to causal inference under interference such as clustered designs, and are eager to see further development of this perspective. 

\paragraph{Acknowledgements.}This work was supported in part by NSF grant \#CNS-1955997, AFOSR grant \#FA9550-23-1-0301, and by NSF CAREER Award \#2143176. 

\bibliographystyle{plainnat}
\bibliography{refs}

\appendix
\section{Omitted Proofs}
\label{sec:proofs_pi_gcr}

\subsection*{Proof of Corollary~\ref{cor:gcr_bias_bound}}
\textit{Proof:}
    Recall that $\bx_i^{\beta^*} = \big( D_i^{\beta^*} \big)^\intercal \bc_i^{\beta^*}$, such that each entry of $\bx_i^{\beta^*}$ is the sum of entries of $\bc_i^{\beta^*}$. Therefore, $\|\bx_i^{>\beta}\|_1 \leq \|\bc_i^{ >\beta}\|_1$, and we can restrict our attention to the first inequality of \eqref{eq:gcr_bias_bound}. By the same argument as Corollary \ref{cor:gcr_unbiased}, under a Bernoulli graph cluster randomized design, it follows that the first term of Theorem \ref{thm:bias_bound_general} is zero because $\theta_i^\beta$ lies in the column space of $\E\big[\tbz_i^\beta \big(\tbz_i^\beta\big)^{\intercal}\big]$. Therefore, it suffices to bound 
\[
    \Big\langle \bc_i^{>\beta}, \E\big[\tbz_i^{>\beta} \big(\tbz_i^\beta\big)^{\intercal}\big]  \E\big[\tbz_i^\beta \big(\tbz_i^\beta\big)^{\intercal}\big]^\dagger \theta_i^\beta - \mathbf{1} \Big\rangle
\]
for each unit $i$. By substituting \eqref{eq:inverse_design_GCR}, $\theta_i^\beta = D_i^\beta \psi_i^\beta$, $\tbz_i^\beta = D_i^\beta \tbw_i^\beta$, $\tbz_i^{>\beta} = D_i^{>\beta} \tbw_i^{>\beta}$, and $\bx_i^{>\beta} = \big(D_i^{>\beta}\big)^\intercal \bc_i^{>\beta}$, we can recast this expression in terms of $\tbw_i$, $\bx_i$, and $\psi_i$:
\begin{align}
    \notag &= \Big\langle \bc_i^{>\beta}, \E\big[ D_i^{>\beta} \tbw_i^{>\beta} \big(D_i^\beta \tbw_i^\beta\big)^{\intercal}\big] D_i^\beta \Big( (D_i^\beta)^\intercal D_i^\beta \Big)^{-1} \E \big[ \tbw_i^\beta \big(\tbw_i^\beta\big)^\intercal \big]^\dagger \Big( (D_i^\beta)^\intercal D_i^\beta \Big)^{-1} (D_i^\beta)^\intercal D_i^\beta \psi_i^\beta - \mathbf{1} \Big\rangle \\
    \notag &= \Big\langle \bc_i^{>\beta}, D_i^{>\beta} \E\big[ \tbw_i^{>\beta} \big(\tbw_i^\beta\big)^{\intercal}\big] \E \big[ \tbw_i^\beta \big(\tbw_i^\beta\big)^\intercal \big]^\dagger \psi_i^\beta - \mathbf{1} \Big\rangle \\
    \notag &= \Big\langle \big(D_i^{>\beta}\big)^\intercal \bc_i^{>\beta}, \E\big[ \tbw_i^{>\beta} \big(\tbw_i^\beta\big)^{\intercal}\big] \E \big[ \tbw_i^\beta \big(\tbw_i^\beta\big)^\intercal \big]^\dagger \psi_i^\beta - \mathbf{1} \Big\rangle \\
    &= \Big\langle \bx_i^{>\beta}, \E\big[ \tbw_i^{>\beta} \big(\tbw_i^\beta\big)^{\intercal}\big] \E \big[ \tbw_i^\beta \big(\tbw_i^\beta\big)^\intercal \big]^\dagger \psi_i^\beta - \mathbf{1} \Big\rangle. \label{eq:w_inner_prod}
\end{align}

To analyze \eqref{eq:w_inner_prod}, we will begin by understanding the matrix product $\E\big[ \tbw_i^{>\beta} \big(\tbw_i^\beta\big)^{\intercal}\big] \E \big[ \tbw_i^\beta \big(\tbw_i^\beta\big)^\intercal \big]^\dagger$ in the second argument of the inner product. The entry corresponding to the pair of subsets $(\cU, \cV)$ in the matrix $\E\big[ \tbw_i^{>\beta} \big(\tbw_i^\beta\big)^{\intercal}\big]$ is $p^{|\cU \cup \cV|}$ under the $\GCR(\cC, p)$ design. The exact expression for the entry corresponding to the pair of subsets $(\cV, \cW)$ in the matrix $\E\big[ \tbw_i^\beta \big(\tbw_i^\beta\big)^\intercal \big]^\dagger$ is given by Lemma 1 of \citet{cortez2023exploiting}, as $\tbw_i$ can be analyzed as a Bernoulli unit randomized design where clusters are considered as units:
\begin{equation} \label{eq:pseudo_entry}
    \Big[ \E\big[ \tbw_i^\beta \big(\tbw_i^\beta\big)^\intercal \big]^\dagger \Big]_{\cV, \cW} 
    = \big( \tfrac{-1}{p} \big)^{|\cV| + |\cW|} \hspace{-8pt}\sum_{\substack{\cX \in \cC_i^\beta \\ (\cV \cup \cW) \subseteq \cX}} \hspace{-8pt} \big( \tfrac{p}{1-p} \big)^{|\cX|}
\end{equation}
We can use these expressions to simplify the matrix product entry $\Big[ \E\big[ \tbw_i^{>\beta} \big(\tbw_i^\beta\big)^{\intercal}\big] \E \big[ \tbw_i^\beta \big(\tbw_i^\beta\big)^\intercal \big]^\dagger \Big]_{\cU,\cW}$ as follows:
\begin{align}
    \notag &\phantom{=} \sum_{\cV \in \cC_i^\beta} p^{|\cU \cup \cV|} \big( \tfrac{-1}{p} \big)^{|\cV| + |\cW|} \hspace{-8pt}\sum_{\substack{\cX \in \cC_i^\beta \\ (\cV \cup \cW) \subseteq \cX}} \hspace{-8pt} \big( \tfrac{p}{1-p} \big)^{|\cX|} \\
    \notag &= \big( \tfrac{-1}{p} \big)^{|\cW|} \cdot p^{|\cU|} \sum_{\substack{\cX \in \cC_i^\beta \\ \cW \subseteq \cX}} \big( \tfrac{p}{1-p} \big)^{|\cX|} \sum_{\cV \subseteq \cX} p^{|\cV \setminus \cU|} \big( \tfrac{-1}{p} \big)^{|\cV|} \tag{reorder summations} \\
    \notag &= \big( \tfrac{-1}{p} \big)^{|\cW|} \cdot p^{|\cU|} \sum_{\substack{\cX \in \cC_i^\beta \\ \cW \subseteq \cX}} \big( \tfrac{p}{1-p} \big)^{|\cX|} \sum_{\cV_1 \subseteq \cX \cap \cU} \big( \tfrac{-1}{p} \big)^{|\cV_1|} \sum_{\cV_2 \subseteq \cX \setminus \cU} (-1)^{|\cV_2|} \tag{$\cV_1 := \cV \cap \cU, \cV_2 := \cV \setminus \cU$} \\
    \notag &= \big( \tfrac{-1}{p} \big)^{|\cW|} \cdot p^{|\cU|} \sum_{\substack{\cX \in \cC_i^\beta \\ \cW \subseteq \cX}} \big( \tfrac{p}{1-p} \big)^{|\cX|} \cdot \mathbb{I} \big( \cX \subseteq \cU \big) \sum_{\cV_1 \subseteq \cX} \big( \tfrac{-1}{p} \big)^{|\cV_1|} \tag{binomial theorem}  \\
    \notag &= \big( \tfrac{-1}{p} \big)^{|\cW|} \cdot p^{|\cU|} \sum_{\substack{\cX \in \cC_i^\beta \\ \cW \subseteq \cX}} \big( \tfrac{p}{1-p} \big)^{|\cX|} \cdot \mathbb{I} \big( \cX \subseteq \cU \big) \cdot \big( \tfrac{p-1}{p} \big)^{|\cX|} \tag{binomial theorem} \\
    &= \big( \tfrac{-1}{p} \big)^{|\cW|} \cdot p^{|\cU|} \sum_{\substack{\cX \in \cC_i^\beta \\ \cW \subseteq \cX \subseteq \cU}} \hspace{-4pt} (-1)^{|\cX|} \label{eq:mat_prod_term}
\end{align}

\noindent Right-multiplying this expression by $\psi_i^\beta$ sums the \eqref{eq:mat_prod_term} terms over all $\cW \in \cC_i^\beta$ with $\cW \ne \varnothing$, giving,
\begin{align}
    \sum_{\substack{\cW \in \cC_i^\beta \\ \cW \ne \varnothing }} \big( \tfrac{-1}{p} \big)^{|\cW|} \cdot p^{|\cU|} \sum_{\substack{\cX \in \cC_i^\beta \\ \cW \subseteq \cX \subseteq \cU}} \hspace{-4pt} (-1)^{|\cX|} \hspace{4pt}
    &= p^{|\cU|} \sum_{\substack{\cX \in \cC_i^\beta \\ \cX \subseteq \cU}} (-1)^{|\cX|} \sum_{\substack{\cW \subseteq \cX \\ \cW \ne \varnothing }} \big( \tfrac{-1}{p} \big)^{|\cW|}  \tag{reorder summations} \\
    \notag &= p^{|\cU|} \sum_{\substack{\cX \in \cC_i^\beta \\ \cX \subseteq \cU}} (-1)^{|\cX|} \Big[ \big( \tfrac{p-1}{p} \big)^{|\cX|} - 1 \Big] \tag{binomial theorem} \\
    &= p^{|\cU|} \sum_{\substack{\cX \in \cC_i^\beta \\ \cX \subseteq \cU}} \Big[ \big( \tfrac{1-p}{p} \big)^{|\cX|} - (-1)^{|\cX|} \Big] \label{eq:mat_times_psi_1}.
\end{align}
For all $0 < p \leq \tfrac{1}{2}$, $\frac{1-p}{p} \geq 1$, such that the bracketed expression in \eqref{eq:mat_times_psi_1} is non-negative. This allows us to lower-bound \eqref{eq:mat_times_psi_1} by $0$. Additionally, we can upper-bound  \eqref{eq:mat_times_psi_1} by expanding the sum to include all subsets of $\cU$ and using the binomial theorem to get
\[
    p^{|\cU|} \sum_{\substack{\cX \subseteq \cU}} \Big[ \big( \tfrac{1-p}{p} \big)^{|\cX|} - (-1)^{|\cX|} \Big] = 1.
\]
Plugging these bounds back into \eqref{eq:w_inner_prod}, we can bound the bias
\[
    \sum_{\substack{\cU \in \cC_i^{\beta^*} \\ |\cU| > \beta}} - x_{i,\cU} \cdot \mathbb{I}\big( x_{i,\cU} > 0 \big)
    \quad \leq \quad \E[\widehat{\TTE}] - \TTE \quad \leq \quad 
    \sum_{\substack{\cU \in \cC_i^{\beta^*} \\ |\cU| > \beta}} - x_{i,\cU} \cdot \mathbb{I}\big( x_{i,\cU} < 0 \big).
\]
This allows us to bound the magnitude of the bias by the 1-norm
\[
    |\E[\widehat{\TTE}] - \TTE| \leq \|\bx_i^{>\beta}\|_1.
\]

Note that \eqref{eq:mat_times_psi_1} is not only bounded within 0 and 1, but it is also only a function of the size of the set $\cU$, and not the identity of $\cU$. As a result, a slightly tighter upper bound would be
\[
    |\E[\widehat{\TTE}] - \TTE| 
    \leq \sum_{\ell=\beta+1}^{\beta^*} \Bigg|\sum_{\substack{\cU \in \cN_i \\ |\cU| = \ell}} x_{i,\cU}\Bigg|.
\]
This is upper bounded by $\|\bx_i^{>\beta}\|_1$ but could be smaller when $x_{i,\cU}$ may have different signs for sets $\cU$ of the same cardinality.
\hfill
$\square$

\subsection*{Proof of Lemma~\ref{lem:gamma_bound}}

We will make use of the following computational lemma in our proof.

\begin{lemma} \label{lem:quadratic_form_sum}
    Suppose that $\bw_i \in \{0,1\}^{|\cC(\cN_i)|} \sim \textrm{Bern}(p)$ and $\psi_i^\beta \in \{0,1\}^{|\cC_i^\beta|}$ is the vector with 0 in its first entry (corresponding to $\cU = \varnothing$) and 1 in all other entries. Then,
    \[
        \big(\psi_i^\beta\big)^\intercal \E\big[\tbw_i^\beta(\tbw_i^\beta)^\intercal\big]^{\dagger} \psi_i^\beta = 
        \sum_{x=1}^{\min\big(\beta,|\cC(\cN_i)|\big)} \binom{|\cC(\cN_i)|}{x} \bigg[ \Big( \frac{1-p}{p} \Big)^x - 2(-1)^x + \Big( \frac{p}{1-p} \Big)^x \bigg]. 
    \]
\end{lemma}

\vspace{10pt}
\textit{Proof:} 
    Using \eqref{eq:pseudo_entry}, we may simplify this quadratic form:
    \[
        \big(\psi_i^\beta\big)^\intercal \E\big[\tbw_i^\beta(\tbw_i^\beta)^\intercal\big]^{\dagger} \psi_i^\beta = \sum_{\substack{\cV \in \cC_i^\beta \\ \cV \ne \varnothing}} \sum_{\substack{\cW \in \cC_i^\beta \\ \cW \ne \varnothing}} \big( \tfrac{-1}{p} \big)^{|\cV| + |\cW|} \hspace{-8pt}\sum_{\substack{\cX \in \cC_i^\beta \\ (\cV \cup \cW) \subseteq \cX}} \hspace{-8pt} \big( \tfrac{p}{1-p} \big)^{|\cX|} 
    \]
    Let us define $x = |\cX|$, $a = |\cV \cap \cW|$, $b = |\cV \setminus \cW|$, and $c = |\cW \setminus \cV|$. Note that $a+b = |\cV|$ and $a+c = |\cW|$. Finally, for ease of notation, let us define $m = \min\big( \beta,|\cC(\cN_i)| \big)$. Then, we can rewrite our quadratic form:
    \[
        \sum_{x=1}^{m} \binom{|\cC(\cN_i)|}{x} \Big( \frac{p}{1-p} \Big)^u \sum_{a=0}^{x} \binom{x}{a} \sum_{b=0}^{x-a} \binom{x-a}{b} \sum_{c=0}^{x-a-b} \binom{x-a-b}{c} \; \Ind\Big(\begin{matrix}a+b \geq 1 \\ a+c \geq 1 \end{matrix}\Big) \Big( \frac{-1}{p} \Big)^{2a+b+c}. 
    \]
    There are two possible (disjoint) cases in which this indicator is non-zero. 

    \textbf{Case 1:} $a=0$; $b,c \geq 1$

    In this case, we may further simplify the quadratic form:
    \begin{align*}
        &\phantom{=} \sum_{x=1}^{m} \binom{|\cC(\cN_i)|}{x} \Big( \frac{p}{1-p} \Big)^x \sum_{b=1}^{x} \binom{x}{b} \Big( \frac{-1}{p} \Big)^b \; \sum_{c=1}^{x-b} \binom{x-b}{c} \Big( \frac{-1}{p} \Big)^c \\
        &= \sum_{x=1}^{m} \binom{|\cC(\cN_i)|}{x} \Big( \frac{p}{1-p} \Big)^x \sum_{b=1}^{x} \binom{x}{b} \Big( \frac{-1}{p} \Big)^b \bigg[ \Big(\frac{p-1}{p}\Big)^{x-b} -1 \bigg] \tag{binomial theorem} \\
        &= \sum_{x=1}^{m} \binom{|\cC(\cN_i)|}{x} \Big( \frac{p}{1-p} \Big)^x \bigg[ \Big(\frac{p-2}{p}\Big)^x - 2\Big(\frac{p-1}{p}\Big)^x + 1 \bigg] \tag{binomial theorem} 
    \end{align*}

    \textbf{Case 2:} $a \geq 1$

    In this case, we may further simplify the quadratic form:

    \begin{align*}
        &\phantom{=} \sum_{x=1}^{m} \binom{|\cC(\cN_i)|}{x} \Big( \frac{p}{1-p} \Big)^x \sum_{a=1}^{x} \binom{x}{a} \Big( \frac{1}{p^2} \Big)^a \sum_{b=0}^{x-a} \binom{x-a}{b} \Big( \frac{-1}{p} \Big)^b \sum_{c=0}^{x-a-b} \binom{x-a-b}{c} \Big( \frac{-1}{p}\Big)^c \\
        &= \sum_{x=1}^{m} \binom{|\cC(\cN_i)|}{x} \Big( \frac{p}{1-p} \Big)^x \sum_{a=1}^{x} \binom{x}{a} \Big( \frac{1}{p^2} \Big)^a \sum_{b=0}^{x-a} \binom{x-a}{b} \Big( \frac{-1}{p} \Big)^b \Big(\frac{p-1}{p}\Big)^{x-a-b} \tag{binomial theorem} \\
        &= \sum_{x=1}^{m} \binom{|\cC(\cN_i)|}{x} \Big( \frac{p}{1-p} \Big)^x \sum_{a=1}^{x} \binom{x}{a} \Big( \frac{1}{p^2} \Big)^a \Big(\frac{p-2}{p}\Big)^{x-a} \tag{binomial theorem} \\
        &= \sum_{x=1}^{m} \binom{|\cC(\cN_i)|}{x} \Big( \frac{p}{1-p} \Big)^x \bigg[ \Big( \frac{p-1}{p} \Big)^{2x} - \Big(\frac{p-2}{p}\Big)^x \bigg] \tag{binomial theorem}
    \end{align*}

    \vspace{10pt}
    Combining the two cases, our quadratic form simplifies to 
    \begin{align*}
        &\phantom{=} \sum_{x=1}^{m} \binom{|\cC(\cN_i)|}{x} \Big( \frac{p}{1-p} \Big)^x \bigg[ \Big( \frac{p-1}{p} \Big)^{2x} - 2\Big(\frac{p-1}{p}\Big)^x + 1 \bigg] \\
        &= \sum_{x=1}^{m} \binom{|\cC(\cN_i)|}{x} \bigg[ \Big( \frac{1-p}{p} \Big)^x - 2(-1)^x + \Big( \frac{p}{1-p} \Big)^x \bigg].
    \end{align*}
\hfill
$\square$

Now, we are ready to prove Lemma~\ref{lem:gamma_bound}.

\vspace{10pt}
\noindent \textit{Proof of Lemma~\ref{lem:gamma_bound}:}
      We may use \eqref{eq:inverse_design_GCR} and the substitution $\theta_i^\beta = D_i^\beta \psi_i^\beta$, to rewrite
    \begin{align*}
        \gamma_i^2 
        &= (\theta_i^\beta)^\intercal \E\big[\tbz_i^\beta (\tbz_i^\beta)^\intercal\big]^\dagger \theta_i^\beta \\
        &= \Big[ (\psi_i^\beta)^\intercal (D_i^\beta)^\intercal \Big] \Big[ D_i^\beta \big( (D_i^\beta)^\intercal D_i^\beta \big)^{-1} \E\big[\tbw_i^\beta (\tbw_i^\beta)^\intercal\big]^\dagger \big( (D_i^\beta)^\intercal D_i^\beta \big)^{-1} \Big] \Big[ (D_i^\beta)^\intercal D_i^\beta \psi_i^\beta \Big] \\
        &= (\psi_i^\beta)^\intercal \E\big[\tbw_i^\beta (\tbw_i^\beta)^\intercal\big]^\dagger \psi_i^\beta \\
        &= \sum_{x=1}^{\min\big(\beta,|\cC(\cN_i)|\big)} \binom{|\cC(\cN_i)|}{x} \bigg[ \Big( \frac{1-p}{p} \Big)^x - 2(-1)^x + \Big( \frac{p}{1-p} \Big)^x \bigg]. \tag{Lemma~\ref{lem:quadratic_form_sum}}
    \end{align*}

    First, we can show that $\Big( \frac{1-p}{p} \Big)^x - 2(-1)^x + \Big( \frac{p}{1-p} \Big)^x$ is always nonzero for $x \geq 1$. This follows from verifying that $\Big( \frac{1-p}{p} \Big)^x + \Big( \frac{p}{1-p} \Big)^x$ is minimized at $p=\frac12$. In particular the partial derivative with respect to $p$ stated below is always negative for $p \in [0,\frac12)$, positive for $p \in (\frac12,1]$, and zero at $p=\frac12$,
    \[
        \frac{\partial}{\partial p} \left[\Big( \frac{1-p}{p} \Big)^x + \Big( \frac{p}{1-p} \Big)^x\right]
        = \frac{x}{p(1-p)} \left(\Big( \frac{p}{1-p}\Big)^x - \Big(\frac{1-p}{p}\Big)^x \right).
    \]
    As a result, by plugging in $p=\frac12$ we can verify that 
    \[\Big( \frac{1-p}{p} \Big)^x - 2(-1)^x + \Big( \frac{p}{1-p} \Big)^x
    \geq 2 - 2(-1)^x \geq 0.\]
    Since each term of the summation is thus nonnegative, we can upper bound $\gamma_i^2$ by
    \[
        \gamma_i^2 \leq \min\bigg(\sum_{x=1}^{|\cC(\cN_i)|} \binom{|\cC(\cN_i)|}{x} \bigg[ \Big( \frac{1-p}{p} \Big)^x - 2(-1)^x + \Big( \frac{p}{1-p} \Big)^x \bigg],
    \]
    \[
         \phantom{a} \hspace{50pt} \sum_{x=1}^{\beta} \binom{|\cC(\cN_i)|}{x} \bigg[ \Big( \frac{1-p}{p} \Big)^x - 2(-1)^x + \Big( \frac{p}{1-p} \Big)^x \bigg]\bigg).
    \]

    From here, we bound the two expressions in the above $\min$. By applying the binomial theorem, we can show the first expression is bounded by
    \begin{align*}
    \sum_{x=1}^{|\cC(\cN_i)|} \binom{|\cC(\cN_i)|}{x} \left[\left(\frac{p}{1-p}\right)^x - 2(-1)^x+\left(\frac{1-p}{p}\right)^x\right] 
        =\left(\frac{1}{1-p}\right)^{|\cC(\cN_i)|}+\left(\frac{1}{p}\right)^{|\cC(\cN_i)|} 
        \leq 2 \, p^{-|\cC(\cN_i)|}.
    \end{align*}

    Next, we bound the second expression. WLOG, we may assume that $0 \leq p \leq \frac{1}{2}$; otherwise we may substitute $1 - p \mapsto p$ since $\gamma_i$ is symmetric with respect to $p$ and $1-p$. From here, we note that we can bound
    \begin{equation} \label{eq:p_power_ineq}
        \left( \frac{p}{1-p} \right)^x + \left( \frac{1-p}{p} \right)^x \leq p^{-x},
    \end{equation}
    for all $0 \leq p \leq \frac{1}{2}$ and $x \geq 1$. This is because the function
    \[
        f(\alpha) = \left( \frac{\alpha}{1-p} \right)^x + \left( \frac{1-\alpha}{p} \right)^x
    \]
    is decreasing on the range $0 \leq \alpha \leq p$ since 
    \[
        \frac{\partial f}{\partial \alpha} = \frac{x}{1-p} \left( \frac{\alpha}{1-p} \right)^{x-1} - \frac{x}{p} \left( \frac{1-\alpha}{p} \right)^{x-1} \leq 0.
    \]
    Note that the first term of this derivative is upper bounded by $2x$ and the second term is upper bounded by $-2x$. Thus, it attains its maximum on this range at $\alpha = 0$. We can use \eqref{eq:p_power_ineq} to bound
    \begin{align*}
        \gamma_i &\leq \sum_{x=1}^{\beta} \binom{|\cC(\cN_i)|}{x} \big[ 2 + p^{-x} \big] \\
        &\leq 2p^{-\beta} \sum_{x=1}^{\beta} \binom{|\cC(\cN_i)|}{x} \tag{$2 \leq p^{-1} \leq p^{-x} \leq p^{-\beta}$} \\
        &\leq 2 \, |\cC(\cN_i)|^{\beta} \cdot p^{-\beta} . 
    \end{align*}
    For the final inequality, note that mapping the  $|\cC(\cN_i)|^\beta$ ordered sequences of length $\beta$ with elements from $\cC(\cN_i)$ to their sets of unique elements is a surjection onto the non-empty subsets of $\cC(\cN_i)$ with size at most $\beta$. Hence, $|\cC(\cN_i)|^{\beta} \geq \sum_{x=1}^{\beta} \binom{|\cC(\cN_i)|}{x}$.
\hfill
$\square$

\subsection*{Proof of Remark~\ref{rem:gcr_pi_better}}

Next, we turn to proving Remark~\ref{rem:gcr_pi_better}, which requires first deriving the closed form of the Horvitz--Thompson estimators under a $\GCR$ design.

\begin{lemma} \label{lem:ht_gcr_est_form}
    Given an interference network on $n$ individuals and a clustering $\cC$ of these individuals, suppose that we sample treatment assignments according to a graph cluster randomized design $\bz \sim \emph{GCR}(\cC,p)$. Then, the Horvitz-Thompson estimator for $\TTE$ is given by the formula
    \[
        \widehat{\TTE}_{\emph{HT}} = \tfrac{1}{n} \sum_{i=1}^{n} \; Y_i(\bz) \sum_{\cU \subseteq \cC(\cN_i)} \bigg[ \big( 1-p \big)^{|\cU|} - \big( -p \big)^{|\cU|} \bigg] \prod_{C \in \cU} \tfrac{w_C - p}{p(1-p)}.
    \]
\end{lemma} 

\begin{proof}
    We have,
    \begin{align*}
        \sum_{\cU \subseteq \cC(\cN_i)} &\bigg[ \big( 1-p \big)^{|\cU|} - \big( -p \big)^{|\cU|} \bigg] \prod_{C \in \cU} \tfrac{w_C - p}{p(1-p)} \\
        &= \sum_{\cU \subseteq \cC(\cN_i)} \prod_{C \in \cU} \tfrac{w_C - p}{p} - \sum_{\cU \subseteq \cC(\cN_i)} \prod_{C \in \cU} \tfrac{p-w_C}{1-p} \\
        &= \prod_{C \in \cC(\cN_i)} \tfrac{w_C}{p} - \prod_{C \in \cC(\cN_i)} \tfrac{1-w_C}{1-p} \tag{distributivity} \\
        &= \frac{\Ind \big( \cN_i \textrm{ entirely treated} \big)}{\Pr \big( \cN_i \textrm{ entirely treated} \big)} - \frac{\Ind \big( \cN_i \textrm{ entirely untreated} \big)}{\Pr \big( \cN_i \textrm{ entirely untreated} \big).}
    \end{align*}
\end{proof}

With this closed form in hand, we can prove Remark~\ref{rem:gcr_pi_better}.


\gcrhtbound*

\begin{proof}
    For ease of notation, let us define the function $g \colon \mathbb{N} \to [0,1]$ with $g(n) = (1-p)^n - (-p)^n$. Note that $0 \leq g(n) \leq 1$ for all $n \in \mathbb{N}$ and all $0 \leq p \leq \frac{1}{2}$. By Lemmas \ref{lem:pi_gcr_est_form} and \ref{lem:ht_gcr_est_form}, we may represent
    \begin{align*}
        \widehat{\TTE}_{\beta} &= \tfrac{1}{n} \sum_{i=1}^{n} \; Y_i(\bz) \sum_{\cU \in \cC_i^\beta} g(|\cU|) \prod_{C \in \cU} \tfrac{w_C - p}{p(1-p)}, \\
        \widehat{\TTE}_{\textrm{HT}} &= \tfrac{1}{n} \sum_{i=1}^{n} \; Y_i(\bz) \sum_{\cU \subseteq \cC(\cN_i)} g(|\cU|) \prod_{C \in \cU} \tfrac{w_C - p}{p(1-p)}, \\
        &= \widehat{\TTE}_{\beta} +  \tfrac{1}{n} \sum_{i=1}^{n} \; Y_i(\bz) \sum_{\cU \in \cC_i^{>\beta}} g(|\cU|) \prod_{C \in \cU} \tfrac{w_C - p}{p(1-p)},
    \end{align*}
    where $\cC_i^{>\beta}$ denotes the subsets of $\cC(\cN_i)$ with cardinality strictly larger than $\beta$. As a result, $\var\Big(\widehat{\TTE}_{\textrm{HT}} \Big) - \var\Big(\widehat{\TTE}_{\beta} \Big)$ is equal to 
    \begin{align} \label{eq:diff_var}
    \var\Big( \tfrac{1}{n} \sum_{i=1}^{n} \; Y_i(\bz) \sum_{\cU \in \cC_i^{>\beta}} g(|\cU|) \prod_{C \in \cU} \tfrac{w_C - p}{p(1-p)} \Big) + 2 \cov\Big(\widehat{\TTE}_{\beta}, \tfrac{1}{n} \sum_{i=1}^{n} \; Y_i(\bz) \sum_{\cU \in \cC_i^{>\beta}} g(|\cU|) \prod_{C \in \cU} \tfrac{w_C - p}{p(1-p)}\Big).
    \end{align}
    
    Using the notation $Y_i(\bz) = Y_i(\bw) = \langle \tbw_i^{\beta^*}, \bx_i \rangle$ for $\bx_i = \big(D_i^{\beta^*}\big)^\intercal \bc_i$, we can expand the covariance term as follows,
    \begin{align}
        \notag \cov\Big(\widehat{\TTE}_{\beta}, \tfrac{1}{n} \sum_{i=1}^{n} \; Y_i(\bz) \sum_{\cU \in \cC_i^{>\beta}} g(|\cU|) \prod_{C \in \cU} \tfrac{w_C - p}{p(1-p)}\Big) 
        &= \tfrac{1}{n^2} \sum_{i=1}^{n} \sum_{i'=1}^{n} \sum_{\cS \in \cC_i^{\beta^*}} \bx_{i,\cS} \sum_{\cS' \in \cC_{i'}^{\beta^*}} \bx_{i',\cS'} \sum_{\cU \in \cC_i^{>\beta}} g(|\cU|) \\ 
        \phantom{a} \sum_{\cU' \in \cC_{i'}^{> \beta}} g(|\cU'|) \cdot \cov \bigg( \prod_{C_1 \in \cS} w_{C_1} \prod_{C_2 \in \cU} \tfrac{w_{C_2} - p}{p(1-p)} \;,\; &\prod_{C_3 \in \cS'} w_{C_3} \prod_{C_4 \in \cU'} \tfrac{w_{C_4} - p}{p(1-p)} \bigg) \label{eq:var_6sum}.
    \end{align}

    By an analogous calculation to that from Lemma 4 of~\cite{cortez2023exploiting}, we may simplify the inner covariance term to
    \[
        \cov = \Ind \bigg( \begin{matrix} \cU \subseteq (\cS \cup \cS' \cup \cU') \\ \cU' \subseteq (\cS \cup \cU \cup \cS') \end{matrix} \bigg) \cdot p^{|(\cS \cup \cS') \setminus (\cU \cup \cU')| - |\cU \cup \cU'|} \cdot \Big( \tfrac{1}{1-p} \Big)^{|(\cU \cap \cU') \setminus (\cS \cup \cS')|} - \Ind \bigg( \begin{matrix} \cU \subseteq \cS \\ \cU' \subseteq \cS' \end{matrix} \bigg) \cdot p^{|\cS \setminus \cU| + |\cS' \setminus \cU'|}.
    \]

    Note that $|\cU'| > \beta$ while $|\cS'| \leq \beta^*$, such that because we assume $\beta \geq \beta^*$ then $\cU' \not\subseteq \cS'$ such that the second term of the above expression is 0. Thus \eqref{eq:var_6sum} reduces to 
    \begin{align*}
    &\tfrac{1}{n^2} \sum_{i=1}^{n} \sum_{i'=1}^{n} \sum_{\cS \in \cC_i^{\beta^*}} \bx_{i,\cS} \sum_{\cS' \in \cC_{i'}^{\beta^*}} \bx_{i',\cS'} \sum_{\cU \in \cC_i^{\beta}} g(|\cU|) \sum_{\cU' \in \cC_i^{> \beta}} g(|\cU'|) \\
    &\phantom{a}\hspace{30pt} \cdot \Ind \bigg( \begin{matrix} \cU \subseteq (\cS \cup \cS' \cup \cU') \\ \cU' \subseteq (\cS \cup \cU \cup \cS') \end{matrix} \bigg) \cdot p^{|(\cS \cup \cS') \setminus (\cU \cup \cU')| - |\cU \cup \cU'|} \cdot \Big( \tfrac{1}{1-p} \Big)^{|(\cU \cap \cU') \setminus (\cS \cup \cS')|}.
    \end{align*}
    
    Our assumption that each $c_{i,S}$ has the same sign implies that each $\bx_{i,\cS}$ also has the same sign. 
    Our same-sign assumption on the $\bx_{i,\cS}$ coefficients and our definition of $g$ ensure that the covariance term in \eqref{eq:diff_var} is non-negative. It follows then that \eqref{eq:diff_var} is non-negative, implying that $\var\Big( \widehat{\TTE}_{\textrm{HT}} \Big) \geq \var \Big( \widehat{\TTE}_{\beta} \Big).$
\end{proof}

\section{Completely Randomized Designs} \label{sec:crd}

In settings with a small number of clusters, it is preferable to completely randomize cluster treatments rather than Bernoulli randomize, since this avoids the possibility of having no treated or no control units. We derive the pseudoinverse estimator for this setting and show that, unlike in the Bernoulli case, it exhibits a bias-variance trade-off that is mediated by properties of the clustering. The relevant feature of the clustering is the number of nodes that have an in-neighbor in every cluster; as the number of such nodes increases, the variance decreases at the cost of increased bias. This bias-variance trade-off is unusual in that it is not controlled by a parameter of the estimator but by the experimental design itself. The analysis of completely randomized designs is more demanding, so we will limit our discussion to the special case of \textit{linear} ($\beta^* = 1$) potential outcomes.

Given a clustering $\cC=\{C_1,\cdots, C_m\}$ into $m$ clusters, a \textit{completely randomized GCR design} $\bz \sim \textrm{GCR}_{\textrm{Comp}}(\cC,k)$ selects $k$ out of the $m$ clusters to \textit{fully} treat, and assigns the individuals in all other clusters to control. Since all individuals within any cluster have the same treatment assignment under $\bz \sim \textrm{GCR}_{\textrm{Comp}}(\cC,k)$. Performing a similar ``unit to cluster'' reduction as Section~\ref{subsec:pi_gcr}, we may write an individual's outcome
\[
    Y_i(\bz) = c_{i,\varnothing} + \sum_{j \in \cN_i} c_{i,\{j\}} z_j = c_{i,\varnothing} + \sum_{C \in \cC(\cN_i)} \sum_{j \in \cN_i \cap C} c_{i,\{j\}} z_j := x_{i,\varnothing} + \sum_{C \in \cC(\cN_i)} x_{i,\{C\}} w_C := Y_i(\bw),
\]
where we define the coefficients $x_{i,\varnothing} = c_{i,\varnothing}$ and $x_{i,\{C\}} = \sum_{j \in \cN_i \cap C} c_{i,\{j\}}$. Note that the second equality uses the fact that all individuals in a cluster $C$ receive the same treatment assignment, which we denote by $w_C$. By the same reasoning as Section~\ref{sec:pinv}, we may express the pseudoinverse estimator for the TTE as 
\begin{equation} \label{eq:gcrcomp_est}
    \widehat{\TTE}_{1} 
    = \frac{1}{n} \sum_{i=1}^{n} Y_i(\bw) \Big\langle \E \big[ \tbw_i \tbw_i^\intercal \big]^{\dagger} \psi_i, \tbw_i \Big\rangle
\end{equation}
where $\bw \in \{0,1\}^m$ is the vector indicating the treatment decision for each \textit{cluster}; under a completely randomized design it has exactly $k$ entries taking value 1 and $(m-k)$ entries taking value 0. We'll write that $\bw \sim \textrm{CRD}(k)$. The vector $\tbw_i \in \{0,1\}^{|\cC(\cN_i)| + 1}$ indicates the complete treatment of each cluster among the neighbors of $i$ (along with the $1$-entry representing the baseline effect). The vector $\psi_i = [0 1 \cdots 1]^T$ collects the coefficients of the TTE estimand. If we wanted to write the estimator in terms of the unit treatment vector as in the original presentation of the pseudoinverse estimator in \eqref{eq:tte_hat}, we would let $\tbz_i = D_i^1 \tbw_i$, $\theta_i = D_i^1 \psi_i$, and $\bx_i = (D_i^1)^\intercal \bc_i$, for $D_i^1$ defined as in Section~\ref{sec:pinv}.

Some of the results that we will present here rely on an additional \textit{monotonicity} assumption across these cluster-aggregated effect coefficients.

\begin{assumption}[Monotonicity] \label{ass:strong_mono}
    The sign of $x_{i, C}$ is constant across all $i$ and $C$ for $C = \emptyset$ or $C \in \cC_i^1$.
\end{assumption}

Assumption~\ref{ass:strong_mono} implies that the effect of each individual either always increases or always decreases when additional clusters are assigned treatment, and this direction is the same across all individuals. Practically, this assumption precludes the possibility of ``cancellation'' between interference effects, which will be crucial for our analysis of the variance under completely randomized designs.

\subsection{The Explicit Pseudoinverse Estimator when $\beta = 1$} \label{subsec:crd_pseudo}

Now, we will derive an explicit form for the pseudoinverse estimator \eqref{eq:gcrcomp_est} under a $\textrm{GCR}_{\textrm{Comp}}(\cC,k)$ design and use this to analyze its bias and variance. Under a $\textrm{GCR}_{\textrm{Comp}}(\cC,k)$ design, the treatment decisions are correlated to ensure that a fixed number of clusters are treated. Each time a cluster $\cC$ is assigned to treatment, this uses up one of the $k$ treatment allotments and reduces the conditional treatment probability of all other clusters. Thus, given two subsets of clusters $\cS, \cT \in \cC_i^1$, 
\[
    \Big[ \E\big[ \tbw_i \tbw_i^\intercal \big] \Big]_{\cS,\cT}
    = \Pr \Big( \bigcap_{C \in \cS \cup \cT} w_C = 1 \Big)
    = \prod_{\ell = 0}^{|\cS \cup \cT| - 1} \!\!\frac{k-\ell}{m-\ell}.
\]

\vspace{4pt}
\noindent As an example, if individual $i$ has neighbors from two clusters $C_1$ and $C_2$, their (cluster) design matrix has the form,
\begin{center}
\begin{tikzpicture}
    \node at (-1.5,1.5) {$\E \big[ \tbw_i \tbw_i^\intercal \big] =$};

    \node at (0,2.5) {\scriptsize $\varnothing$};
    \node at (0,1.5) {\scriptsize $\{C_1\}$};
    \node at (0,0.5) {\scriptsize $\{C_2\}$};

    \node at (1,3.25) {\scriptsize $\varnothing$};
    \node at (2.25,3.25) {\scriptsize $\{C_1\}$};
    \node at (3.5,3.25) {\scriptsize $\{C_2\}$};

    \draw[thick] (0.75,0) -- (0.5,0) -- (0.5,3) -- (0.75,3);

    \draw[thick] (4,0) -- (4.25,0) -- (4.25,3) -- (4,3);

    \node at (1,2.5) {$1$};
    \node at (1,1.5) {$\frac{k}{n}$};
    \node at (1,0.5) {$\frac{k}{n}$};

    \node at (2.25,2.5) {$\frac{k}{n}$};
    \node at (2.25,1.5) {$\frac{k}{n}$};
    \node at (2.25,0.5) {$\frac{k(k-1)}{n(n-1)}$};

    \node at (3.5,2.5) {$\frac{k}{n}$};
    \node at (3.5,1.5) {$\frac{k(k-1)}{n(n-1)}$};
    \node at (3.5,0.5) {$\frac{k}{n}$};

    \node at (4.5,1.5) {$.$};
    \end{tikzpicture}
\end{center}

Note that the entries in this matrix are functions not only of the treatment level $k$ but also of the number of clusters $m$. The following lemma shows that the design matrix for complete randomization, unlike Bernoulli randomization, need not be invertible. 

\begin{lemma} \label{lem:crd_det}
    Suppose that $\bw \sim \textrm{CRD}(k)$ and Assumption~\ref{ass:low_deg} holds with $\beta=1$. Then,
    \[
        \emph{det} \Big( \E \big[ \tbw_i \tbw_i^\intercal \big] \Big) = \frac{k^{|\cC(\cN_i)|} \cdot (m-k)^{|\cC(\cN_i)|} \cdot (m-|\cC(\cN_i)|)}{m^{(|\cC(\cN_i)|+1)} \cdot (m-1)^{|\cC(\cN_i)|}}. 
    \]
\end{lemma}

\begin{proof}
     \textit{Proof:} \: We carry out a cofactor expansion of this determinant. Throughout the proof, we let $\mathbf{1}_{x \times y}$ denote the $x \times y$ (sub)matrix of all $1$s, and $\mathbf{I}_x$ denote the $x \times x$ identity (sub)matrix. For each $x \in \mathbb{N}$, we also recursively define the $x \times x$ matrices
    \[
        \mathbf{A}_x = \underbrace{\begin{bmatrix}
            1 & \tfrac{k-1}{m-1} \cdot \mathbf{1}_{1 \times (x-1)}\\
            \tfrac{k-1}{m-1} \cdot \mathbf{1}_{(x-1) \times 1} & \mathbf{A}_{x-1}
        \end{bmatrix}}_{\tfrac{k-1}{m-1} \cdot \mathbf{1}_{x \times x} + \tfrac{m-k}{n-1} \cdot \mathbf{I}_x},
        \hspace{20pt}
        \mathbf{B}_x = \begin{bmatrix}
            1 & \tfrac{k-1}{m-1} \cdot \mathbf{1}_{1 \times (x-1)}\\
            \mathbf{1}_{(x-1) \times 1} & \mathbf{A}_{x-1}
        \end{bmatrix}.
    \]
    Note that 
    \[
        \E \big[ \tbw_i \tbw_i^\intercal \big] = \begin{bmatrix}
            1 & \tfrac{k}{m} \cdot \mathbf{1}_{1 \times |\cC(\cN_i)|} \\
            \rule{0pt}{14pt} \tfrac{k}{m} \cdot \mathbf{1}_{|\cC(\cN_i)| \times 1} & \tfrac{k}{m} \cdot \mathbf{A}_{d_i}
        \end{bmatrix}
    \]
    \[
        \Rightarrow \quad 
        \textrm{det} \Big( \E \big[ \tbw_i \tbw_i^\intercal \big] \Big) = \frac{k^{|\cC(\cN_i)|}}{m^{|\cC(\cN_i)|}} \cdot \textrm{det} \bigg( \begin{bmatrix}
        1 & \tfrac{k}{m} \cdot \mathbf{1}_{1 \times |\cC(\cN_i)|} \\
        \rule{0pt}{14pt} \mathbf{1}_{|\cC(\cN_i)| \times 1} & \mathbf{A}_{|\cC(\cN_i)|} \end{bmatrix} \bigg).
    \]
    Performing a cofactor expansion of the latter matrix along its first row, we have
    \begin{equation} \label{eq:design_cofactor}
        \textrm{det} \Big( \E \big[ \tbw_i \tbw_i^\intercal \big] \Big) = \tfrac{k^{|\cC(\cN_i)|}}{m^{|\cC(\cN_i)|}} \bigg[ \textrm{det} \big( \mathbf{A}_{|\cC(\cN_i)|} \big) - |\cC(\cN_i)| \cdot \tfrac{k}{m} \cdot \textrm{det} \big( \mathbf{B}_{|\cC(\cN_i)|} \big) \bigg].
    \end{equation}
    Observe that the second cofactor in this expansion is exactly $\mathbf{B}_{|\cC(\cN_i)|}$, and each successive cofactor can be obtained from the previous one via a single row transposition. Hence, it suffices to reason about the determinants of the $\mathbf{A}_x$ and $\mathbf{B}_x$ matrices. 

    We compute the determinant of $\mathbf{A}_{x}$ as the product of its eigenvalues. Note that $\mathbf{1}_{x \times x}$ has eigenvalue $x$ with multiplicity $1$ and $0$ with multiplicity $x-1$. Hence, $\mathbf{A}_{x} = \tfrac{k-1}{m-1} \cdot \mathbf{1}_{x \times x} + \tfrac{m-k}{m-1} \cdot \mathbf{I}_{x}$ has determinant 
    \[
        \textrm{det} \big( \mathbf{A}_{x} \big) = \Big( x \cdot \tfrac{k-1}{m-1} + \tfrac{m-k}{m-1} \Big) \cdot \Big( \tfrac{m-k}{m-1} \Big)^{x-1}.
    \]
    Next, we argue by induction on $x$ that $\textrm{det} \big( \mathbf{B}_{x} \big) = \Big( \tfrac{m-k}{m-1} \Big)^{x-1}$. The base case is immediate, as $\mathbf{B}_1 = \begin{bmatrix} 1 \end{bmatrix}$. Given $x \geq 2$, we performing a cofactor expansion along the first row of $\mathbf{B}_{x}$ to obtain,
    \[
        \textrm{det} \big( \mathbf{B}_{x} \big) = \textrm{det}(\mathbf{A}_{x-1}) - (x-1) \cdot \tfrac{k-1}{m-1} \cdot \textrm{det}(\mathbf{B}_{x-1}).
    \]
    Observe that the second cofactor in this expansion is exactly $\mathbf{B}_{x-1}$, and each successive cofactor can be obtained from the previous one via a single row transposition. Applying our formula for $\mathbf{A}_{x-1}$ and the inductive hypothesis, we have
    \begin{align*}
        \textrm{det} \big( \mathbf{B}_{x} \big) 
        &= \Big( (x-1) \cdot \tfrac{k-1}{m-1} + \tfrac{m-k}{m-1} \Big) \cdot \Big( \tfrac{m-k}{m-1} \Big)^{x-2} - (x-1) \cdot \tfrac{k-1}{m-1} \cdot \Big( \tfrac{m-k}{m-1} \Big)^{x-2} \\
        &= \Big( \tfrac{m-k}{m-1} \Big)^{x-2} \cdot \Big[ \Big( (x-1) \cdot \tfrac{k-1}{m-1} + \tfrac{m-k}{m-1} \Big) - (x-1) \cdot \tfrac{k-1}{m-1} \Big]  \\
        &= \Big( \tfrac{m-k}{m-1} \Big)^{x-2} \cdot \tfrac{m-k}{m-1} \\
        &= \Big( \tfrac{m-k}{m-1} \Big)^{x-1},
    \end{align*}
    completing the induction. Finally, plugging into \eqref{eq:design_cofactor}, we have
    \begin{align*}
        \textrm{det} \Big( \E \big[ \tbw_i \tbw_i^\intercal \big] \Big) 
        &= \frac{k^{|\cC(\cN_i)|}}{m^{|\cC(\cN_i)|}} \bigg[ \Big( |\cC(\cN_i)| \cdot \tfrac{k-1}{m-1} + \tfrac{m-k}{m-1} \Big) \cdot \Big( \tfrac{m-k}{m-1} \Big)^{|\cC(\cN_i)|-1} - |\cC(\cN_i)| \cdot \tfrac{k}{m} \cdot \Big( \tfrac{m-k}{m-1} \Big)^{|\cC(\cN_i)|-1} \bigg] \\
        &= \frac{k^{|\cC(\cN_i)|} \cdot (m-k)^{(|\cC(\cN_i)|-1)}}{m^{|\cC(\cN_i)|} \cdot (m-1)^{(|\cC(\cN_i)|-1)}} \bigg[ |\cC(\cN_i)| \cdot \Big( \tfrac{k-1}{m-1} - \tfrac{k}{m} \Big) + \tfrac{m-k}{m-1} \bigg] \\
        &= \frac{k^{|\cC(\cN_i)|} \cdot (m-k)^{|\cC(\cN_i)|} \cdot (m-|\cC(\cN_i)|)}{m^{(|\cC(\cN_i)|+1)} \cdot (m-1)^{|\cC(\cN_i)|_i}}.
    \end{align*}\hfill$\square$
\end{proof}

Lemma~\ref{lem:crd_det} shows that the design matrix will fail to be invertible whenever $m = |\cC(\cN_i)|$, meaning individual $i$ has at least one incoming edge from \textit{every} cluster in the network. Such a situation may be inevitable if the graph is small (or naturally partitions into a small number of clusters) and has relatively high expansion.

\begin{lemma} \label{lem:pi_crd_form}
    Suppose that $\bz \sim \textrm{GCR}_{\textrm{Comp}}(\cC,k)$. If Assumption~\ref{ass:low_deg} holds with $\beta^* = 1$, then the pseudoinverse estimator for the total treatment effect is given by:
    \[
        \widehat{\TTE}_1 = \frac{1}{n} \sum_{i=1}^{n} Y_i(\bw) \cdot \begin{cases}
        \frac{m^2(m-1)}{k(m-k)(n-|\cC(\cN_i)|)} \sum\limits_{C \in \cC(\cN_i)} \big( w_C - \frac{k}{m} \big) & |\cC(\cN_i)| < m, \\
        \rule{0pt}{20pt}\frac{mk^2}{(k^2 + m)^2} \sum\limits_{C \in \cC(\cN_i)} \big( w_C + \frac{1}{k} \big) & |\cC(\cN_i)| = m.
    \end{cases}
    \]
\end{lemma}

\begin{proof}
    In light of Lemma~\ref{lem:crd_det}, we may compute the pseudoinverse of the cluster design matrix for individuals with $|\cC(\cN_i)| < m$ by taking the standard matrix inverse:
    \[
        \E \big[ \tbw_i \tbw_i^\intercal \big]^{-1} = \alpha \begin{bmatrix}
            \tfrac{k(m+|\cC(\cN_i)|(k-1)-k)}{m-1} & -k & -k & \hdots & -k \\[4pt]
            -k & m\!-\!(|\cC(\cN_i)|-1) & 1 & \hdots & 1 \\[4pt]
            -k & 1 & m\!-\!(|\cC(\cN_i)|-1) & \ddots & \vdots \\[4pt]
            \vdots & \vdots & \ddots & \ddots & 1 \\[4pt]
            -k & 1 & \hdots & 1 & m\!-\!(|\cC(\cN_i)|-1)
        \end{bmatrix},
    \]
    where $\alpha = \frac{m(m-1)}{k(m-k)(m-|\cC(\cN_i)|)}$ which can be verified by a direct computation. Multiplying this by $\psi_i$ gives
    \[
        \E \big[ \tbw_i \tbw_i^\intercal \big]^{\dagger} \psi_i = \frac{m(m-1)}{k(m-k)(m-|\cC(\cN_i)|)} \begin{bmatrix}
            -|\cC(\cN_i)| \cdot k \\
            m \\
            \vdots \\
            m
        \end{bmatrix},
    \]
    which in turn gives the inner product 
    \[
        \Big\langle \E \big[ \tbw_i \tbw_i^\intercal \big]^{\dagger} \psi_i, \tbw_i \Big\rangle = \tfrac{m^2(m-1)}{k(m-k)(m-|\cC(\cN_i))} \sum_{C \in \cC(\cN_i)} \Big( w_C - \tfrac{k}{m} \Big).
    \]

    \noindent Now, we turn our attention to individuals $i$ with $|\cC(\cN_i)| = m$. For these individuals, the cluster design matrix has pseudoinverse
    \[
        \E \big[ \tbw_i \tbw_i^\intercal \big]^{\dagger} = \frac{m}{k(m-k)(k^2+m)^2} \begin{bmatrix}
            mk(m-k) & k^2(m-k) & k^2(m-k) & \hdots & k^2(m-k) \\[4pt]
            k^2(m-k) & \star & \diamond & \hdots & \diamond \\[4pt]
            k^2(m-k) & \diamond & \star & \ddots & \vdots \\[4pt]
            \vdots & \vdots & \ddots & \ddots & \diamond \\[4pt]
            k^2(m-k) & \diamond & \hdots & \diamond & \star
        \end{bmatrix},
    \]
    where
    \begin{align*}
        \star &= (m-2)k^4 + k^3 + 2(m-1)^2k^2 + m(m-1)^2, \\
        \diamond &= - k^4 + k^3 - 2(m-1)k^2 - m(m-1).
    \end{align*}
    Multiplying this by $\E \big[ \tbw_i \tbw_i^\intercal \big]$, we find that 
    \[
        \E \big[ \tbw_i \tbw_i^\intercal \big]^{\dagger} \E \big[ \tbw_i \tbw_i^\intercal \big] 
        = \frac{1}{k^2+m}
        \begin{bmatrix}
            m & k & k & \hdots & k \\[4pt]
            k & k^2+m-1 & -1 & \hdots & -1 \\[4pt]
            k & -1 & k^2+m-1 & \ddots & \vdots \\[4pt]
            \vdots & \vdots & \ddots & \ddots & -1 \\[4pt]
            k & -1 & \hdots & -1 & k^2+m-1
        \end{bmatrix}.
    \]
    
    \vspace{4pt}
    Since this matrix is symmetric, so $\E \big[ \tbw_i \tbw_i^\intercal \big]^\dagger \E \big[ \tbw_i \tbw_i^\intercal \big] = \E \big[ \tbw_i \tbw_i^\intercal \big] \E \big[ \tbw_i \tbw_i^\intercal \big]^\dagger$ and the Hermitian conditions of the pseudoinverse are satisfied. It remains to check that $\E \big[ \tbw_i \tbw_i^\intercal \big] \E \big[ \tbw_i \tbw_i^\intercal \big]^\dagger \E \big[ \tbw_i \tbw_i^\intercal \big] = \E \big[ \tbw_i \tbw_i^\intercal \big]$ and $\E \big[ \tbw_i \tbw_i^\intercal \big]^\dagger \E \big[ \tbw_i \tbw_i^\intercal \big] \E \big[ \tbw_i \tbw_i^\intercal \big]^\dagger = \E \big[ \tbw_i \tbw_i^\intercal \big]^\dagger$. This can be verified through direct computation, which we omit here for the sake of brevity.

    Now, multiplying $\E \big[ \tbw_i \tbw_i^\intercal \big]^\dagger$ by $\psi_i$ gives
    \[
        \E \big[ \tbw_i \tbw_i^\intercal \big]^{\dagger} \psi_i = \frac{mk}{(k^2+m)^2} \begin{bmatrix}
            m \\
            k \\
            \vdots \\
            k \\
        \end{bmatrix},
    \]
    which in turn gives the inner product 
    \[
        \Big\langle \E \big[ \tbw_i \tbw_i^\intercal \big]^{\dagger} \psi_i, \tbw_i \Big\rangle = \tfrac{mk^2}{(k^2+m)^2} \sum_{C \in \cC(\cN_i)} \Big( w_C + \tfrac{1}{k} \Big).
    \]
\end{proof}

\subsection{Bias}

To reason about the bias of this pseudoinverse estimator, we make use of Theorem~\ref{thm:bias_bound_general}. Any individual $i \in [n]$ with $|\cC(\cN_i)| < m$ has an invertible cluster design matrix, so the reasoning from  Corollary~\ref{cor:gcr_unbiased} tells us that these individuals will not contributed to the bias. The remaining individuals have a non-invertible design matrix. The null space of these matrices has dimension $1$, suggesting that $\psi_i$ may still lie in the column space and ensure unbiasedness. The following example demonstrates that this need not be the case, and can be generalized to show that this is never the case.

\begin{example}
    We'll consider the $3 \times 3$ design matrix from the example in Section~\ref{subsec:crd_pseudo} with $k=1$ and $n=2$ (meaning exactly one of $C_1,C_2$ is randomly selected for treatment). Plugging these values in allows us to simplify this design matrix to,
    \[
        \E \big[ \tbw_i \tbw_i^\intercal \big] = \begin{bmatrix}
            1 & \frac{1}{2} & \frac{1}{2} \\[4pt]
            \frac{1}{2} & \frac{1}{2} & 0 \\[4pt]
            \frac{1}{2} & 0 & \frac{1}{2}
        \end{bmatrix}.
    \]
    The column space of this matrix is spanned by its first two columns (note that the third column is the first minus the second). The vector $\theta_1 = \begin{bmatrix} 0 & 1 & 1 \end{bmatrix}^\intercal$ does not lie in the column space; the only vectors in the column space with $0$ in their first entry are the scalar multiples of $\begin{bmatrix} 0 & -1 & 1 \end{bmatrix}^\intercal$.
\end{example}

In fact, for any $i$ with $d_i = n$, $\theta_i$ will always lie outside of the column space of the design matrix, meaning they will all incur some bias. We quantify this bias in the following lemma, which provides both an exact expression and a bound on the bias based on Theorem~\ref{thm:bias_bound_general}.

\begin{lemma} \label{lem:crd_bias}
    Suppose that $\bw \sim \Comp(k)$ and that Assumption~\ref{ass:low_deg} holds with $\beta = 1$. Then:
    \begin{enumerate}[(a)]
        \item if $|\cC(\cN_i)| < m$ for all $i \in [n]$, then $\widehat{\TTE}_1$ is unbiased.
        
	\item the exact bias of $\widehat{\TTE}_1$ is
        \begin{equation*}
            \E \big[ \widehat{\TTE}_1 \big] - \TTE = \frac{m}{(k^2+m)n} \sum_{i=1}^{n} \Ind \big( |\cC(\cN_i)| = m \big) \cdot \Big( k \cdot x_{i,\varnothing} - \sum_{C \in \cC} x_{i,\{C\}} \Big);
        \end{equation*}

        \item if Assumption~\ref{ass:bounded} holds, so $|Y_i(\bz)| \leq B$ for all $i \in [n]$ and $\bz \in \{0,1\}^n$, the bias of $\widehat{\TTE}$ satisfies the bound 
        \begin{equation*}
            \Big|\E \big[ \widehat{\TTE}_1 \big] - \TTE \Big| \leq \frac{\# \{ i : |\cC(\cN_i)| = m \}}{n} \cdot \frac{m(k+2)B}{k^2+m}.
        \end{equation*}
\end{enumerate}
\end{lemma}

\begin{proof}
    \textit{Proof: } To apply Theorem~\ref{thm:bias_bound_general}, we must calculate the vector $\Big( \E [\tbw_i \tbw_i^\intercal]^\dagger \E [\tbw_i \tbw_i^\intercal] - I\Big) \psi_i$ for each $i \in [n]$. We do this in two cases. First, when $|\cC(\cN_i)| < m$, Lemma~\ref{lem:crd_det} ensures that $\E [\tbw_i \tbw_i^\intercal]$ is invertible, meaning $\E [\tbw_i \tbw_i^\intercal]^\dagger = \E [\tbw_i \tbw_i^\intercal]^{-1}$; these nodes do not contribute to the bias (verifying part (a) of the lemma). On the other hand, if $|\cC(\cN_i)| = m$, then we can use the explicit form of $\E [\tbw_i \tbw_i^\intercal]^\dagger \E [\tbw_i \tbw_i^\intercal]$ given in the proof of Lemma~\ref{lem:pi_crd_form} to calculate:
    \[
        \Big( \E [\tbw_i \tbw_i^\intercal]^\dagger \E [\tbw_i \tbw_i^\intercal] - I\Big) \psi_i = \frac{m}{k^2+m} \begin{bmatrix} k \\ -1 \\ \vdots \\ -1 \end{bmatrix}
    \]
    Now, from Theorem~\ref{thm:bias_bound_general}, we obtain the exact bias
    \[
        \E \big[ \widehat{\TTE}_1 \big] - \TTE = \tfrac{m}{(k^2+m)n} \sum_{i=1}^{n} \Ind \big( |\cC(\cN_i)| = m \big) \cdot \Big( k \cdot x_{i,\varnothing} - \sum_{C \in \cC} x_{i,\{C\}} \Big).
    \]

    Under Assumption 3, we may bound
    \[
        \Big|k \cdot x_{i,\varnothing} - \sum_{C \in \cC} x_{i,\{C\}}\Big| = \Big| (k+1) Y(\mathbf{0}) - Y(1) \Big| \leq B(k+2),
    \]
    so
    \[
        \Big|\E \big[ \widehat{\TTE}_1 \big] - \TTE \Big| \leq \frac{\# \{ i : |\cC(\cN_i)| = m \}}{n} \cdot \frac{m(k+2)B}{k^2+m}.
    \]
\end{proof}

As mentioned above, for a reasonably large network with reasonably many clusters, we expect that no units will have $|\cC(\cN_i)| = m$, meaning $\widehat{\TTE}_1$ will be unbiased.

\subsection{Variance}

To apply Theorem~\ref{thm:tte_var_bound} to $\textrm{GCR}_{\textrm{Comp}}(\cC,k)$ designs, we first notice that the proof and results still hold given our representation of the estimator in terms of $\bw$, such that
\begin{equation} \label{eq:tte_var_bound_W}
	\var(\widehat{\TTE}_1) 
        \leq \tfrac{B^2}{n^2} \sum_{i,j=1}^n \gamma_i\gamma_j \cdot \Ind\big(\cov\big( \big\langle \widehat{\bx}_i, \psi_i \big\rangle, \big\langle \widehat{\bx}_j, \psi_j \big\rangle \big) > 0\big)
    \end{equation}
    for $\displaystyle \gamma_i=\sqrt{(\psi_i)^\intercal\E\big[\tbw_i (\tbw_i)^\intercal\big]^\dagger\psi_i}$ and $\widehat{\bx}_i = Y_i(\bw) \E \big[ \tbw_i \tbw_i^\intercal \big]^{\dagger} \tbw_i$.
To make the bound concrete, we must identify when $\psi_i\cov(\hat{\bx}_i, \hat{\bx}_j)\psi_j \leq 0$ and compute $\gamma_i$. For the former, consider two units $i$ and $j$. If the neighborhoods of units $i$ and $j$ are disjoint, then the entries of $\tbw_i$ will be distinct from the entries of $\tbw_j$. This means that all of the entries of $\tbw_i$ will be negatively correlated with the entries of $\tbw_j$, since
\[
    \cov(w_i, w_j) = \frac{k(k-1)}{m(m-1)} - \frac{k^2}{m^2} = -\frac{k(m-k)}{m^2(m-1)} < 0
\]
under complete randomization. It is thus natural to expect that $\psi_i\cov(\hat{\bx}_i, \hat{\bx}_j)\psi_j$ will be negative in this case, which we verify in the following lemma.

\begin{lemma} \label{lem:crd_var_neg_corr}
    Suppose that $\bw \sim \Comp(k)$, Assumption \ref{ass:low_deg} holds with $\beta=1$, and Assumption~\ref{ass:strong_mono} holds, so that the sign of $x_{i, C}$ is constant. Then, for two units $i, j \in [n]$ with  $\cC(\cN_i) \cap \cC(\cN_j) = \varnothing$, we have
    \[
         \cov \big( \psi_i^\intercal \: \widehat{\bx}_i, \psi_{j}^\intercal \: \widehat{\bx}_{j} \big) < 0.
    \]
\end{lemma}

As the proof involves a lot of tedious algebra and casework, we defer it to Section \ref{sec:proof_crd_cov_neg}. 

The next step is to compute $\gamma_i=\sqrt{|\cC_i^1|\cdot \psi_i^\intercal\E[\tbw_i\tbw_i^\intercal]^{\dagger}\psi_i}$. Since we assume that $\beta=1$, we have $|\cC_i^1|=|\cC(\cN_i)|+1$, leaving only the task of computing the quadratic form. As the form of the design matrix pseudoinverse differs when $|\cC(\cN_i)|<m$ versus when $|\cC(\cN_i)|=m$, our calculation also distinguishes these cases.

\begin{lemma} \label{lem:unit_crd_gamma}
	Suppose that $\bw\sim \Comp(k)$ and that Assumption~\ref{ass:low_deg} holds with $\beta = 1$. Then
	\begin{equation*}
		\gamma_i^2 = \begin{cases}
        {\frac{(|\cC(\cN_i)|+1)|\cC(\cN_i)|(m-1)}{(k/m)(1-k/m)(m-|\cC(\cN_i)|)}} & |\cC(\cN_i)| < m, \\[8pt]
        {\frac{(m+1)m^2k^2}{(k^2+m)^2}} & |\cC(\cN_i)| = m.
    \end{cases}
    \end{equation*}
\end{lemma}

Note that this lemma still holds if we replace $|\cC(\cN_i)|$ with an upper bound.

\begin{proof}\textit{Proof: } 
    Using the intermediary calculations from the proof of Lemma~\ref{lem:pi_crd_form}, we find that for any $i$ with $|\cC(\cN_i)| < m$,
    \begin{align*}
        \psi_i^\intercal\E[\tbw_i\tbw_i^\intercal]^\dagger\psi_i &= \frac{m(m-1)}{k(m-k)(m-|\cC(\cN_i)|)} \begin{bmatrix} 
            0 & 1 & \hdots & 1
        \end{bmatrix}
        \begin{bmatrix}
            -|\cC(\cN_i)| \cdot k \\
            m \\
            \vdots \\
            m
        \end{bmatrix}\\
        &=
        \frac{m^2|\cC(\cN_i)|(m-1)}{k(m-k)(m-|\cC(\cN_i)|)}\\
        &=\frac{|\cC(\cN_i)|(m-1)}{k/m(1-k/m)(m-|\cC(\cN_i)|)}.
    \end{align*}
    Similarly, for any $i$ with $|\cC(\cN_i)| = m$, 
    \[
        \psi_i^\intercal\E[\tbw_i\tbw_i^\intercal]^\dagger\psi_i = \frac{mk}{(k^2+m)^2} \begin{bmatrix} 
            0 & 1 & \hdots & 1
        \end{bmatrix}
        \begin{bmatrix}
            m \\
            k \\
            \vdots \\
            k \\
        \end{bmatrix} 
        =
        \frac{m^2k^2}{(k^2+m)^2}.
    \]
    Plugging these expressions into the formula $\gamma_i = \sqrt{|\cC_i^1| \cdot \psi_i^\intercal\E[\tbw_i\tbw_i^\intercal]^\dagger\psi_i}$, and noting that $|\cC_i^1| = (|\cC(\cN_i)|+1)$ when $\beta=1$ gives the expression from the lemma.
\end{proof}

From Lemma~\ref{lem:unit_crd_gamma}, we see that the value of $\gamma_i^2$ increases monotonically with $|\cC(\cN_i)|$ for $|\cC(\cN_i)|<m$, but then decreases at $|\cC(\cN_i)|=m$. Indeed, for $|\cC(\cN_i)|=m-1$, we have $\gamma_i^2=m(m-1)^2/((k/m)(1-k/m))$, while for $|\cC(\cN_i)|=m$, we have $\gamma_i=(m+1)m^2/(k+m/k)^2$. The numerators of both expressions are of order $m^3$, but the denominator of the latter expression is much larger (in particular, it is greater than $1$, unlike $k/m(1-k/m)$). Thus, the value of $\gamma_i$ will typically be much smaller when $|\cC(\cN_i)|=m$. Since these units with $|\cC(\cN_i)|=m$ control the scale of the bias, the relative bias and variance of our estimator are mediated by the number of units with $|\cC(\cN_i)|=m$: when there are no units with $|\cC(\cN_i)|=m$, the mean-squared error is simply the variance, while when there are many units with $|\cC(\cN_i)|=m$, the mean-squared error is dominated by the bias. We note that this differs from typical bias-variance trade-offs in that the analyst does not control the number of such nodes, which is a property of the graph.

This bias-variance trade-off under $\GCR_{\Comp}(\cC, k_c)$ designs is a somewhat unusual phenomenon that differs from those previously observed in the literature. Other estimators are either unbiased for all clusterings $\cC$ (e.g.,~the Horvitz--Thompson estimator of \eqref{eq:tte_ht}) or biased for all clusterings $\cC$ (e.g.,~the difference in means estimator). For estimators that are unbiased for all clusterings, we typically expect their variances strictly increase as the values of $|\cC(\cN_i)|$ increase. Here, we find that the variance of $\widehat{\TTE}_1$ increases as the value of $|\cC(\cN_i)|$ for a node increases, until we reach the threshold $|\cC(\cN_i)|=m$, at which point the variance decreases and we incur bias instead. This bias-variance trade-off mediated by the quantity $\#\{i: |\cC(\cN_i)|=m\}$ is, to the best of our knowledge, unique to this setting, and our results uncovering these properties are novel to our work. The empirical bias of $\widehat{\TTE}_1$ in our simulations in Section~\ref{subsec:comp_gcr_bias_sim} illustrates the dependence on $\#\{i: |\cC(\cN_i)|=m\}$ predicted by our upper bounds. Thus, our bounds correctly capture the dependence of the bias on the clustering.

Finally, combining our two lemmas with Theorem~\ref{thm:tte_var_bound}, we obtain the following result on the variance of $\widehat{\TTE}_1$ under completely randomized designs when $\beta=1$.

\begin{corollary}
    \label{cor:crd_var_bound}
    Suppose that $\bz\sim \GCR_{\Comp}(\cC, k)$ where each cluster has size at most $N$, $\max_i |\cC(\cN_i)| = C < m$, Assumption~\ref{ass:low_deg} holds with $\beta=1$, Assumption~\ref{ass:bounded} holds for some $B$, Assumption~\ref{ass:strong_mono} holds, and $k/m<1/2$. Then 
	\begin{equation*}
		\var(\widehat{\TTE}_1)\leq O\left(\frac{B^2}{n}\cdot\frac{C^3 N \dmax}{k/m}\cdot \frac{m}{m-C}\right).
	\end{equation*}
\end{corollary}

\begin{proof}{\it Proof:}
By Lemma~\ref{lem:crd_var_neg_corr}, the only nonzero terms in the variance are those pairs $i,j$ for which $\psi_i^T\cov(\hat{\bx}_i, \hat{\bx}_j)\psi_j>0$.
As each unit $i$ is connected to at most $C$ clusters, and since each of those clusters has at most $N$ units, each of which has at most $\dmax$ neighbors, there are at most $n C N \dmax$ non-zero terms in the variance bound given by Theorem~\ref{thm:tte_var_bound}. 
Lemma \ref{lem:unit_crd_gamma} implies that each of these non-zero $\gamma_i\gamma_j$ terms will be at most $mC^2/( (k/m)(1-k/m)(m-C))$, for a total variance of $B^2 m C^3 N \dmax/(n (k/m)(1-k/m)(m-C)$. The final bound comes from the assumption that $k/m < 1/2$ implies $1/(1-k/m) < 2$. 
\end{proof}

We can compare this bound to the variance bound from Corollary~\ref{thm:pi_gcr} for GCR$(\cC,p)$ designs, specialized to the $\beta = 1$ setting with treatment probability $p = \frac{k}{m}$:
\[
    \var(\widehat{\TTE}_1) \leq O\left( \frac{B^2}{n} \cdot \frac{C^2N\dmax}{k/m} \right). 
\]
When comparing this bound to the bound in Corollary~\ref{cor:crd_var_bound}, note that the variances (left hand side) are with regards to different designs. The two bounds agree up to a factor of $Cm/(m-C)$. Note that the assumption that $\max_i |\cC(\cN_i)| = C < m$ ensures that all nodes fall into the first case of Lemma~\ref{lem:unit_crd_gamma}, but as discussed earlier, this is a worst case scenario as Lemma \ref{lem:unit_crd_gamma} suggests the variance bound would only be smaller if for some units $|\cC(\cN_i)| = m$. Corollary~\ref{cor:crd_var_bound} establishes that the variance of $\widehat{\TTE}_1$ is polynomial in $\dmax$, rather than exponential, for completely randomized unit designs when $\beta=1$, thus generalizing the $\beta=1$ case of the variance bound of \citet{cortez2023exploiting} to the completely randomized setting. More generally, this section and its results show how our framework can be used to obtain bounds on the variance of the pseudoinverse estimator under arbitrary designs. 

\subsection{Experiments Illustrating Bias in the $\GCR_{\Comp}$ Design}
\label{subsec:comp_gcr_bias_sim}

In this section, we experimentally verify the theoretical results on the bias-variance trade-off exhibited by $\widehat{\TTE}_1$ under a $\GCR_{\Comp}$ design. We do this by focusing on a single one of the settings considered in Section~\ref{subsec:choose_cluster}, namely the METIS clusterings of the $\text{SBM}(0.5, 0.2)$ graph, with responses generated so that $\TTE_i=1$. For each of these clusterings, we compute the number of units with $\cC(\cN_i)=m$, which is the quantity that Lemma~\ref{lem:crd_bias} indicated would control the bias, as well as the bias, variance, and RMSE in simulation. The results are shown in Figure~\ref{fig:cluster_gcr_bias}.

\begin{figure}[t]
    \includegraphics[width=0.8\textwidth]{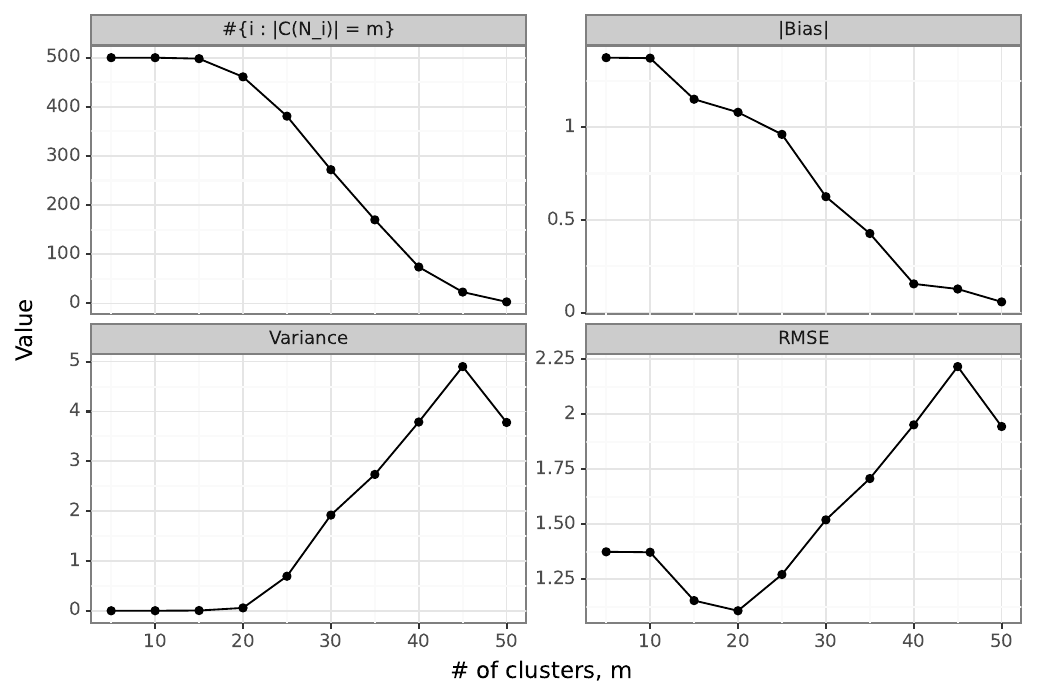}
    {The value of $\{i: |\cC(\cN_i)|=m\}$ (i.e., the number of nodes that have a neighbor in each of the $m$ clusters), as well as bias, variance, and RMSE in simulation for each METIS clustering of the $\text{SBM}(0.5, 0.2)$ graph. Each cluster has a size of approximately $n/m$. As predicted by our theory, the bias decreases as $\{i: |\cC(\cN_i)|=m\}$ decreases, while the variance increases, leading to a bias-variance trade-off and local minimum in the RMSE.}{}
    \label{fig:cluster_gcr_bias}
\end{figure}

We see in Figure~\ref{fig:cluster_gcr_bias} that, as predicted by the bound of Lemma~\ref{lem:crd_bias}, the bias of $\widehat{\TTE}_1$ decreases as $|\{i: |\cC(\cN_i)|=m\}$ decreases, while the variance of $\widehat{\TTE}_1$ increases. This creates a bias-variance trade-off in the RMSE, leading to a local minimum when using the METIS clustering into $m=20$ clusters. Furthermore, we see in row six of Table~\ref{tab:cluster_selection} that our bounds select a clustering that has nearly optimal RMSE for this graph and response model, meaning that our bounds can effectively guide a practitioner in navigating this bias-variance trade-off.
\subsection{Proof of Lemma~\ref{lem:crd_var_neg_corr}}
\label{sec:proof_crd_cov_neg}

\begin{proof}
    Substituting the definition of $\widehat{\bx}_i$ and $\widehat{\bx}_i'$, we have,
    \begin{align*}
        \cov \big( \psi_i^\intercal \: \widehat{\bx}_i, \psi_{i'}^\intercal \: \widehat{\bx}_{i'} \big)
        &= \cov \bigg( Y_i(\bw) \Big\langle \E \big[ \tbw_i \tbw_i^\intercal \big] \psi_i, \tbw_i \Big\rangle \;,\; Y_{i'}(\bw) \Big\langle \E \big[ \tbw_{i'} \tbw_{i'}^\intercal \big] \psi_{i'}, \tbw_{i'} \Big\rangle \bigg) \\
        &= \E \bigg[ Y_i(\bw) \Big\langle \E \big[ \tbw_i \tbw_i^\intercal \big] \psi_i, \tbw_i \Big\rangle \cdot Y_{i'}(\bw) \Big\langle \E \big[ \tbw_{i'} \tbw_{i'}^\intercal \big] \psi_{i'}, \tbw_{i'} \Big\rangle \bigg] \\
        &\hspace{40pt} - \E \bigg[ Y_i(\bw) \Big\langle \E \big[ \tbw_i \tbw_i^\intercal \big] \psi_i, \tbw_i \Big\rangle \bigg] \cdot \E \bigg[ Y_{i'}(\bw) \Big\langle \E \big[ \tbw_{i'} \tbw_{i'}^\intercal \big] \psi_{i'}, \tbw_{i'} \Big\rangle \bigg]
    \end{align*}
    Since $i \in \cN_i$ and $\cC(\cN_i) \cap \cC(\cN_{i'}) = \varnothing$, we know that $\cC(i) \not\in \cC(\cN_{i'})$, so $|\cC(\cN_{i'})| < m$. Similarly, $|\cC(\cN_i)| < m$. Hence, we can use our calculation from the proof of Lemma~\ref{lem:pi_crd_form} to rewrite this
    \begin{align}
        \tfrac{m^4(m-1)^2}{k^2(m-k)^2(m-|\cC(\cN_i)|)(m-|\cC(\cN_{i'})|)} \bigg( \E & \bigg[ Y_i(\bw) \cdot Y_{i'}(\bw) \sum_{C \in \cC(\cN_i)} \big( w_C - \tfrac{k}{m} \big) \sum_{C' \in \cC(\cN_{i'})} \big( w_{C'} - \tfrac{k}{m} \big) \bigg] \\
        &- \E \bigg[ Y_i(\bw) \sum_{C \in \cC(\cN_i)} \big( w_C - \tfrac{k}{m} \big) \bigg] \E \bigg[ Y_{i'}(\bw) \sum_{C' \in \cC(\cN_{i'})} \big( w_{C'} - \tfrac{k}{m} \big) \bigg] \bigg) \notag \\
        = \tfrac{m^4(m-1)^2}{k^2(m-k)^2(m-|\cC(\cN_i)|)(m-|\cC(\cN_{i'})|)} \E & \bigg[ Y_i(\bw) \cdot Y_{i'}(\bw) \hspace{-8pt} \sum_{C \in \cC(\cN_i)} \hspace{-8pt} \big( w_C - \tfrac{k}{m} \big) \hspace{-8pt} \sum_{C' \in \cC(\cN_{i'})} \hspace{-8pt} \big( w_{C'} - \tfrac{k}{m} \big) \bigg] - \hspace{-8pt} \sum_{C \in \cC(\cN_i)} \hspace{-4pt} x_{i,C} \hspace{-4pt} \sum_{C' \in \cC(\cN_{i'})} \hspace{-4pt} x_{i',C'} \label{eq:covar_with_exp_prod}.
    \end{align}  
    Here, the second line uses an expectation calculation from the proof of Lemma~\ref{lem:crd_bias}. Next, let us focus our attention on the expectation in the latter expression. By substituting in the definitions of $Y_i(\bw)$ and $Y_{i'}(\bw)$ (under our assumed $\beta=1$ order potential outcomes model), we can rewrite this expectation
    \begin{align*}
        \E &\bigg[ \Big( x_{i,\varnothing} + \sum_{h \in \cC(\cN_i)} x_{i,h} \Big) \Big( x_{i',\varnothing} + \sum_{h' \in \cC(\cN_{i'})} x_{i',h'} \Big) \sum_{C \in \cC(\cN_i)} \big( w_C - \tfrac{k}{m} \big) \sum_{C' \in \cC(\cN_{i'})} \big( w_{C'} - \tfrac{k}{m} \big) \bigg] \\
        = \:&x_{i,\varnothing} \cdot x_{i', \varnothing} \sum_{C \in \cC(\cN_i)} \sum_{C' \in \cC(\cN_{i'})} \E \Big[ \big( w_C - \tfrac{k}{m} \big) \big( w_{C'} - \tfrac{k}{m} \big) \Big] \tag{$\star$} \\
        &+ \: x_{i,\varnothing} \sum_{h' \in \cC(\cN_{i'})} x_{i',h'} \sum_{C \in \cC(\cN_i)} \sum_{C' \in \cC(\cN_{i'})} \E \Big[ w_{h'} \big( w_C - \tfrac{k}{m} \big) \big( w_{C'} - \tfrac{k}{m} \big) \Big] \tag{$\star\star$}\\
        &+ \: x_{i',\varnothing} \sum_{h \in \cC(\cN_i)} x_{i,h} \sum_{C \in \cC(\cN_i)} \sum_{C' \in \cC(\cN_{i'})} \E \Big[ w_h \big( w_C - \tfrac{k}{m} \big) \big( w_{C'} - \tfrac{k}{m} \big) \Big] \\
        &+ \: \sum_{h \in \cC(\cN_i)} x_{i,h} \sum_{h' \in \cC(\cN_{i'})} x_{i',h'} \sum_{C \in \cC(\cN_i)} \sum_{C' \in \cC(\cN_{i'})} \E \Big[ w_h w_{h'} \big( w_C - \tfrac{k}{m} \big) \big( w_{C'} - \tfrac{k}{m} \big) \Big].
    \end{align*}
    Now, let us separately consider the expectations in ($\star$) and ($\star\star$); the following line has the same form, so it is handled analogously. In all of these expressions, note that our assumption that $\cC(\cN_i) \cap \cC(\cN_{i'}) = \varnothing$ ensures that $w_C \ne w_{C'}$ (and also $w_h \ne w_{C'}$ and $w_{h'} \ne w_C$). Thus, in ($\star$),
    \[
        \E \Big[ \big( w_C - \tfrac{k}{m} \big) \big( w_{C'} - \tfrac{k}{m} \big) \Big] = \tfrac{k(k-1)}{m(m-1)} - \tfrac{k^2}{m^2} = \tfrac{-k(m-k)}{m^2(m-1)}.
    \]
    This is always negative, and Assumption~\ref{ass:strong_mono} ensures that $x_{i,\varnothing} \cdot x_{i',\varnothing} \geq 0$, so the entire expression ($\star)$ is $\leq 0$.

    In $(\star\star)$, there are two possibilities for $w_{C'}$ which we must handle separately. First, if $w_{C'} = w_{h'}$, then
    \[
        \E \Big[ w_{h'} \big( w_C - \tfrac{k}{m} \big) \big( w_{C'} - \tfrac{k}{m} \big) \Big] = \tfrac{k(k-1)}{m(m-1)} - \tfrac{k^2}{m^2} - \tfrac{k^2(k-1)}{m^2(m-1)} + \tfrac{k^3}{m^3} = \tfrac{-k(m-k)^2}{m^3(m-1)}.
    \]
    Next, if $w_{C'} \ne w_{h'}$, then  
    \[
        \E \Big[ w_{h'} \big( w_C - \tfrac{k}{m} \big) \big( w_{C'} - \tfrac{k}{m} \big) \Big] = \tfrac{k(k-1)(k-2)}{m(m-1)(m-2)} - \tfrac{2k^2(k-1)}{m^2(m-1)} + \tfrac{k^3}{m^3} = \tfrac{-k(m-k)(mk-2m+2k)}{m^3(m-1)(m-2)}.
    \]
    Together, this allows us to simplify ($\star\star$) to
    \[
        x_{i,\varnothing} \sum_{h' \in \cC(\cN_{i'})} x_{i',h'} \cdot |\cC(\cN_i)| \cdot \Big( \tfrac{-k(m-k)}{m^2(m-1)} + (|\cC(\cN_{i'})|-1) \cdot \tfrac{-k(m-k)(mk-2m+2k)}{m^3(m-1)(m-2)} \Big).
    \]
    Within the large parentheses, note that the first fraction is always negative, and the second fraction is negative when $k \geq 2$, as this makes $mk - 2m + 2k > 0$. If $k=1$, then this parenthesized expression simplifies to $\tfrac{|\cC(\cN_{i'})|-1-m}{m^3} < 0$. Again using Assumption~\ref{ass:strong_mono}, we find that the entire expression ($\star\star$) is $\leq 0$. 

    These observations allow us to bound \eqref{eq:covar_with_exp_prod},
    \begin{align*}
        \eqref{eq:covar_with_exp_prod} &\leq \sum_{h \in \cC(\cN_i)} x_{i,h} \sum_{h' \in \cC(\cN_{i'})} x_{i',h'} \bigg( \tfrac{m^4(m-1)^2}{k^2(m-k)^2(m-|\cC(\cN_i)|)(m-|\cC(\cN_{i'})|)} \sum_{C \in \cC(\cN_i)} \sum_{C' \in \cC(\cN_{i'})} \\
        &\phantom{a} \hspace{30pt} \E \Big[ w_h w_{h'} \big( w_C - \tfrac{k}{m} \big) \big( w_{C'} - \tfrac{k}{m} \big) \Big] - 1 \bigg).
    \end{align*}

    Let us simplify this expectation. Here, there are three possibilities for $w_C$ and $w_{C'}$ which we must handle separately. First, if $w_C = w_h$ and $w_{C'} = w_{h'}$, then
    \[
        \E \Big[ w_h w_{h'} \big( w_C - \tfrac{k}{m} \big) \big( w_{C'} - \tfrac{k}{m} \big) \Big] = \tfrac{k(k-1)}{m(m-1)} - \tfrac{2k^2(k-1)}{m^2(m-1)} + \tfrac{k^3(k-1)}{m^3(m-1)} = \tfrac{k(k-1)(m-k)^2}{m^3(m-1)}.
    \]
    Next, if either $w_C = w_h$ or $w_{C'} = w_{h'}$, but not both, then
    \[
        \E \Big[ w_h w_{h'} \big( w_C - \tfrac{k}{m} \big) \big( w_{C'} - \tfrac{k}{m} \big) \Big] = \tfrac{k(k-1)(k-2)}{m(m-1)(m-2)} - \tfrac{k^2(k-1)(k-2)}{m^2(m-1)(m-2)} - \tfrac{k^2(k-1)}{m^2(m-1)} + \tfrac{k^3(k-1)}{m^3(m-1)} = \tfrac{-2k(k-1)(m-k)^2}{m^3(m-1)(m-2)}.
    \]
    Finally, if $w_C \ne w_h$ and $w_{C'} \ne w_{h'}$, then
    \[
        \E \Big[ w_h w_{h'} \big( w_C - \tfrac{k}{m} \big) \big( w_{C'} - \tfrac{k}{m} \big) \Big] = \tfrac{k(k-1)(k-2)(k-3)}{m(m-1)(m-2)(m-3)} - \tfrac{2k^2(k-1)(k-2)}{m^2(m-1)(m-2)} + \tfrac{k^3(k-1)}{m^3(m-1)} = \tfrac{-k(k-1)(m-k)(mk-6m+6k)}{m^3(m-1)(m-2)(m-3)}.
    \]
    Together, this allows us to rewrite our bound on \eqref{eq:covar_with_exp_prod} as 
    \begin{align*}
        \eqref{eq:covar_with_exp_prod} &\leq \sum_{h \in \cC(\cN_i)} x_{i,h} \sum_{h' \in \cC(\cN_{i'})} x_{i',h'} \bigg( -1 + \tfrac{m^4(m-1)^2}{k^2(m-k)^2(m-|\cC(\cN_i)|)(m-|\cC(\cN_{i'})|)} \cdot \bigg[ \tfrac{k(k-1)(m-k)^2}{m^3(m-1)} -  \\
        &\hspace{20pt} (|\cC(\cN_i)| + |\cC(\cN_{i'})| - 2) \tfrac{2k(k-1)(m-k)^2}{m^3(m-1)(m-2)} - (|\cC(\cN_i)| - 1) (|\cC(\cN_{i'})| - 1) \tfrac{k(k-1)(m-k)(mk-6m+6k)}{m^3(m-1)(m-2)(m-3)} \bigg] \bigg) \\
        &= \sum_{h \in \cC(\cN_i)} x_{i,h} \sum_{h' \in \cC(\cN_{i'})} x_{i',h'} \bigg( -1 + \tfrac{m(m-1)(k-1)(m-2|\cC(\cN_i)|-2|\cC(\cN_{i'})|+2)}{k(m-2)(m-|\cC(\cN_i)|)(m-|\cC(\cN_{i'})|)} - \\  
        &\hspace{20pt} \tfrac{m(m-1)(k-1)(mk-6m+6k)(|\cC(\cN_i)|-1)(|\cC(\cN_{i'})|-1)}{k(m-k)(m-2)(m-3)(m-|\cC(\cN_i)|)(m-|\cC(\cN_{i'})|)} \bigg). \\
    \end{align*}
    By our same-signs assumption, our proof will be complete once we argue that the parenthesized expression in the last expression is negative, which we do via a case analysis.

    \textbf{Case 1:} $k=1$

    Both fractions in the parenthesized expression zero out, and we are left with $-1 < 0$.

    \textbf{Case 2:} $k=2$

    The parenthesized expression simplifies to $-1 + \tfrac{m(m-1)(m-2|\cC(\cN_i)|)(m-2|\cC(\cN_{i'})|)}{2(m-2)^2(m-|\cC(\cN_i)|)(m-|\cC(\cN_{i'})|)}$. Since $\cC(\cN_i)$ and $\cC(\cN_{i'})$ are disjoint, $|\cC(\cN_i)| + |\cC(\cN_{i'})| \leq m$. Furthermore if $(m-2|\cC(\cN_{i'})|) < 0$, then we must have that $m-2|\cC(\cN_i)| > 0$. In this case, the parenthesized fraction is negative. Thus, we are left to consider the possibility where $m-2|\cC(\cN_{i'})| > 0$. Since $|\cC(\cN_i)| > 1$, we have,
    \[
        \tfrac{(m-1)(m-2|\cC(\cN_i)|)}{(m-2)(m-|\cC(\cN_i)|)} \leq \tfrac{(m-1)(m-|\cC(\cN_i)|-1)}{(m-2)(m-|\cC(\cN_i)|)} = \tfrac{m^2-n|\cC(\cN_i)|-2n + |\cC(\cN_i)|+1}{m^2-n|\cC(\cN_i)|-2m+2|\cC(\cN_i)|} \leq 1.
    \]
    Thus, we can upper-bound the parenthesized expression by $-1 + \tfrac{m(m-2|\cC(\cN_{i'})|)}{2(m-2)(m-|\cC(\cN_{i'})|)}$. Noting that $|\cC(\cN_{i'})| \geq 1$, we can bound this fraction,
    \[
        \tfrac{m(m-2|\cC(\cN_{i'})|)}{m(m-2|\cC(\cN_{i'})|) + (m^2 - 4n + 4|\cC(\cN_i)|')} \leq \tfrac{m(m-2|\cC(\cN_{i'})|)}{m(m-2|\cC(\cN_{i'})|) + (m^2 - 4n + 4)} = \tfrac{m(m-2|\cC(\cN_{i'})|)}{m(m-2|\cC(\cN_{i'})|) + (m-2)^2} \leq \tfrac{m(m-2|\cC(\cN_{i'})|)}{m(m-2|\cC(\cN_{i'})|) + (m-2)^2} < 1. 
    \]
    Again, we find that the parenthesized expression must be negative. 

    \textbf{Case 3:} $|\cC(\cN_i)| = 1$ (analogously, $|\cC(\cN_{i'})| = 1$), \quad $k \geq 2$

    Note that the factor of $|\cC(\cN_i)| - 1$ in the second fraction makes it assume value 0 in this case. Thus, we can simplify the expression to $-1 + \tfrac{m(k-1)(m-2|\cC(\cN_{i'})|)}{k(m-2)(m-|\cC(\cN_{i'})|)}$. Taking the derivative with respect to $|\cC(\cN_{i'})|$, we obtain $\tfrac{-m^2(k-1)}{k(m-2)(m-|\cC(\cN_{i'})|)^2} < 0$. Thus, the parenthesized expression is maximized when $|\cC(\cN_{i'})| = 1$, with value $-1 + \tfrac{m(k-1)}{k(m-1)} = \tfrac{-(m-k)}{k(m-1)} < 0$.

    \textbf{Case 4:} $k \geq 6$

    Here, the factor $mk -6m + 6k \geq 6k$ is positive, as is every other factor in the second fraction. Thus, we can upper-bound the parenthesized expression by $-1 + \tfrac{m(m-1)(k-1)(m-2|\cC(\cN_i)|-2|\cC(\cN_{i'})|+2)}{k(m-2)(m-|\cC(\cN_i)|)(m-|\cC(\cN_{i'})|)}$.
    Note that the factor $\tfrac{(m-1)(k-1)}{k(m-2)} = \tfrac{mk-2k-(m-1-k)}{mk-2k} \leq 1$ since $k \leq m-1$. Thus, we may further upper bound the expression by $-1 + \tfrac{m(m-2|\cC(\cN_i)|-2|\cC(\cN_{i'})|+2)}{(m-|\cC(\cN_i)|)(m-|\cC(\cN_{i'})|)}$. Since $|\cC(\cN_i)|, |\cC(\cN_{i'})| \geq 1$, we can again upper bound this by
    \[
        -1 + \tfrac{m(m-|\cC(\cN_i)|-|\cC(\cN_{i'})|)}{(m-|\cC(\cN_i)|)(m-|\cC(\cN_{i'})|)} = \tfrac{-|\cC(\cN_i)||\cC(\cN_{i'})|}{(m-|\cC(\cN_i)|)(m-|\cC(\cN_{i'})|)} < 0.
    \]

    \textbf{Case 5:} $k \in \{3,4,5\}, \quad |\cC(\cN_i)|, |\cC(\cN_{i'})| \geq 2$

    Taking the derivative of the parenthesized expression with respect to $|\cC(\cN_i)|$, we obtain
    \begin{gather*}
        \tfrac{m(m-1)(k-1)}{k(m-2)(m-|\cC(\cN_{i'})|)} \cdot \tfrac{\partial}{\partial |\cC(\cN_i)|} \Big[ \tfrac{(m-2|\cC(\cN_i)|-2|\cC(\cN_{i'})|+2)}{(m-|\cC(\cN_i)|)}  - \tfrac{(mk-6m+6k)(|\cC(\cN_i)|-1)(|\cC(\cN_{i'})|-1)}{(m-k)(m-3)(m-|\cC(\cN_i)|)} \Big] \\
        = \tfrac{m(m-1)(k-1)}{k(m-2)(m-|\cC(\cN_{i'})|)} \cdot \Big[ \tfrac{-m-2|\cC(\cN_{i'})|+2}{(m-|\cC(\cN_i)|)^2} - \tfrac{(mk-6m+6k)(|\cC(\cN_{i'})|-1)}{(m-k)(m-3)} \cdot \tfrac{m-1}{(m-|\cC(\cN_i)|)^2} \Big] \\
        = \tfrac{m(m-1)(k-1)}{k(m-2)(m-3)(m-k)(m-|\cC(\cN_{i'})|)(m-|\cC(\cN_i)|)^2} \cdot \Big[ (m-k)(m-3)(-m-2|\cC(\cN_{i'})|+2) \\
        \phantom{a} \hspace{40pt} - (m-1)(mk-6m+6k)(|\cC(\cN_{i'})|-1) \Big]
    \end{gather*}
    In this last expression, the outer fraction is positive, so we can restrict to understanding the sign of the bracketed expression. We do this separately for each value of $k$.
    
    \underline{$k=3$}: The expression simplifies to $-m^3 + m^2 |\cC(\cN_{i'})| + 5m^2 - 9m|\cC(\cN_{i'})|$. If $m=4$, this further simplifies to $16-20|\cC(\cN_{i'})| < 0$. Otherwise, if $m \geq 5$, we can rewrite it $-m(m-5)(m-|\cC(\cN_i)|) - 4m|\cC(\cN_{i'})| < 0$.

    \underline{$k=4$}: The expression simplifies to $-m(m^2 - 7m + 12|\cC(\cN_{i'})|)$. Under our assumption that $|\cC(\cN_{i'})| \geq 2$, we can upper bound this by $-m(m^2 - 7m + 24) < 0$.

    \underline{$k=5$}: The expression simplifies to $-m(m^2 - 9m + |\cC(\cN_{i'})|m + 15|\cC(\cN_{i'})|)$. Under our assumption that $|\cC(\cN_{i'})| \geq 2$, we can upper bound this by $-m(m^2 - 7m + 30) < 0$.
\end{proof}

\end{document}